%% file: full-main.tex
\documentclass[11pt,runningheads]{llncs}
\input{preamble}
\usepackage{fullpage}
\usepackage[titles]{tocloft}
\begin{document}
\title{The Parallel Reversible Pebbling Game: Analyzing the Post-Quantum Security of iMHFs}
\date{\today}
\titlerunning{The Parallel Reversible Pebbling Game: Analyzing the Post-Quantum Security of iMHFs}

%\author{}
%\institute{}

\author{Jeremiah Blocki%\inst{1}%\orcidID{0000-0002-5542-4674} 
\and Blake Holman
\and Seunghoon Lee%\inst{1}%\orcidID{0000-0003-4475-5686}
}

\institute{Purdue University, West Lafayette, IN, 47906, USA\\
\email{\{jblocki,holman14,lee2856\}@purdue.edu}
}

\maketitle

\input{abstract}

\input{ec-intro}

\input{ec-quantum-pebbling}
\input{ec-attack}
\input{genpeb}
\input{openquestions}
\input{acknowledgement}

\def\shortbib{0}
\bibliographystyle{alpha}
\bibliography{references,abbrev3,crypto}
\newpage
\appendix
\input{reversible}
\input{examples}
\input{missingproof}
\newpage
\input{algorithms}

\end{document}

%% file: preamble.tex
\usepackage{etoolbox}%
\newcommand{\FullVersion}[2]{%
    \ifstrequal{full}{full}% change the second argument for the camera-ready version, i.e., full <-> camera
    {#1}% full version
    {#2}% camera-ready version
}
\usepackage[final]{changes}

\usepackage{makeidx}
\usepackage{euscript}
\usepackage[nointegrals]{wasysym}
\usepackage{soul}
\usepackage{amsmath,amssymb,amsfonts,stmaryrd}

\usepackage{dsfont}
\usepackage{latexsym}
\usepackage{subcaption}
\usepackage{graphicx}
\usepackage{fancybox}
\usepackage[backref=page,breaklinks=true]{hyperref}
\usepackage{breakcites}
\usepackage{color}
\usepackage{wrapfig}
\usepackage{svrsymbols}
\usepackage{tikz}
\usetikzlibrary{shadings}
\usetikzlibrary{decorations.pathreplacing,backgrounds}
\usetikzlibrary {positioning,chains,fit,shapes,calc}
\usetikzlibrary {positioning,chains,fit,shapes,calc,arrows,patterns,decorations.pathreplacing}
\tikzset{
    position/.style args={#1:#2 from #3}{
        at=(#3.#1), anchor=#1+180, shift=(#1:#2)
    }
}
\usepackage{setspace}
\usepackage{xcolor}
\usepackage[linesnumbered,ruled,vlined,noend]{algorithm2e}

\SetCommentSty{mycommfont}

\usepackage[noend]{algpseudocode}
\usepackage[framemethod=tikz]{mdframed}
\usepackage{xspace}
\usepackage{pgfplots}
\usepackage{subcaption}
\usepackage{framed}
\pgfplotsset{compat=1.5}
\usepackage{standalone}
\usepackage{mathtools}
\usepackage{marvosym}

\spnewtheorem{claim}{Claim}{\bfseries}{\itshape}

\newenvironment{proofof}[1]{\begin{trivlist} \item {\bf Proof
#1:~~}}
  {\qed\end{trivlist}}
\renewenvironment{proofof}[1]{\par\medskip\noindent{\bf Proof of #1: \ }}{\hfill\qed\par\medskip}
\newcommand{\namedref}[2]{\hyperref[#2]{#1~\ref*{#2}}}
\newcommand{\namedrefeq}[2]{\hyperref[#2]{#1~(\ref*{#2})}}
\newcommand{\thmlab}[1]{\label{thm:#1}}
\newcommand{\thmref}[1]{\namedref{Theorem}{thm:#1}}
\newcommand{\lemlab}[1]{\label{lem:#1}}
\newcommand{\lemref}[1]{\namedref{Lemma}{lem:#1}}
\newcommand{\claimlab}[1]{\label{claim:#1}}
\newcommand{\claimref}[1]{\namedref{Claim}{claim:#1}}
\newcommand{\corlab}[1]{\label{cor:#1}}
\newcommand{\corref}[1]{\namedref{Corollary}{cor:#1}}
\newcommand{\seclab}[1]{\label{sec:#1}}
\newcommand{\secref}[1]{\namedref{Section}{sec:#1}}
\newcommand{\applab}[1]{\label{app:#1}}
\newcommand{\appref}[1]{\namedref{Appendix}{app:#1}}

\newcommand{\remlab}[1]{\label{rem:#1}}

\newcommand{\figlab}[1]{\label{fig:#1}}
\newcommand{\figref}[1]{\namedref{Figure}{fig:#1}}
\newcommand{\alglab}[1]{\label{alg:#1}}
\renewcommand{\algref}[1]{\namedref{Algorithm}{alg:#1}}

\newcommand{\deflab}[1]{\label{def:#1}}
\newcommand{\defref}[1]{\namedref{Definition}{def:#1}}
\newcommand{\eqnlab}[1]{\label{eq:#1}}
\newcommand{\eqnref}[1]{\namedrefeq{Equation}{eq:#1}}

\newcommand{\qedmath}{\tag*{$\Box$}}

\newenvironment{remindertheorem}[1]{\medskip \noindent {\bf Reminder of  #1.  }\em}{}

\def \depth    {\mdef{\mathsf{depth}}}
\def \sinks    {\mdef{\mathsf{sinks}}}

\newcommand{\genpeb}{\mathsf{GenPeb}}

\newcommand{\startlight}{\mathsf{StartLight}}
\newcommand{\lightreq}{\mathsf{LightReq}}
\newcommand{\balloonreq}{\mathsf{BalloonReq}}
\newcommand{\startballoon}{\mathsf{StartBalloon}}
\newcommand{\Endballoon}{\mathsf{EndBalloon}}
    \newcommand{\midballoon}{\mathsf{MidBalloon}}
\newcommand{\rev}{\mathsf{rev}}

\usetikzlibrary{decorations.pathreplacing}
\usetikzlibrary {positioning,chains,fit,shapes,calc,snakes}

\tikzset{
    position/.style args={#1:#2 from #3}{
        at=(#3.#1), anchor=#1+180, shift=(#1:#2)
    }
}

\definecolor{purduedigitalheadlinegold}{HTML}{98700D}
\definecolor{purduecampusgold}{HTML}{C28E0E}
\definecolor{purduecoalgray}{HTML}{4D4038}
\definecolor{purdueevertrueblue}{HTML}{5B6870}
\definecolor{purduemoondustgray}{HTML}{BAA892}
\definecolor{purdueslayterskyblue}{HTML}{6E99B4}
\definecolor{mahogany}{rgb}{0.75, 0.25, 0.0}
\definecolor{darkblue}{rgb}{0.0, 0.0, 0.55}
\definecolor{darkpastelgreen}{rgb}{0.01, 0.75, 0.24}
\definecolor{darkgreen}{rgb}{0.0, 0.2, 0.13}
\definecolor{darkgoldenrod}{rgb}{0.72, 0.53, 0.04}
\definecolor{forestgreen}{rgb}{0.13, 0.55, 0.13}
\definecolor{darkred}{rgb}{0.55, 0.0, 0.0}
\definecolor{tangocolordarkchameleon}{HTML}{4E9A06}
\newcommand{\COMMENTED}[1]{{}}

\newcommand{\Ex}[1]{\ensuremath{\mathbb{E}[#1]}}

\def \indeg    {\mdef{\mathsf{indeg}}}

\def \path    {\mdef{\mathsf{PATH}}}

\def \lpath    {\mdef{\mathsf{LongestPath}}}

\def \rgenpeb    {\mdef{\mathsf{RGenPeb}}}

\def \parents    {\mdef{\mathsf{parents}}}
\def \ancestors    {\mdef{\mathsf{ancestors}}}

\def \ArgonA   {\mdef{\mathsf{Arg}\text{-}\mathsf{A}}}
\def \ArgonB   {\mdef{\mathsf{Arg}\text{-}\mathsf{B}}}
\def \DRS   {\mdef{\mathsf{DRS}}}
\def \skipnode   {\mdef{\mathsf{Skip}}}

\def \lastd   {\mdef{\mathsf{LastDelete}}}
\def \lasta   {\mdef{\mathsf{LastAdd}}}
\def \trans   {\mdef{\mathsf{Trans}}}
\def \blockpebble   {\mdef{\mathsf{BlockPebble}}}
\def \lastblockpebble   {\mdef{\mathsf{LastBlockPebble}}}
\def \numskip   {\mdef{\mathsf{NumSkip}}}
\def \prevpeb   {\mdef{\mathsf{RevPeb}}}

\newcommand\sample{\mathrel{\overset{\vbox to.7ex{\hbox{\fontsize{6}{0}\selectfont\vspace{.5ex} \$}}}{\leftarrow}}}

\renewcommand{\O}[1]{\ensuremath{\mathcal{O}\left(#1\right)}}

\makeatletter
\newcommand{\lrarrow}{\mathrel{\mathpalette\lrarrow@\relax}}
\newcommand{\lrarrow@}[2]{%
  \vcenter{\hbox{\ooalign{%
    $\m@th#1\mkern6mu\rightarrow$\cr
    \noalign{\vskip1pt}
    $\m@th#1\leftarrow\mkern6mu$\cr
  }}}%
}
\makeatother

\newcommand{\qPeb}[1]{\ensuremath{\mathcal{P}^{\lrarrow}_{#1}}}
\newcommand{\pqPeb}[1]{\ensuremath{\mathcal{P}^{\lrarrow,\|}_{#1}}}
\newcommand{\rqPeb}[1]{\ensuremath{\widetilde{\mathcal{P}}^{\lrarrow}_{#1}}}
\newcommand{\rpqPeb}[1]{\ensuremath{\widetilde{\mathcal{P}}^{\lrarrow,\|}_{#1}}}

\renewcommand{\figurename}{Fig.}

\newcommand{\peb}{\Pi} % pebbling complexity
\newcommand{\Peb}{{\cal P}} % set of legal pebblings

\newcommand{\ppeb}{\peb^{\parallel}} % parallel versions of def above
\newcommand{\qpeb}{\peb^{\lrarrow}}
\newcommand{\pqpeb}{\peb^{\lrarrow,\|}}
\newcommand{\pPeb}{\Peb^{\parallel}}

\newcommand{\Cspace}{s}
\newcommand{\Ctime}{t}
\newcommand{\Ccc}{{cc}}
\newcommand{\Cspacetime}{{st}}
\newcommand{\Cevery}{\{\Cspace,\Ctime,\Cspacetime,\Ccc\}}

\newcommand{\mdef}[1]{{\ensuremath{#1}}\xspace}  % Math Def which can also be used in normal text.
\newcommand{\myfunc}[1]{\mdef{\mathsf{#1}}}      % Functions denoted with text (e.g. size()) should use mathsf.
\newcommand{\superscript}[1]{\ensuremath{^{\mbox{\tiny{\textit{#1}}}}}\xspace}
\def \th {\superscript{th}}     % 'The i-th entry it a list...' --> i\th
\def \st {\superscript{st}}     % 'The 1-st entry it a list...' --> 1\st
\def \poly  {\mdef{\myfunc{poly}}}             % Poly function
\newcommand{\abs}[1]{\mdef{\left|#1\right|}}         % Absolute value
\newcommand{\ceil}[1]{\mdef{\left\lceil#1\right\rceil}}               % Absolute value
\newcommand{\ignore}[1]{}

\newif\ifnotes\notestrue
\ifnotes
\newcommand{\samson}[1]{\textcolor{purple}{{\bf (Samson:} {#1}{\bf ) }} \marginpar{\tiny\bf
             \begin{minipage}[t]{0.5in}
               \raggedright SZ:
            \end{minipage}}}            
            \newcommand{\jeremiah}[1]{\textcolor{red}{{\bf (Jeremiah:} {#1}{\bf ) }} \marginpar{\tiny\bf
             \begin{minipage}[t]{0.5in}
               \raggedright JB:
            \end{minipage}}}  	
            \newcommand{\seunghoon}[1]{\textcolor{blue}{{\bf (Seunghoon:} {#1}{\bf ) }} \marginpar{\tiny\bf
             \begin{minipage}[t]{0.5in}
               \raggedright SL:
            \end{minipage}}}						
\else
\newcommand{\samson}[1]{}
\newcommand{\jeremiah}[1]{}
\newcommand{\seunghoon}[1]{}
\fi

\makeatletter
\renewcommand*{\@fnsymbol}[1]{\textcolor{mahogany}{\ensuremath{\ifcase#1\or *\or \dagger\or \ddagger\or
 \mathsection\or \triangledown\or \mathparagraph\or \|\or **\or \dagger\dagger
   \or \ddagger\ddagger \else\@ctrerr\fi}}}
\makeatother

\providecommand{\email}[1]{\href{mailto:#1}{\nolinkurl{#1}\xspace}}

\definecolor{mahogany}{rgb}{0.75, 0.25, 0.0}
\definecolor{darkblue}{rgb}{0.0, 0.0, 0.65}
\definecolor{darkpastelgreen}{rgb}{0.01, 0.75, 0.24}
\definecolor{darkgreen}{rgb}{0.0, 0.2, 0.13}
\definecolor{darkgoldenrod}{rgb}{0.72, 0.53, 0.04}
\definecolor{darkred}{rgb}{0.55, 0.0, 0.0}
\hypersetup{
     colorlinks   = true,
     citecolor    = mahogany,
		 linkcolor		= darkblue
}

%% file: abstract.tex
% !TEX root = tcc-main.tex
\begin{abstract}
The classical (parallel) black pebbling game is a useful abstraction which allows us to analyze the resources (space, space-time, cumulative space) necessary to evaluate a function $f$ with a static data-dependency graph $G$. Of particular interest in the field of cryptography are data-independent memory-hard functions $f_{G,H}$ which are defined by a directed acyclic graph (DAG) $G$ and a cryptographic hash function $H$. The pebbling complexity of the graph $G$ characterizes the amortized cost of evaluating $f_{G,H}$ multiple times as well as the total cost to run a brute-force preimage attack over a fixed domain $\mathcal{X}$, i.e., given $y \in \{0,1\}^*$ find $x \in \mathcal{X}$ such that $f_{G,H}(x)=y$. While a classical attacker will need to evaluate the function $f_{G,H}$ at least $m=|\mathcal{X}|$ times a quantum attacker running Grover's algorithm only requires $\O{\sqrt{m}}$ blackbox calls to a quantum circuit $C_{G,H}$ evaluating the function $f_{G,H}$. Thus, to analyze the cost of a quantum attack it is crucial to understand the space-time cost (equivalently width times depth) of the quantum circuit $C_{G,H}$. We first observe that a legal black pebbling strategy for the graph $G$ does not necessarily imply the existence of a quantum circuit with comparable complexity --- in contrast to the classical setting where any efficient pebbling strategy for $G$ corresponds to an algorithm with comparable complexity for evaluating $f_{G,H}$. Motivated by this observation we introduce a new parallel reversible pebbling game which captures additional restrictions imposed by the No-Deletion Theorem in Quantum Computing. We apply our new reversible pebbling game to analyze the reversible space-time complexity of several important graphs: Line Graphs, Argon2i-A, Argon2i-B, and DRSample. Specifically, (1) we show that a line graph of size $N$ has reversible space-time complexity at most $\O{N^{1+\frac{2}{\sqrt{\log N}}}}$. (2) We show that any $(e,d)$-reducible DAG has reversible space-time complexity at most $\O{Ne+dN2^d}$. In particular, this implies that the reversible space-time complexity of Argon2i-A and Argon2i-B are at most $\O{N^2 \log \log N/\sqrt{\log N}}$ and $\O{N^2/\sqrt[3]{\log N}}$, respectively. (3) We show that the reversible space-time complexity of DRSample is at most $\O{N^2 \log \log N/\log N}$. We also study the cumulative pebbling cost of reversible pebblings extending a (non-reversible) pebbling attack of Alwen and Blocki on depth-reducible graphs. % It is an open question to construct a constant indegree graph with reversible space-time complexity $\Omega(N^2)$ or to give a general pebbling attack showing that no such graph exists. 

%(1) We provide a recursive quantum pebbling attack which shows that the space-time complexity of the line graph is at most $\O{N^{1+\epsilon}}$ for any constant $\epsilon$.
\keywords{Parallel Reversible Pebbling \and Argon2i \and DRSample \and Data-Independent Memory-Hard Function}
\end{abstract}

%% file: ec-intro.tex
% !TEX root = tcc-main.tex
\section{Introduction}
\seclab{sec:intro}
The (parallel) black pebbling game~\cite{HP70,Coo73} is a powerful abstraction which can be used to analyze the resources (space, space-time, amortized space-time) necessary to evaluate any function $f_G$ with a static data-dependency graph $G$. In the black pebbling game we are given a directed acyclic graph (DAG) $G = (V,E)$ where nodes intuitively represent intermediate data values and edges represent dependencies between these values, e.g., if $z = x \times y$ then we would add directed edges from nodes $x$ and $y$ to node $z$ to indicate that $x$ and $y$ are required to compute $z$. However, while the parallel black pebbling game is a useful abstraction for classical computation it is not a suitable model for reversible computation as in quantum computation. In this paper, we introduce a parallel reversible pebbling game as an abstraction which can be used to analyze the resources required to build a reversible quantum circuit evaluating our function $f_G$. We use the parallel reversible pebbling game to analyze the space-time cost of several important graphs (the line graph, Argon2i-A, Argon2i-B, DRSample) associated with prominent data-independent memory-hard functions (iMHFs) --- used in cryptography to design egalitarian proof of work puzzles and to protect low-entropy secrets (e.g., passwords) against brute-force attacks.  

\paragraph{Review: Parallel Black Pebbling.}
The classical parallel black pebbling game begins with no pebbles on the graph ($P_0=\{\}$), and during each round of the pebbling game, we may only place a new pebble on a node $v$ if all of $v$'s parents were pebbled in the previous round. Intuitively, if the data value $X_v$ corresponding to node $v$ is computed as  $X_v \coloneqq H(X_u ,X_{v-1})$ then $G$ would include directed edges $(u,v)$ and $(v-1,v)$ indicating that we cannot compute value $X_v$ (resp. place a pebble on node $v$) unless $X_u$ and $X_{v-1}$ are already available in memory (resp. we already have pebbles on nodes $u$ and $v-1$). More formally, if $P_i \subseteq V$ denotes the set of pebbled nodes during round $i$, then we require that $\parents(P_{i+1}\setminus P_i,G) \subseteq P_i$ where $\parents(S,G) = \bigcup_{v \in S} \{u:(u,v) \in E\}$. In the black pebbling game we are given a subset $T\subseteq V$ of target nodes (corresponding to output data values) and the goal of the black pebbling game is to eventually place a pebble on each node in $T$. A pebbling $P=(P_0, P_1,\ldots, P_t)$ is legal if $P_0=\{\}$ and $\parents(P_{i+1}\setminus P_i,G) \subseteq P_i$ for each $i < t$.  Intuitively, the requirement that $\parents(P_{i+1}\setminus P_i,G) \subseteq P_i$ enforces the natural constraint that we cannot compute a new data value before all dependent data values are available in memory.  In the sequential pebbling game, we additionally require that $\left| P_{i+1}\setminus P_i\right| \leq 1$ so that only one new pebble can be placed on the graph in each round  while the parallel pebbling game has no such restriction.  Thus, a legal parallel (resp. sequential) pebbling of a data-dependency graph $G$ naturally corresponds to a parallel (resp. sequential)  algorithm to compute $f_G$ and the number of pebbles $|P_i|$ on the graph in each round $i$ corresponds to memory usage during each round of computation. 

The sequential black pebbling game has been used to analyze space complexity~\cite{HPV77,PTC76} and to examine space-time tradeoffs~\cite{Cob66,Coo73,Pau75,PV76,Tom81}. In the field of cryptography, the parallel black pebbling game has been used to analyze the security of data-independent memory-hard functions (iMHFs). An iMHF $f_{G,H}$ is defined using a cryptographic hash function $H$ and a data-dependency graph $G$~\cite{STOC:AlwSer15,C:AlwBlo16,EC:AlwBloPie17,TCC:BloZho17}. The output of $f_{G,H}(x)$ is defined to be the label $X_N$ of the final sink node $N$ in $G$ where the label $X_1=H(X)$ of the first (source) node is obtained by hashing the input and the label of each internal node $v$ is obtained by hashing the labels of all of $v$'s parents, e.g., if $\parents(v,G)=\{u,v-1\}$ then we would set $X_v=H(X_u,x_{v-1})$. In many cryptographic applications (e.g., password hashing), we want to ensure that it is moderately expensive to evaluate $f_{G,H}$ to ensure that a brute-force pre-image attack (given $y$ find some $x$ such that $f_{G,H}(x)=y$) is prohibitively expensive even when the domain $\mathcal{X}$ of inputs is smaller (e.g., low entropy passwords).  When modeling the cryptographic hash function $H$ as a random oracle, one can prove that the cost to evaluate $f_{G,H}$ in the parallel random oracle model is exactly captured by the pebbling cost of $G$~\cite{STOC:AlwSer15,TCC:AlwTac17,EC:AlwBloPie18}. Thus, we would like to pick a graph $G$ with high pebbling costs and/or understand the pebbling costs associated with candidate iMHFs. Prior work demonstrated that the amortized space-time complexity of prominent iMHF candidates, including Password Hashing Competition winner Argon2i, was lower than previously hoped~\cite{C:AlwBlo16,EC:AlwBloPie17,ESP:AlwBlo17,TCC:BloZho17}. On the positive side, recent work has shown how to use depth-robust graphs~\cite{EGS75} to construct iMHFs with (essentially) optimum amortized space-time complexity~\cite{EC:AlwBloPie17,CCS:AlwBloHar17,C:BHKLXZ19}. However, it is important to note that the classical black pebbling game does not include any rules constraining our ability to remove pebbles. We are allowed to remove pebbles from the graph at any point in time which corresponds to freeing memory and can be done to reduce the space usage. While the classical pebbling game allows us to discard pebbles at any point in time to free memory, this action is often not possible in a quantum circuit due to the No-Deletion Theorem \cite{PatBra00}. In this sense, the black pebbling game cannot be used to model reversible computation as in a quantum circuit and an efficient parallel black pebbling for a graph $G$ does not necessarily imply the existence of a quantum circuit $C_{G,H}$ with comparable cost.
  
\paragraph{Review: Measuring Pebbling Costs.} There are several natural ways to measure the cost of a pebbling. The space cost of a pebbling $P=(P_0,\ldots,P_t)$ measures the maximum number of pebbles on the graph during any round, i.e., $\max_i |P_i|$ and the space complexity of a graph measures the minimum space cost over all legal pebblings of $G$. Similarly, the space-time cost of a pebbling $P=(P_0,\ldots,P_t)$ measures the product $t \times \max_i |P_i|$ and the cumulative pebbling cost is $\sum_{i} |P_i|$. Intuitively, space complexity measures the amount of memory (e.g., RAM) required for a computation and space-time cost measures the full cost of the computation by telling how long the memory will be locked up during computation. Cumulative pebbling cost gives the amortized space-time complexity of pebbling multiple copies of the graph $G$, i.e., when we are evaluating our function $f_G$ on multiple different inputs in parallel~\cite{STOC:AlwSer15}. 

\paragraph{(Quantum) Pre-Image Attacks.} Understanding the amortized space-time complexity of a graph $G$ is important to estimate the cost of a classical brute-force pre-image attack over a domain $\mathcal{X}$ of size $m$. In particular, suppose we are given a target output $y$ (e.g., $y=f_{G,H}(x')$ for a secret input $x \in \mathcal{X}$) and we wish to find some input $x' \in \mathcal{X}$ such that $y= f_{G,H}(x')$. Classically, the space-time cost of a black-box pre-image attack would require us to evaluate the function  $f_{G,H}$ on $\Omega(m)$ inputs. If the cumulative pebbling cost of $G$ is given by $\sum_{i} |P_i|$ then the total space-time cost of the pre-image attack would scale proportionally to $m \sum_{i} |P_i|$, i.e., $m$ times the amortized space-time complexity. Thus, a more efficient black pebbling strategy for $G$ yields a lower-cost pre-image attack.  

In the context of quantum computing, Grover's algorithm \cite{STOC:Grover96} substantially reduces the cost of a brute-force pre-image attack over a domain $\mathcal{X}$ of size $m$. In particular, Grover's algorithm only requires $\O{\sqrt{m}}$ black-box queries to the function $f_{G,H}$ evaluating the function $f_{G,H}$ and this is optimal --- any quantum algorithm using $f_{G,H}$ as a black box must make at least $\Omega(\sqrt{m})$ queries \cite{BennettBBV97}. If we instantiate $f_{G,H}$ with a quantum circuit of width $w$ and depth $d$ then full Grover circuit would have width $W= \O{w}$ and depth $D = d \times \O{\sqrt{m}}$. In particular, the total space-time (equivalently width-depth) cost of the attack would be $wd \times \O{\sqrt{m}}$. Thus, to analyze the cost of a quantum pre-image attack it is crucial to understand the space-time (or width-depth) cost of a quantum circuit $C_{G,H}$ computing $f_{G,H}$. Our goal will be to treat $H$ as a black box and use  graph pebbling to characterize the space-time cost. A natural first attempt would be to use the classical black pebbling game to analyze the parallel pebbling cost of $G$ as above. If this approach worked we could simply leverage prior (parallel) black pebbling analysis of prominent iMHF candidates~\cite{C:AlwBlo16,EC:AlwBloPie17,ESP:AlwBlo17,TCC:BloZho17} to analyze the cost of a quantum pre-image attack.  Unfortunately, this approach breaks down because a legal black pebbling strategy {\em does not} necessarily correspond to a valid quantum circuit $C_{G,H}$ with comparable cost. Thus, we will require a different pebbling game to analyze the width-depth cost of the quantum circuit $C_{G,H}$.

\paragraph{Notation.} We use the notation $[N]$ (resp. $[a,b]$) to denote the set $\{1,\ldots,N\}$ (resp. $\{a,a+1,\ldots, b\}$) for a positive integer $N$ (resp. $a \leq b$). The notation $\sample$ denotes a uniformly random sampling, e.g., we say $x\sample[N]$ when $x$ is a uniformly sampled integer from $1$ to $N$. For simplicity, we let $\log(\cdot)$ be a $\log$ base $2$, i.e., $\log x\coloneqq\log_2x$. 

Let $G=(V,E)$ be a directed acyclic graph (DAG) where we denote $N$ to be the number of nodes in $V=[N]$.  
Given a node $v\in V$, we define $\parents(v,G)$ to be the \emph{immediate parents} of node $v$ in $G$, and we extend this definition to a subset of nodes as well; for a set $W\subseteq V$, we define $\parents(W,G)\coloneqq \bigcup_{w\in W}\{u:(u,w)\in E\}$. 
We let $\ancestors(v,G)$ be the set of all ancestors of $v$ in $G$, i.e., $\ancestors(v,G) \coloneqq \bigcup_{i\geq 1} \parents^i(v,G)$, 
where $\parents^1(v,G)=\parents(v,G)$ and $\parents^i(v,G)= \parents(\parents^{i-1}(v,G),G)$. Similarly, for a set $W\subseteq V$, we define $\ancestors(W,G)\coloneqq \bigcup_{i\geq 1}\parents^i(W,G)$, where $\parents^1(W,G)=\parents(W,G)$ and recursively define $\parents^i(W,G)=\parents(\parents^{i-1}(W,G),G)$. 

We denote the set of all sink nodes of $G$ with $\sinks(G)\coloneqq\{v\in V : \nexists(v,u)\in E\}$ -- note that $\ancestors(\sinks(G),G)=V$. We define $\depth(v,G)$ to refer to the number of the longest directed path in $G$ ending at node $v$ and we define $\depth(G)=\max_{v\in V}\depth(v,G)$ to refer to the number of nodes in the longest directed path in $G$. Given a node $v\in V$, we define $\indeg(v)\coloneqq|\parents(v,G)|$ to denote the number of incoming edges into $v$, and we also define $\indeg(G)\coloneqq\max_{v\in V}\indeg(v)$. Given a set $S\subseteq V$ of nodes, we use $G-S$ to refer to the subgraph of $G$ obtained by deleting all the nodes in $S$ and all edges that are incident to $S$. We also use the notation $S_{\leq k} \coloneqq S\cap [k]$ denotes the subset of $S$ that only intersects with $[k]$. We say that a DAG $G=(V,E)$ is \emph{$(e,d)$-depth robust} if for any subset $S\subseteq V$ such that $|S|\leq e$ we have $\depth(G-S)\geq d$. Otherwise, we say that $G$ is \emph{$(e,d)$-reducible} and call the subset $S$ a \emph{depth-reducing set} (which is of size at most $e$ and yields $\depth(G-S)<d$).

We denote with $\Peb_{G,T}$ and $\pPeb_{G,T}$ the set of all legal sequential and parallel \emph{classical} pebblings of $G$ with target set $T$, respectively. In the case where $T=\sinks(G)$, we simply write $\Peb_G$ and $\pPeb_G$, respectively.

\subsection{Our Results} We introduce the parallel reversible pebbling game as a tool to analyze the (amortized) space-time cost of a quantum circuit evaluating a function $f$ with a static data-dependency graph $G$. Prior work~\cite{Bennett89,10.1007/3-540-45627-9_26,8715092} introduced a sequential reversible pebbling game. As we discuss, there are several key subtleties that arise when extending the sequential reversible pebbling game to the parallel setting. We argue that any parallel reversible pebbling $P=(P_0,\ldots, P_t)$ of the graph $G$ corresponds to a quantum circuit $C_P$ evaluating $f$ with comparable costs, e.g., the depth of the quantum circuit $C_P$ corresponds to the number of pebbling rounds $t$ and the width of the circuit corresponds to the space complexity of the pebbling, i.e., $\max_i |P_i|$. Thus, any reversible pebbling attack will yield a more efficient quantum pre-image attack\footnote{While one could use the parallel reversible pebbling game as a heuristic to {\em lower bound} the cost of a quantum pre-image attack we stress that, at this time, there is no pebbling reduction which provably lower bounds the cost of a quantum pre-image attack on $f_{G,H}$ using reversible pebbling cost of the underlying DAG $G$. We do have pebbling reductions for classical (non-reversible) pebblings in the parallel random oracle model~\cite{STOC:AlwSer15}, but there are several technical barriers which make it difficult to extend this reduction to the quantum random oracle model. }.

As an application, we use the parallel reversible pebbling game to analyze the space-time cost of several important password hashing functions $f_{G,H}$ including PBKDF2, BCRYPT, Argon2i, and DRSample.

\paragraph{Reversible Pebbling Attacks on Line Graphs.} We first focus on analyzing the reversible pebbling cost of a line graph $L_N$ with $N$ nodes $\{1,\ldots, N\}$ and edges $(i,i+1)$ for each $1 \leq i <N$.  Classically, there is a trivial black pebbling strategy for the line graph with simply walks a single pebble from node $1$ to node $N$ over $N$ pebbling rounds, i.e., in each round $i$ we place a new pebble on node $i$ and then delete the pebble on node $i-1$. This pebbling strategy is clearly optimal as the maximum space usage is just $1$ and the space-time cost is just $N \times 1 = N$. However, this simple pebbling strategy is no longer legal in the reversible pebbling game and it is a bit tricky just to find a reversible pebbling strategy whose space-time cost is significantly lower than $\O{N^2}$ --- the space-time cost of the na\"ive pebbling strategy which avoids removing pebbles. In \thmref{thm:line} we show that the (sequential) reversible space-time complexity of a line graph is $\O{N^{1+\frac{2}{\sqrt{\log N}}}}$. A similar argument seems to be implicitly assumed by Bennett~\cite{Bennett89} though the argument was never explicitly formalized as a reversible pebbling strategy. The result improves upon a result of Li and Vit\'{a}nyi~\cite{LiVit96} who showed that the space-time complexity is at most  $\O{N^{\log 3} \log N}$\footnote{The pebbling of Li and Vit\'{a}nyi~\cite{LiVit96} runs in time $\O{N^{\log 3}}$ while using at most $\O{\log N}$ pebbles. Our pebbling strategy uses more pebbles to reduce the  overall space-time cost by improving the pebbling time.}.

%We extend Bennett's (sequential) reversible pebbling strategy \cite{Bennett89} to get a parallel reversible pebbling attack with space-time complexity $\O{N^{1+\frac{2}{\sqrt{\log N}}}}$.
%We give a recursive quantum pebbling attack with space-time complexity $\O{N^{1+\epsilon}}$ for any constant $\epsilon >0$.

Because the space-time complexity of the line graph $G=L_N$ is so low, it is a poor choice for an iMHF $f_{G,H}$ or for password hashing~\cite{SP:BloHarZho18}. However, the line graph $L_N$ naturally corresponds to widely deployed password hashing algorithms like BCRYPT \cite{USENIX:ProMaz99} and PBKDF2 \cite{rfc2898} which use hash iteration to increase costs where the parameter $N$ controls the number of hash iterations. Thus, to understand the cost of a (quantum) brute-force password cracking attack it is useful to analyze the (reversible) pebbling cost of $L_N$.  
%\seunghoon{Check this paragraph}

\paragraph{Reversible Pebbling Attack for Depth-Reducible DAGs.} In \thmref{thm:ed-reducible} we give a generic parallel reversible pebbling attack on any $(e,d)$-reducible DAG $G$ with space-time cost $\O{Ne + dN2^d}$ which corresponds to a meaningful attack whenever $e=o(N)$ and $d2^d = o(N)$. A DAG $G$ is said to be $(e,d)$-reducible if there is a subset $S\subseteq V$ of at most $e$ nodes such that any length $d$ path $P$ in $G$ contains at least one node in $S$. As we show this leads to meaningful reversible pebbling attacks on Argon2i, the winner of the Password Hashing Competition. Specifically, we demonstrate how to construct depth-reducing sets for Argon2i-A (an older version of Argon2i) and Argon2i-B (the current version of Argon2i) with $e=o(N)$ and $d2^{d}=o(N)$. This leads to 
reversible pebbling attacks with space-time complexity $\O{N^2\log\log N/\sqrt{\log N}}$ and $\O{N^2/\sqrt[3]{\log N}}$ against Argon2i-A and Argon2i-B, respectively --- see \corref{cor:argon2i}. 

In the classical pebbling setting, Alwen and Blocki \cite{C:AlwBlo16} previously gave a generic pebbling attack on $(e,d)$-reducible DAGs with amortized space-time cost $\O{Ne + N^2 d/e}$. However, this pebbling attack is not legal in the reversible setting, and without amortization, the space-time cost is still $N^2$ --- the average number of pebbles on the graph per round is just $e+Nd/e$ but at the peak, the pebbling strategy still requires $\Omega(N)$ pebbles. In our pebbling strategy, the maximum space usage is $\O{e+d2^d}$.

\paragraph{Reversible Pebbling Attack against DRSample.} Finally, we use the parallel reversible pebbling game to analyze DRSample~\cite{CCS:AlwBloHar17} --- a proposal to update the edge distribution in Argon2i with a depth-robust graph. With high probability, a randomly sampled DRSample DAG $G$ will not be $(e,d)$-reducible for parameters $e,d$ as large as $e=\Omega(N/\log N)$ and $d=\Omega(N)$. Thus, the generic reversible pebbling attack on $(e,d)$-reducible graphs does not seem to apply. We give an alternate pebbling strategy by partitioning the nodes of $G$ into $\lceil N/b\rceil$ consecutive blocks of size $b$ and converting a parallel reversible pebbling of the line graph $L_{
\lceil N/b\rceil}$ into a legal reversible pebbling of $G$. The reversible pebbling strategy will be cost-effective as long as we have an efficient pebbling strategy for $L_{\lceil N/b\rceil}$ and the graph $G$ does not contain too many ``long'' edges $(u,v)$ with $|v-u| \geq b$ --- we show that DRSample does not contain too many long edges when $b=N/\log^2 N$. Combined with our parallel reversible pebbling strategies for the line graph, this leads to an attack on DRSample with space-time cost at most $\O{N^2 \log \log N/\log N}$ --- see \corref{cor:stcostDRS}. 

More generally, in \thmref{thm:reduce} we give an efficient reversible pebbling algorithm which transforms a legal reversible pebbling $P'=(P_1',\ldots,P_{t'}')$ of the line graph $L_{\lceil N/b\rceil}$ into a legal reversible pebbling $P=(P_1,\ldots,P_t)$ of a DAG $G=(V,E)$. The reversible pebbling requires $t=\O{b t'}$ rounds and space $b s' + (\#skip)$ where $\#skip$ is upper bounded by the number of long edges $(u,v) \in E$ with $|v-u|\geq b$ and $s' = \max_i |P_i'|$ upper bounds the space usage of the pebbling $P'$. Thus, the total space-time complexity will be $\O{b^2 s't' + N \#skip}$ and we will be able to obtain an efficient reversible pebbling attack as long as $b=o(N)$ and $(\#skip)=o(N)$ --- we show that this is the case for DRSample. 

\paragraph{Cumulative Pebbling Cost and Parallel Reversible Pebbling.} Alwen and Blocki~\cite{C:AlwBlo16} gave a general parallel black pebbling attack on any $(e,d)$-reducible graph. This general pebbling attack was used to upper bound the cumulative cost of many prominent iMHFs including Argon2i-A \cite{C:AlwBlo16} and Argon2i-B \cite{ESP:AlwBlo17}. More generally the attack shows that {\em any} constant indegree DAG $G$ has cumulative pebbling cost at most $\O{N^2 \log \log N/\log N}$. We show how the pebbling attack of Alwen and Blocki~\cite{C:AlwBlo16} can be extended to the parallel reversible pebbling game\footnote{Alwen, Blocki and Pietrzak~\cite{EC:AlwBloPie17} later provided a recursive version of the pebbling attacks of Alwen and Blocki ~\cite{C:AlwBlo16} which can further reduces the cumulative pebbling cost of a DAG which is $(e_i,d_i)$-reducible at a sequence of points $(e_i,d_i)$ with $d_i < d_{i-1}$ and $e_i \geq d_{i-1}$. The recursive pebbling attack yields tighter asymptotic upper bounds for some iMHF candidates \cite{TCC:BloZho17,EC:AlwBloPie17}. We conjecture that these recursive pebbling attacks can also be generalized to the reversible pebbling setting though we leave this as an open problem. }. In particular, we can show that the cumulative reversible pebbling costs of an $(e,d)$-reducible DAG with maximum indegree $\delta$ is upper bounded by $\O{e N + g\delta N + \frac{N^2 d}{g}}$ for any parameter $g \geq d$ matching the non-reversible pebbling attacks of Alwen and Blocki~\cite{C:AlwBlo16} --- see \thmref{thm:revcc}. More specifically, since any DAG $G$ with constant indegree $\delta=O(1)$ is $(e,d)$-reducible with $d=N/\log^2 N$ and $e= \O{N \log \log N/\log N}$~\cite{C:AlwBlo16} we can plug in $g=e$ to obtain a reversible pebbling strategy with cumulative cost at most $\O{N^2 \log \log N/\log N}$ --- see \corref{generalCMCDAG}. We can also upper bound the cumulative reversible pebbling costs of Argon2i-A and Argon2i-B as $\O{N^{1.75} \log N}$ and $\O{N^{1.8}}$ respectively\FullVersion{ --- see \corref{Argon2iCMC}}{ --- see the full version for the details}.

\subsection{Technical Overview}\seclab{sec:overview}

\paragraph{Defining the Parallel Reversible Pebbling Game.} We begin by defining and motivating the parallel reversible pebbling game. We want to ensure that any legal (parallel) reversible pebbling strategy for $G$ corresponds to a quantum circuit $C_{G,H}$ evaluating $f_{G,H}$ that could be used as part of a pre-image attack using Grover's algorithm.

We first consider the parallel quantum random oracle model~\cite{AC:BDFLSZ11} where the random oracle is a function $H:\{0,1\}^{\leq 2\lambda}\rightarrow \{0,1\}^\lambda$. In the parallel quantum random oracle model we are given access to a quantum oracle maps basis states of the form $|x_1 , y_1,\ldots, x_k,y_k, z \rangle$ to the new state $|x_1 , y_1 \oplus H(x_1),\ldots, x_k,y_k \oplus H(x_k), z \rangle$. Here, $x_1,\ldots,x_k$ denote the queries, $y_1,\ldots, y_k$ denote the output registers and $z$ denotes any auxiliary data. Notice that if $y_i = 0^\lambda$ then the $i$\th output register will just be $H(x_i)$ after the query is submitted. 

Now consider the function $f(x) = H^N(x)$ where $H^1(x)=H(x)$ and $H^{i+1}(x) = H(H^i(x))$. The data-dependency graph for $f$ is simply the line graph $G=L_N$. In our reversible pebbling game, we want to ensure that each pebbling transition corresponds to a legal state transition in the quantum random oracle model. If $N=5$, then the pebbling configuration $P_i=\{2, 3, 4\}$ intuitively corresponds to a quantum state containing the labels $X_2=H^2(x)$, $X_3=H^3(x)$ and $X_4=H^4(x)$. From this state, we could use $X_4$ and an input register and submit the query $|X_4, 0^\lambda \rangle$ to the random oracle to obtain  $X_5=H(X_4)$ from the resulting state $|X_4, H(X_4) \rangle$. Similarly, while we cannot simply delete $X_3$ we could uncompute this value by using $X_3$ as an output register and submitting the random oracle query $|X_2, X_3\rangle$ to obtain the new state $|X_2, H(X_2) \oplus X_3\rangle= |X_2, 0^\lambda \rangle$ in which the label $X_3$ has been removed. However, without the label $X_1$ there is no way to uncompute $X_2$ without first recomputing $X_1$. 

The above example suggests that we extend the parallel pebbling game by adding the rule that  $\parents(P_{i} \setminus P_{i+1},G) \subseteq P_i$, i.e., a pebble can only be deleted if all of its parents were pebbled at the end of the previous pebbling round. While this rule is necessary, it is not yet sufficient to prevent impossible quantum state transitions. In particular, the rule would not rule out the pebbling transition from $P_i = \{1,2,\ldots ,i\}$ to the new configuration $P_{i+1}= \{\}$ where all labels have been removed from memory. This pebbling transition would correspond to a quantum transition from a state in which labels $X_1,\ldots, X_i$ are stored in memory to a new state where all of these labels have been uncomputed after just one (parallel) query to the random oracle. Because quantum computation is reversible this would also imply that we could directly transition from the original state (no labels computed) to a state in which all of the labels $X_1,\ldots, X_i$ are available after just one (parallel) query to the quantum random oracle.  However, it is known that computing $X_i=H^i(x)$ requires at least $i$ rounds of computation even in the parallel quantum random oracle model~\cite{ITC:BloLeeZho21}. Thus, the pebbling transition from $P_i = \{1,2,\ldots ,i\}$ to $P_{i+1}= \{\}$ must be disallowed by our reversible pebbling rules as the corresponding quantum state transition is impossible.

We address this last issue by adding another pebbling rule: if $v \in \parents(P_i \setminus P_{i-1},G) \cup \parents(P_{i-1}\setminus P_i,G)$, then $v \in P_i$. Intuitively, the rule ensures that if the label $X_v$  appeared in an input register to either compute or uncompute some other data label then we cannot also uncompute $X_v$ in this round, i.e., we must keep a pebble at node $v$.

We make several observations about the reversible pebbling game. First, any legal reversible pebbling of a DAG $G$ is also a legal (classical) parallel black pebbling of $G$ since we only added additional pebbling restrictions. More formally, if $\pPeb_G$ (resp. $\Peb_G$) denotes the set of all legal parallel (resp. sequential) black pebblings of $G$ and $\pqPeb{G}$ (resp. $\qPeb{G}$) denotes the set of all legal parallel (resp. sequential) reversible pebblings of $G$ then we have $\pqPeb{G} \subseteq \pPeb_G$ and $\qPeb{G} \subseteq \Peb_G$. Thus, any lower bounds  on the classical parallel pebbling cost of $G$ will immediately carry over to the reversible setting. However, upper bounds will not necessarily carry over since classical pebbling attacks may not be legal in the reversible pebbling game. Second, we observe that the following sequential reversible pebbling strategy works for any DAG $G=(V=[N],E)$. In the first $N$ rounds, pebble all nodes in topological order without deleting any pebbles. In the next $N-1$ rounds remove pebbles from all nodes (excluding $\sinks(G)$) in reverse topological order. More formally, assuming that $1,\ldots, N$ is a topological order and that node $N$ is the only sink node we have $P_i = [i]$ for each $i \leq N$ and $P_{N+j} = [N] \setminus [N-j,N-1]$ for each $j \leq N-1$. The pebbling requires $N$ pebbles and finishes in $t=2N-1$ rounds so the space-time cost is $2N^2-N$. We refer to the above sequential strategy as the na\"ive  reversible pebbling for a graph $G$.

\paragraph{Reversible Pebbling Attack on Line Graphs.}
We give a reversible pebbling attack on a line graph $L_N$ of size $N$ with the space-time cost $\O{N^{1+\frac{2}{\sqrt{\log N}}}}$. This can be achieved by generalizing Li and Vit\'{a}nyi's work \cite{LiVit96}. Li and Vit\'{a}nyi \cite{LiVit96} gave a reversible pebbling strategy on a line graph of size $N$ with space-time cost $\O{N^{\log 3}\log N}$ by translating ideas of Bennett \cite{Bennett89} into a reversible pebbling argument. Intuitively, if we define $N(k)$ using the recurrence relationship $N(k) = k + \sum_{j=0}^{k-1} N(j)$, solving to $N(k)=2^k-1$, then they show that the line graph with $N(k)$ nodes can be pebbled using space $S(k)= S(k-1)+1 = k$ and time $T(k) = 3T(k-1)+1 = \O{3^k}$ for a total space-time cost of $\O{k3^{k}} = \O{ (N(k))^{\log 3} \log N(k)}$. Their pebbling strategy works as follows: (1) recursively apply the pebbling strategy to place a pebble on node $N(k-1)$ using space at most $S(k-1)$ and time at most $T(k-1)$, (2) place a pebble on node $v_1=N(k-1)+1$, (3) recursively apply the strategy (in reverse) to clear any leftover pebbles from nodes $1$ to $N(k-1)$ in time $T(k-1)$ and (additional) space at most $S(k-1)$. We are left with $(k-1) + \sum_{j=1}^{k-2} N(j) = N(k-1)$ remaining nodes which will be handled recursively using time $T(k-1)$ and (additional) space $S(k-1)$. 

We observe that by increasing the space usage slightly we can decrease the pebbling time to obtain a superior space-time cost. We note that Bennett \cite{Bennett89} mentions a similar idea in his paper, but that this idea was not formalized as a reversible pebbling strategy either by Bennett \cite{Bennett89} or by Li and Vit\'{a}nyi \cite{LiVit96}. The key modification is as follows: we redefine $N(k) = ck + \sum_{j=0}^{k-1} c N(j)$ solving to $N(k)=\Theta\left( (c+1)^k\right)$. We can now recursively pebble a line graph with $N(k)$ nodes in sequential time $T(k) = (2c+2) T(k-1) + c = \O{(2c+2)^k}$ and space $S(k) = c+S(k-1) = ck$. Intuitively, the recursive pebbling strategy will begin by dropping pebbles on each of the nodes $N(k-1)+1, 2N(k-1)+2,...,cN(k-1)+c$ using space at most $S(k-1)+c$ and time $2c \cdot T(k-1)$. We are left with $c(k-1) + \sum_{j=0}^{k-2} c N(j) = N(k-1)$ remaining nodes which can then be handled recursively. Setting $c=2^k$, we have $k=\Theta(\sqrt{\log N(k)})$ yielding an upper bound of  $\O{N(k)^{1+(2+o(1))\frac{1}{\sqrt{\log N(k)}}}}$ on the sequential  space-time cost. 

We can obtain a minor improvement by exploiting parallelism to save time while increasing space usage slightly. In particular, our parallel strategy uses space $\O{c 2^k}$ and time $\O{(c+2)^k}$ with total space-time cost $\O{c (2c+4)^k}$. Setting $c+1 =2^k$ we have a slightly better upper bound $\O{N(k)^{1+\frac{2}{\sqrt{\log N(k)}}}}$ on the space-time cost. Further details can be found in \FullVersion{\appref{app:line}}{the full version}.

\paragraph{Generic Reversible Pebbling Attack on Depth-Reducible Graphs.} We give a generic reversible pebbling attack on any $(e,d)$-reducible DAG $G=(V=[N],E)$ with maximum indegree $2$. The space-time cost of our reversible pebbling attack is at most $\O{Ne + N d2^d}$. Thus, the attack will be superior to the na\"ive reversible pebbling strategy as long as $e = o(N)$ and $d2^d =o(N)$. We begin with a depth-reducing set $S \subseteq V$ of size $|S|\leq e$. Our reversible pebbling strategy will never remove pebbles from the set $S$ until all of the sink nodes in $G$ are pebbled and we are ready to remove pebbles from the remaining nodes. On each round $i \leq N$ we will place a new pebble on node $\{i\}$. To ensure that this step is legal, we consider the subgraph formed by all of node $i$'s ancestors in $G-S$. Since $G-S$ does not contain a directed path of length $d$ and each node has at most $2$ parents there are at most $2^d$ ancestors of node $i$ in $G-S$. Once again applying the observation that the depth of $G-S$ is at most $d$ we can start to repebble $i$'s ancestors in round $i-d-1$ to ensure that $i$'s immediate parents are pebbled by round $i-1$. After we place a pebble on node $i$ we can remove pebbles from $i$'s ancestors in $G-S$  over the next $d$ rounds. Since we only keep pebbles on the set $S$ and the ancestors of {\em up to} $2d$ nodes in $G-S$, the maximum space usage of this reversible pebbling strategy will be $\O{e+d 2^d}$. 

We apply the generic attack to Argon2i-A and Argon2i-B. In particular, we apply ideas from the previous work~\cite{ESP:AlwBlo17,TCC:BloZho17} to show that Argon2i-A (resp. Argon2i-B) graphs are $(e,d)$-reducible with $e=\O{N \log \log N/\sqrt{\log N}}$ and $d=\log N/\log \log N$ (resp. $e=\O{N/\sqrt[3]{\log N}}$ and $d=(\log N)/2$). This leads to reversible pebbling attacks with cost $\O{N^2 \log \log N/\sqrt{\log N}}$ and $\O{N^2/\sqrt[3]{\log N}}$) for Argon2i-A and Argon2i-B, respectively. An intriguing open question is whether or not these are the best reversible pebbling attacks for Argon2i-A and Argon2i-B? 

\paragraph{Reversible Pebbling Attack on DRSample.} We provide a general reversible pebbling attack on any DAG $G$ with the property that $G$ contains few  \emph{skip nodes} (defined below). Intuitively, given a DAG $G=(V,E)$ with $|V|=N$ and a parameter $b \geq 1$, we can imagine partitioning the nodes of $V$ into consecutive blocks $B_1=\{v_1,\ldots,v_b\},B_2=\{v_{b+1},\ldots,v_{2b}\},\ldots, B_{\lceil N/b \rceil}=\{v_{(\lceil N/b \rceil-1)b+1},\ldots,v_N\}$ such that we have $\lceil N/b \rceil$ blocks in total and each block contains exactly $b$ nodes (with the possible exception of the last block if $N/b$ is not an integer). We call a node $u$ in block $B_i$ a \emph{skip node} if $G$ contains a directed edge $(u,v)$ from $u$ to some node $v \in B_j$ with $j> i+1$ and we call the edge $(u,v)$ a \emph{skip edge}, i.e., the edge $(u,v)$ skips over the block $B_{i+1}$ entirely. 

We first observe that if the graph $G$ contained no skip edges then it would be trivial to transform a (parallel) reversible pebbling $P'$ of the line graph $L_{\lceil N/b\rceil}=(V',E')$ with space-time cost $\pqpeb_\Cspacetime(P')$  into a (parallel) reversible pebbling $P$ of $G$  with space-time cost $\O{b^2 \pqpeb_\Cspacetime(P') } $ (see \defref{def:complexity} for the definition of $\pqpeb_\Cspacetime(\cdot)$). In particular, placing a pebbling on node $v' \in V'$ of the line graph corresponds to $b$ rounds in which we pebble all nodes in block $B_{v'}$. Thus, the pebbling time increases by a factor of $\O{b}$, and the total space usage also increases by a factor $b$. Unfortunately, this strategy may result in an illegal reversible pebbling when $G$ contains skip edges. However, we can modify the above strategy to avoid removing pebbles on skip nodes which intuitively increases our space usage by $s$ --- the total number of skip nodes in the graph $G$. \FullVersion{The procedure $P=\trans(G,P',b)$ is formally described in \algref{alg:trans} in \appref{algorithms}, and an example for the reversible pebbling strategy can be found in in \figref{fig:metapebbling} in \appref{examples}}{The procedure $P=\trans(G,P',b)$ and an example for the reversible pebbling strategy are formally described in the full version}. As long as $s$ is sufficiently small, we obtain an efficient parallel reversible pebbling attack on $G$. In particular, given a reversible pebbling $P'$ of the line graph $L_{\lceil N/b\rceil}=(V',E')$ with space-time cost $\pqpeb_\Cspacetime(P')$  we can find a reversible pebbling $P$ of $G$ with space-time cost $\O{s N + b^2 \pqpeb_\Cspacetime(P') } $. Combining this observation with our efficient reversible pebbling attacks on the line graph we can see that the space-time costs will be at most $\O{s N + b^2 (N/b)^{1+\frac{2}{\sqrt{\log (N/b)}}}} $. For graphs like DRSample \cite{CCS:AlwBloHar17}, we can show that (whp) the number of skip nodes is at most $s= \O{\frac{N \log \log N}{\log N}}$ when we set the block size $b=\O{\frac{N}{\log^2 N}}$ leading to a reversible pebbling attack with space-time cost $\O{\frac{N^2 \log \log N}{\log N}}$.

\paragraph{Cumulative Cost for Reversible Pebblings: Depth-Reducing Reversible Pebbling Attacks.}
Alwen and Blocki~\cite{C:AlwBlo16} gave a non-reversible pebbling attack with reduced cumulative pebbling cost for any $(e,d)$-reducible DAG $G$. While their pebbling attack is non-reversible, we observe that almost all pebbling rounds respect the constraints of reversible pebbling. We then identify the few non-reversible rounds and how these steps can be patched to respect the additional constraints of reversible pebbling. See details in \secref{sec:revpeb}.

\subsection{Related Work}\seclab{sec:relatedwork}

\paragraph{Related Pebbling Games.} 
Prior work \cite{Bennett89,10.1007/3-540-45627-9_26,8715092} introduced a reversible pebbling game to capture restrictions imposed by the Quantum No-Deletion Theorem and analyze space-time tradeoffs in quantum computing. However, the pebbling game considered in these works is sequential and only allows for the addition/removal of one pebble in each round. Thus, the sequential reversible pebbling game is not suitable for analyzing the space-time cost of a quantum circuit evaluating $f_{G,H}$ since the circuit can evaluate $H$ multiple times in parallel. We note that there are several important subtleties that must be considered when extending the game to the parallel setting. 

More recently, Kornerup et al.~\cite{kornerup2021spooky} introduced a new (sequential) pebbling game called the {\em spooky pebble game} to model measurement-based deletion in quantum computation. Intuitively, measurement-based deletion allows for the conversion of some qubits into (cheaper) classical bits which can later be used to restore the quantum state. The spooky pebble game only allows for sequential computation and the cost model ignores classical storage. One disadvantage of instantiating a spooky pebbling attack as part of a quantum pre-image attack is that the final attack requires many intermediate measurements which introduces additional technical challenges, i.e., we need to ensure that {\em each and every } intermediate measurement does not disturb the state of the nearby qubits or the rest of the quantum computer~\cite{Divincenzo00thephysical}. By contrast, a pebbling attack in our parallel reversible pebbling game naturally corresponds to a quantum circuit which does not require any intermediate measurements and our cost model accounts for the total storage cost (classical + quantum). While Kornerup et al.~\cite{kornerup2021spooky} introduced a spooky pebbling attack on the line graph, we note this spooky pebbling strategy does not yield an efficient reversible pebbling attack in our model as their pebbling attack inherently relies on frequent intermediate measurements to reduce the number of qubits. 

\begin{remark}
One could always try to eliminate the intermediate measurements by applying the ``principle of deferred measurement" \cite{nielsen2002}. However, ``deferred measurement" increases the space and/or depth of a quantum circuit. For example, if the quantum circuit $C$ acts on $s$ qubits and performs $m$ intermediate measurements then we can obtain an equivalent quantum circuit $C'$  with no intermediate measurements with the caveat that $C'$ operates on $s'=s+\poly(m)$ qubits. The space blowup is especially high if $C$ makes many intermediate measurements, e.g., $s=\O{\log m}$. Fefferman and Remscrim~\cite{STOC:FefRem21} gave a space-efficient version of the transform, but their transform yields a large penalty in running time cost, i.e., the transform incurs a multiplicative $\poly(t2^s)$ overhead in the total running time $t$.

If we apply spooky pebbling in the context of Grover's search then the total number of intermediate measurements $m$ would be exponential, i.e., even if we have a quantum circuit $C_f$ evaluating a function $f:\{0,1\}^k\rightarrow\{0,1\}^k$ with just a single intermediate measurement, performing the full Grover's search to find a pre-image of $f$ would involve $m=\O{2^{k/2}}$ intermediate measurements and applying ``deferred measurement'' to the full Grover circuit would incur a massive time (or space) penalty. Thus, finding a quantum circuit $C_f$ which has reduced space-time cost and does not require any intermediate measurements would yield a more compelling quantum pre-image attack.
\end{remark}

%\added{It will be an intriguing open question to define a fine-grained cost model for the spooky pebbling game where the cost of the classical and the quantum storage is different. In this case, the cost model would also need to account for the total number of ghost pebbles placed on the graph. In an ideal world where one could always ensure that the intermediate measurements never disturb the state of other nearby qubits, the cost of placing a ghost pebble would be minimal. On the other hand, if one needed to apply on of the above ``deferred measurement'' transforms to eliminate intermediate measurements then the cost of placing each ghost pebble would be quite large.}

%% file: ec-quantum-pebbling.tex
% !TEX root = tcc-main.tex
\section{Parallel Reversible Pebbling Games}
\seclab{sec:quantum_pebbling}

The biggest difference between the classical and reversible pebbling games occurs when removing pebbles from a pebbling configuration. In a classical setting, we can always delete any pebbles in any point in time when they are no longer needed. On the other hand, in a reversible setting, this is not feasible by quantum no-cloning theorem. Since we can only free a pebble by querying a random oracle at the same input, we can observe that a pebble can be deleted only if we know all of its parents, i.e., all of its parents were previously pebbled. The following definition captures this property:

\begin{definition}[Parallel/Sequential Reversible Graph Pebbling]\deflab{def:quantum-pebbling}
Let $G=(V,E)$ be a DAG and let $T\subseteq V$ be a target set of nodes to be pebbled.  A \emph{pebbling configuration} (of $G$) at round $i$ is a subset $P_i\subseteq V$. Let $P=\left(P_0,\dots, P_t\right)$ be a sequence of pebbling configurations. Below are the following properties which define various aspects of reversible pebblings.
\begin{enumerate}
\item The pebbling should start with no pebbles ($P_0=\emptyset$) and end with pebbles on all of the target nodes i.e., $T \subseteq P_t$. 
\item A pebble can be added only if all of its parents were pebbled at the end of the previous pebbling round, i.e., $\forall i\in[t] : x\in (P_i\setminus P_{i-1}) \Rightarrow \parents(x,G)\subseteq P_{i-1}$.
\item \emph{(Quantum No-Deletion Property)} A pebble can be deleted only if all of its parents were pebbled at the end of the previous pebbling round, i.e., $\forall i\in[t] : x\in (P_{i-1}\setminus P_i) \Rightarrow \parents(x,G)\subseteq P_{i-1}$.
\item \emph{(Quantum Reversibility)} If a pebble was required to generate new pebbles (or remove pebbles), then we must keep the corresponding pebble around, i.e., $\forall i\in[t]:x\in \parents(P_i\setminus P_{i-1},G)\cup\parents(P_{i-1}\setminus P_i,G)\Rightarrow x\in P_i$.
\item \emph{(Remove Excess Pebbles)} We also consider an optional constraint that $P_t = T$. If a pebbling does not satisfy this optional constraint we call it a relaxed pebbling.  
\item (Sequential pebbling only) At most one pebble is added or removed in each round, i.e., $\forall i\in[t] : \left|(P_i \cup P_{i-1}) \setminus (P_i \cap P_{i-1}) \right| \leq 1$.

\end{enumerate}

Now we give pebbling definitions with respect to the above properties. 
\begin{itemize}
    \item A \emph{legal parallel reversible pebbling} of $T$ is a sequence $P=(P_0,\ldots,P_t)$ of pebbling configurations of $G$ where $P_0=\varnothing$ and which satisfies conditions (1), (2), (3), (4) and (5) above. If our pebbling additionally satisfies condition (6) then we say that it is a sequential pebbling. Similarly, if our pebbling does not satisfy condition (5) then we call our pebbling strategy a relaxed pebbling.  
      \item A \emph{legal reversible pebbling sequence} is a sequence of pebbling configurations $\left(P_0,\dots, P_t\right)$ which satisfies properties (2) and (3) and (4) without requiring $P_0 = \{\}$.
    \item A \emph{legal (non-reversible) pebbling sequence} is a sequence of pebbling configurations $\left(P_0,\dots, P_t\right)$ satisfying condition (2).
  
\end{itemize}

We denote with $\qPeb{G,T}$ and $\pqPeb{G,T}$ the set of all legal sequential and parallel reversible pebblings of $G$ with a target set $T$, respectively. We denote with $\rqPeb{G,T}$ and $\rpqPeb{G,T}$ the set of all legal \emph{relaxed} sequential and parallel reversible pebblings of $G$ with target set $T$, respectively. Note that we have $\qPeb{G,T}\subseteq\pqPeb{G,T}$ and $\rqPeb{G,T} \subseteq \rpqPeb{G,T}$. We will mostly be interested in the case where $T=\sinks(G)$ in which case we simply write $\qPeb{G}$ and $\pqPeb{G}$ or $\rqPeb{G}$ and $\rpqPeb{G}$, respectively.

%We further define a \emph{relaxed} version of a (parallel/sequential) reversible pebbling, where we allow extra pebbles in addition to the pebbles on the target nodes at termination of pebbling steps, i.e., we let $P=(P_0,\ldots,P_t)$ be a \emph{legal relaxed (parallel/sequential) reversible pebbling} if the condition (1) above is relaxed to (1') such as:
% \begin{enumerate}
% \item[(1$'$)] The pebbling should end with pebbles on the target nodes, i.e., $P_t\supseteq T$.
% \end{enumerate}

\end{definition}

\begin{remark}\remlab{grover} We first note that from any parallel relaxed reversible pebbling of $G$ we can obtain a quantum circuit $C_{G,H}$ which computes $f_{G,H}$. If our pebbling is not relaxed then the circuit $C_{G,H}$ will map the basis state $|x,y,z\rangle$ to the new state $|x, y \oplus f_{G,H}(x), z \rangle$ with no ancilla bits although this property is not necessary for Grover's search.  Including the requirement that a reversible pebbling eliminates excess pebbles makes it easier to apply the pebbling attack as a recursive subroutine. Thus, in this paper, we will focus on finding non-relaxed reversible pebbling attacks. We also note that the space-time cost of a relaxed/non-relaxed reversible pebbling is not fundamentally different. In particular, if $(P_1,\ldots,P_t)$ is a relaxed pebbling where $P_t = T$ contains the final sink node $N$, then $(P_1,\ldots,P_t,P_{t-1}\cup T,\ldots,P_1\cup T, T)$ is a legal and complete (non-relaxed) reversible pebbling of $G$. The running time increases by a multiplicative factor of $2$ and the space increases by an additive factor of $|T| \leq |P_t|$ where $T$ is the target set. In particular, the overall space-time costs increase by a multiplicative factor of $4$ {\em at most}. In the remainder of the paper, when we write ``legal reversible pebbling" we assume that the pebbling is parallel and non-relaxed by default. 
\end{remark}

\begin{definition}[Reversible Pebbling Complexity]\deflab{def:complexity}
Given a DAG $G=(V,E)$, we essentially use the same definitions for the reversible pebbling complexity as defined in the previous literature \cite{STOC:AlwSer15,EC:AlwBloPie17,EC:AlwBloPie18}. That is, the standard notion of \emph{time, space, space-time} and \emph{cumulative pebbling complexity (CC)} of a \emph{reversible} pebbling $P=\{P_0,\ldots,P_t\}\in\pqPeb{G}$ are also defined to be:
\begin{itemize}
\item (time complexity) $\pqpeb_\Ctime(P) = t$,
\item (space complexity) $\pqpeb_\Cspace(P) = \max_{i\in[t]}|P_i|$,
\item (space-time complexity) $\pqpeb_\Cspacetime(P)=\pqpeb_\Ctime(P)\cdot\pqpeb_\Cspace(P)$, and
\item (cumulative pebbling complexity) $\pqpeb_\Ccc(P)=\sum_{i\in[t]}|P_i|$.
\end{itemize}
For $\alpha\in\Cevery$ and a target set $T\subseteq V$, the \emph{parallel} reversible pebbling complexities of $G$ are defined as
\[ \pqpeb_\alpha(G,T)=\min_{P\in\pqPeb{G,T}}\pqpeb_\alpha(P). \]
When $T=\sinks(G)$ we simplify notation and write $\pqpeb_\alpha(G)$.

We define the \emph{time, space, space-time} and \emph{cumulative pebbling complexity} of a \emph{sequential} reversible pebbling $P=\{P_0,\ldots,P_t\}\in\qPeb{G}$ in a similar manner: $\qpeb_\Ctime(P) = t$, $\qpeb_\Cspace(P) = \max_{i\in[t]}|P_i|$, $\qpeb_\Cspacetime(P)=\qpeb_\Ctime(P)\cdot\qpeb_\Cspace(P)$, and $\qpeb_\Ccc(P)=\sum_{i\in[t]}|P_i|$. Similarly, for $\alpha\in\Cevery$ and a target set $T\subseteq V$, the \emph{sequential} reversible pebbling complexities of $G$ are defined as $\qpeb_\alpha(G,T)=\min_{P\in\qPeb{G,T}}\qpeb_\alpha(P)$. 
When $T=\sinks(G)$ we simplify notation as well and write $\qpeb_\alpha(G)$.
\end{definition}

When compared to the definition of a \emph{classical} pebbling, we can observe that a reversible pebbling has more restrictions, i.e., it only allows us to have pebbles exactly on the target nodes at the end of the pebbling steps, and it further requires quantum no-deletion property and quantum reversibility. % which means that we can delete pebbles only if all of its parents were pebbled at the previous round. 
This implies that any legal reversible pebblings are also legal classical pebblings, i.e.,  $\pPeb_{G,T} \subseteq \pqPeb{G,T}$ (resp. $\Peb_{G,T} \subseteq \qPeb{G,T}$). This implies that for any graph $G$, target set $T$ and cost metric $\alpha\in\{s,t,st,cc\}$, we have $\ppeb_\alpha(G,T)\leq \pqpeb_\alpha(G,T)$ (resp. $\peb_\alpha(G,T)\leq \qpeb_\alpha(G,T)$) for a DAG $G=(V,E)$ and a target set $T\subseteq V$, where $\ppeb_\alpha(G,T)$ (resp. $\peb_\alpha(G,T)$) denotes the parallel (resp. sequential) classical pebbling complexities which are defined essentially the same as in \defref{def:complexity} with a \emph{classical} pebbling $P=\{P_0,\ldots,P_t\}\in\pPeb_G$ (resp. $\Peb_G$). This means that any lower bound on the classical pebbling complexity of a graph $G$ immediately carries over to the reversible setting and an upper bound (attack) on the reversible pebbling cost immediately carries over to the setting classical pebbling.

In the context of quantum pre-image attacks, parallel space-time costs are arguably the most relevant metric. In particular, the depth of the full Grover circuit scales with the number of queries to our quantum circuit $C_{G,H}$ for $f_{G,H}$ multiplied by the number of pebbling rounds for $G$. Similarly, the width of the full Grover circuit will essentially be given by the space usage of our pebbling. Thus, the space-time of Grover's algorithm will scale directly with $\pqpeb_\Cspace(P)$. The cumulative pebbling complexity would still be relevant in settings where we are running multiple instances of Grover's algorithm in parallel and can amortize space usage over multiple inputs. In this paper, we primarily focus on analyzing reversible space-time costs, as this would likely be the most relevant metric in practice. However, cumulative pebbling complexity still can be worthwhile to study and we provide some initial results in this direction. 

%% file: ec-attack.tex
% !TEX root = tcc-main.tex
\allowdisplaybreaks
\section{Reversible Pebbling Attacks and Applications on iMHFs}
\seclab{sec:attack}

\subsection{Warmup: Parallel Reversible Pebbling Attack on a Line Graph}
\seclab{sec:attack:linegraph}

We first consider two widely deployed hash functions, PBKDF2 \cite{rfc2898} and BCRYPT \cite{USENIX:ProMaz99}, as motivating examples for analyzing a line graph. Basically, they are constructed by hash iterations so they can be modeled as a line graph when simplified.  Hence, the pebbling analysis of a line graph tells us about the costs of PBKDF2 and BCRYPT. Although there has been some effort to replace such password-hash functions with memory-hard functions such as Argon2 or SCRYPT~\cite{SP:BloHarZho18}, PBKDF2 and BCRYPT are still commonly used by a number of organizations.  Thus, it is still important to understand the costs of an offline brute-force attack on passwords protected by functions like PBKDF2 and BCRYPT. In fact, NIST recommends using memory-hard functions for password hashing \cite{NIST17} but they still allow PBKDF2 and BCRYPT when used with long enough hash iterations. 
Hence, there is still value to analyze the quantum resistance of these functions. Our reversible pebbling attack on DRSample relies on efficient pebbling strategies for line graphs as a subroutine providing further motivation to understand the reversible pebbling costs of a line graph. 

As we illustrated in \secref{sec:overview}, we give a (sequential/parallel) reversible pebbling strategy for a line graph $L_N$ using recursion. It can be done by recursively define the sequence of consecutive locations $I(k)$ as $I(k) = I(k-1)'\circ I(k-2)'\circ \ldots\circ I(0)'$ for $k>0$ and $I(0)=\{\}$, where for $0\leq j<k$, $I(j)'$ is defined to be a concatenation of $c$ copies of $I(j)$ and $i_j$ (which is an incident node to $I(j)$), i.e., $I(j)' := I(j)^{(1)}\circ i_j^{(1)}\circ I(j)^{(2)}\circ i_j^{(2)}\circ\ldots\circ I(j)^{(c)}\circ i_j^{(c)},$ where $A^{(\ell)}$ denotes the $\ell\th$ copy of $A$. Intuitively, we can sequentially pebble $I(k)$ by pebbling $I(k-1)',I(k-2)',\ldots,I(0)'$. Here, pebbling $I(j)'$ means that we pebble $I(j)^{(\ell)}$, $i_j^{(\ell)}$, and unpebble $I(j)^{(\ell)}$, and we move on to the next copy to pebble $I(j)^{(\ell+1)}$. We can parallelize this strategy by removing and adding pebbles on the consecutive copies at the same time, which requires more space usage but saves time. Here, we only state the space-time cost of our reversible pebbling strategy on a line graph in \thmref{thm:line}. Details of our pebbling strategy can be found in \FullVersion{\appref{app:line}}{the full version}.

\begin{theorem}\thmlab{thm:line}
Let $L_N$ be a line graph of size $N$. Then we have $\qpeb_\Cspacetime(L_N)=\O{N^{1+(2+o(1))\frac{1}{\sqrt{\log N}}}}$ and $\pqpeb_\Cspacetime(L_N)=\O{N^{1+\frac{2}{\sqrt{\log N}}}}$.
\end{theorem}
\FullVersion{\begin{proof}
The proof directly comes from \lemref{lem:line} in \appref{app:line}.\qed
\end{proof}}{The proof of \thmref{thm:line} can be found in the full version.}

\subsection{Reversible Pebbling Attacks on $(e,d)$-reducible DAGs}
\seclab{sec:attack:reducible}

In this section, we introduce another type of reversible pebbling attack on $(e,d)$-reducible DAGs with depth-reducing sets with $d$ very small. \FullVersion{Recall that a DAG $G=(V,E)$ is $(e,d)$-reducible if there exists a subset $S\subseteq V$ with $|S|\leq e$ such that the subgraph $G-S$ does not contain a path of length $d$. Here, we call such subset $S$ a \emph{depth-reducing set}. }{}In this paper, we only consider DAGs with constant indegree, and especially the current state-of-the-art constructions of iMHFs have indegree $2$. Therefore, we will assume that $\indeg(G)=2$ for the DAGs that we consider.

Since the graph has indegree $2$, if we find a depth-reducing set $S$ such that $G-S$ has depth $d$, then we observe that $|\ancestors(v,G-S)|\leq 2^d$ for any node $v$ in $G-S$. If $d$ is small, i.e., $d\ll\log N$, then $2^d\ll N$ and we can expect that the space-time cost for pebbling such $(e,d)$-reducible DAG becomes $o(N^2)$. More precisely, we start with giving a regular pebbling strategy (without quantum restrictions) for such DAGs.

\paragraph{Classical Black Pebbling Strategy.} 
We begin by giving a classical pebbling strategy with small space-time complexity. Note that prior pebbling strategies focused exclusively on minimizing cumulative pebbling cost, but the pebbling attacks of Alwen and Blocki \cite{C:AlwBlo16}\footnote{If $G$ is $(e,d)$-reducible then Alwen and Blocki \cite{C:AlwBlo16} showed that $\displaystyle\ppeb_{cc}(G)\leq\min_{g\geq d}\left(eN+gN\cdot\indeg(G)+\frac{N^2d}{g}\right)\allowbreak =o(N^2)$.} for $(e,d)$-reducible graphs still have the space-time cost $\Omega(N^2)$. 

We first introduce the following helpful notation. For nodes $x$ and $y$ in a DAG $G=(V,E)$, let $\lpath_G(x,y)$ denote the number of nodes in the longest path from $x$ to $y$ in $G$. Then for a node $w\in V$, a depth-reducing set $S\subseteq V$, and a positive integer $i\in\mathbb{Z}_{> 0}$, we first define a set $A_{w,S,i}$ which consists of the nodes $v$ where the longest directed path from $v$ to $w$ in $G-S_{\leq w-1}$ has length $i$, i.e., it contains exactly $i$ nodes.
\[ A_{w,S,i} \coloneqq \left\{v:\lpath_{G-S_{\leq w-1}}(v,w)= i\right\}.\]
It is trivial by definition that for any $v\in V$, $A_{v,S,1}=\{v\}$.

Let $G=(V=[N],E)$ be an $(e,d)$-reducible DAG. We observe that $\depth(G_{\leq k}-S_{\leq k})\leq d$ is still true for any $k\leq N$. At round $k$, we have always ensured that we have pebbles on the set $S_{\leq k}$ and on $\{k\}$ itself. Further, at round $k$, we can look $d$ steps into the future so that at round $k+d$ we can pebble node $k+d$ without delay. Hence, we start to repebble $\ancestors(k+d,G-S)$ in this round and because $\depth(G_{\leq k}-S_{\leq k})\leq d$ we are guaranteed to finish within $d$ rounds --- just in time to pebble node $k+d$. Taken together, in round $k$, we have pebbles on $\{k\}$, $S_{\leq k}$, and $\ancestors(k+i,G-S)$ for all $i\leq d$. More precisely, for $v\in V$, let $P_v = S_{\leq v} \cup \left(\bigcup_{j=1}^d\bigcup_{i=j}^d A_{v-1+j,S,i}\right)$. Since each ancestor graph has size at most $2^d$ and there are at most $d$ of them, we observe that the total number of pebbles in each round is at most $1+|S_{\leq k}|+\sum_{i=1}^d |\ancestors(k+i,G-S)| \leq 1+e+d2^d$. Hence, we have that $\ppeb_\Cspacetime(G) \leq N(1+e+d2^d)$.

\paragraph{Reversible Pebbling Strategy.} While the above strategy works in the classical setting it will need to be tweaked to obtain a legal reversible pebbling. In particular, after node $k+d$ is pebbled we cannot immediately remove pebbles from all nodes in $\ancestors(k+d,G-S)$ because this would violate our quantum reversibility property. Instead, we can reverse the process and unpebble nodes in  $\ancestors(k+d,G-S)$ over the next $G-S$ rounds --- with the possible exception of nodes $v \in  \ancestors(k+d,G-S)$ which are part of $ \ancestors(k+d+j,G-S)$ and are still required for some future node $k+d+j$. Thus, if a DAG $G$ is $(e,d)$-reducible we can establish the following result.

\begin{theorem}\thmlab{thm:ed-reducible}
Let $G=(V=[N],E)$ be an $(e,d)$-reducible DAG. Then $\pqpeb_\Cspacetime(G)=\O{Ne+Nd2^d}$.
\end{theorem}

We will give the proof of \thmref{thm:ed-reducible} later in the subsection. To prove \thmref{thm:ed-reducible}, we first would need to give a legal reversible pebbling for an $(e,d)$-reducible DAG $G$. \lemref{lem:edlegal} provides the desired reversible pebbling for $G$. \FullVersion{}{The proof of \lemref{lem:edlegal} can be found in the full version.}

\begin{lemma}\lemlab{lem:edlegal}
Let $G=(V=[N],E)$ be an $(e,d)$-reducible DAG and let $S\subseteq V$ be a depth-reducing set. Define
\[ B_v := \bigcup_{j=1}^{d+1}\bigcup_{i=j}^{d+1} \left( A_{v+1-j,S,i}\cup A_{v-1+j,S,i} \right),\]
for $v\in V$. Then $P=(P_0,P_1,\ldots,P_{2N})$, where each pebbling configuration is defined by
\begin{itemize}
\item $P_0=\varnothing$,
\item for $v\in[N]$, $P_v := S_{\leq v} \cup B_v$, and
\item for $N < v \leq 2N$, $P_v := P_{2N-v}\cup\{N\}$,
\end{itemize}
is a legal parallel reversible pebbling for $G$.
\end{lemma}
\FullVersion{Before proving \lemref{lem:edlegal}, we observe the following key claim. The proof of \claimref{claim:Pv} can be found in \appref{missing}.

\newcommand{\claimPv}{For $v\in[N]$, $\parents(P_v\setminus P_{v-1},G)\cup\parents(P_{v-1}\setminus P_v,G)\subseteq P_{v-1}\cap P_v$.}
\begin{claim}\claimlab{claim:Pv}
\claimPv
\end{claim}

\begin{proofof}{\lemref{lem:edlegal}}
We want to show that it satisfies conditions in \defref{def:quantum-pebbling}.\\[5pt]
\noindent \ul{Conditions (1) and (5): $P_{2N}=\{N\}$.}
\begin{itemize}
\item It is clear that $P_{2N} = P_0\cup\{N\}=\{N\}$ which is the only target node of the pebbling game.
\end{itemize}

\noindent \ul{Condition (2): $\forall v\in[2N] : x\in (P_v\setminus P_{v-1}) \Rightarrow \parents(x,G)\subseteq P_{v-1}$.}
\begin{itemize}
\item If $v\in [N]$, by \claimref{claim:Pv}, we have $\parents(P_v\setminus P_{v-1})\subseteq P_{v-1}\cap P_v \subseteq P_{v-1}$.
\item If $N<v\leq 2N$, we have $P_v\setminus P_{v-1} = (P_{2N-v}\cup\{N\})\setminus(P_{2N-v+1}\cup\{N\}) = P_{2N-v}\setminus P_{2N-v+1}$. Let $w=2N-v+1$, then we have that $w\in[N]$ and $P_v\setminus P_{v-1} = P_{w-1}\setminus P_w$. Now we want to show that $\parents(P_{w-1}\setminus P_w,G)\subseteq P_{v-1} = P_w\cup\{N\}$, which also holds by \claimref{claim:Pv}.
\end{itemize}

\noindent \ul{Condition (3): $\forall v\in[2N] : x\in (P_{v-1}\setminus P_v) \Rightarrow \parents(x,G)\subseteq P_{v-1}$.}
\begin{itemize}
\item If $v\in [N]$, by \claimref{claim:Pv}, we have $\parents(P_{v-1}\setminus P_v)\subseteq P_{v-1}\cap P_v \subseteq P_{v-1}$.
\item If $N<v\leq 2N$, we have $P_{v-1}\setminus P_v = (P_{2N-v+1}\cup\{N\})\setminus(P_{2N-v}\cup\{N\}) = P_{2N-v+1}\setminus P_{2N-v}$. Then similarly, letting $w=2N-v+1$, we have that $w\in[N]$ and $P_{v-1}\setminus P_v = P_w\setminus P_{w-1}$. Now we want to show that $\parents(P_w\setminus P_{w-1},G)\subseteq P_{v-1} = P_w\cup\{N\}$, which also holds by \claimref{claim:Pv}.
\end{itemize}

\noindent \ul{Condition (4): $\forall v\in[2N]:x\in \parents(P_v\setminus P_{v-1},G)\cup\parents(P_{v-1}\setminus P_v,G)\Rightarrow x\in P_v$.}
\begin{itemize}
\item If $v\in[N]$, this is clear from \claimref{claim:Pv} since $\parents(P_v\setminus P_{v-1},G)\cup\parents(P_{v-1}\setminus P_v,G) \subseteq P_{v-1}\cap P_v \subseteq P_v$.
\item If $N<v\leq 2N$, by similar argument from above, by letting $w=2N-v+1$, we have that $w\in[N]$ and $\parents(P_v\setminus P_{v-1},G)\cup\parents(P_{v-1}\setminus P_v,G) = \parents(P_{w-1}\setminus P_w,G)\cup\parents(P_w\setminus P_{w-1},G)\subseteq P_{w-1}\cap P_w \subseteq P_{w-1} \subseteq P_{w-1}\cup\{N\}=P_v$.
\end{itemize}
Taken together, we can conclude that for any $v\in [2N]$, $P_v$ is a legal reversible pebbling configuration for $G$.
\end{proofof}}{}

Now we are ready to prove \thmref{thm:ed-reducible}.

\begin{proofof}{\thmref{thm:ed-reducible}}
Let $P=\{P_0,P_1,\ldots,P_{2N}\}$ as defined in \lemref{lem:edlegal}, in which we showed that it is a legal quantum pebbling. Clearly, $\pqpeb_\Ctime(P)=2N$. Further, we observe that $\pqpeb_\Cspace(P)\leq\max_{v\in V}\{|S_{\leq v}|+|B_v|+1\}$. Since we assume that $\indeg(G)=2$, we have
\begin{align*}
|B_v| &= \left|\bigcup_{j=1}^{d+1}\bigcup_{i=j}^{d+1} \left( A_{v+1-j,S,i}\cup A_{v-1+j,S,i} \right)\right| \\
&\leq \sum_{j=1}^{d+1}\sum_{i=j}^{d+1} | A_{v+1-j,S,i}|+ |A_{v-1+j,S,i}|\\
&\leq \sum_{j=1}^{d+1}\sum_{i=j}^{d+1} 2^{i+1} = 8d2^d+2.
\end{align*}
Taken together, \FullVersion{$\pqpeb_\Cspacetime(P)=\pqpeb_\Ctime(P)\pqpeb_\Cspace(P)\leq 2N(e+8d2^d+3)=\O{Ne+Nd2^d}$}{$\pqpeb_\Cspacetime(P)\leq 2N(e+8d2^d+3)=\O{Ne+Nd2^d}$}. Hence, \FullVersion{we can conclude that }{}$\pqpeb_\Cspacetime(G)=\min_{P\in\pqPeb{G,\{N\}}}\pqpeb_\Cspacetime(P)=\O{Ne+Nd2^d}$.
\end{proofof}

\paragraph{Analysis of Argon2i.}
There are a number of variants for the Argon2i graphs. We will focus on Argon2i-A \cite{BDK15,AC:BonCorSch16} and Argon2i-B\footnote{We will follow the naming convention of Alwen and Blocki~\cite{ESP:AlwBlo17} throughout the paper and use Argon2i-A to refer to Argon2i-A v1.1 and Argon2i-B to refer to v1.2+.} \cite{BDKJ16} here. Recall that Argon2i-A is a graph $G=(V=[N],E)$, where $E=\left\{(i,i+1):i\in[N-1]\right\}\cup\{(r(i),i)\}$, where $r(i)$ is a random value that is picked uniformly at random from $[i-2]$. Argon2i-B has the same structure, except that $r(i)$ is not picked uniformly at random but has a distribution as follows:
\[ \Pr\left[r(i)=j\right]=\Pr_{x\in[N]}\left[ i\left(1-\frac{x^2}{N^2}\right)\in (j-1,j] \right]. \]

\newcommand{\argonlemma}{Let $G_\ArgonA=(V_A=[N],E_A)$ and $G_\ArgonB=(V_B=[N],E_B)$ be randomly sampled graphs according to the Argon2i-A and Argon2i-B edge distributions, respectively. Then with high probability, the following holds:
\begin{enumerate}
\item $G_\ArgonA$ is $(e_1,d_1)$-reducible for $e_1=\frac{N}{d'}+\frac{N\ln\lambda}{\lambda}$ and $d_1=d'\lambda$, for any $0<\lambda<N$ and $0<d'<\frac{N}{\lambda}$.
\item $G_\ArgonB$ is $(e_2,d_2)$-reducible for $e_2=\frac{N}{d'}+\frac{2N}{\sqrt{\lambda}}$ and $d_2=d'\lambda$, for any $0<\lambda<N$ and $0<d'<\frac{N}{\lambda}$.
\end{enumerate}}
\begin{lemma}\lemlab{lem:argon2i}
\argonlemma
\end{lemma}

Alwen and Blocki~\cite{C:AlwBlo16,ESP:AlwBlo17}  established similar bounds to \lemref{lem:argon2i}, but focused on parameter settings where the depth $d$ is large. By contrast, we will need to pick a depth-reducing set with a smaller depth parameter $d \ll \log N$ to minimize the $d2^d$ cost term in our pebbling attack. The full proof of \lemref{lem:argon2i} can be found in \FullVersion{\appref{missing}}{the full version}. Here, we only give a brief intuition of the proof. To reduce the depth of a graph, we follow the approach of Alwen and Blocki~\cite{C:AlwBlo16,ESP:AlwBlo17} and divide $N$ nodes into $\lambda$ layers of size $N/\lambda$ and then reduce the depth of each layer to $d'$ so that the final depth becomes $d=d'\lambda$. To do so, we delete all nodes with parents in the same layer, and then delete one out of $d'$ nodes in each layer. And then we count the number of nodes to be deleted in both steps for each graph.

Applying the result from \lemref{lem:argon2i} to \thmref{thm:ed-reducible}, we have the following space-time cost of reversible pebbling for Argon2i-A and Argon2i-B. Intuitively, we obtain \corref{cor:argon2i} by setting $\lambda= \sqrt{\log N}$ and $d' = \lambda/\ln \lambda \approx 2 \sqrt{\log N}/\log \log N$ (resp. $\lambda= \sqrt[3]{\log^2 N}$ and $d'=\sqrt[3]{\log N}/2$) in \lemref{lem:argon2i} for Argon2i-A (resp. Argon2i-B).  The full proof of \corref{cor:argon2i} can be found in \FullVersion{\appref{missing}}{the full version}.

\newcommand{\argoncorollary}{Let $G_\ArgonA=(V_A=[N],E_A)$ and $G_\ArgonB=(V_B=[N],E_B)$ be randomly sampled graphs according to the Argon2i-A and Argon2i-B edge distributions, respectively. Then with high probability, $\pqpeb_\Cspacetime(G_\ArgonA)=\O{\frac{N^2\log\log N}{\sqrt{\log N}}}$, and $\pqpeb_\Cspacetime(G_\ArgonB)=\O{\frac{N^2}{\sqrt[3]{\log N}}}$.}
\begin{corollary}\corlab{cor:argon2i}
\argoncorollary
\end{corollary}

\FullVersion{
\begin{remark}
Our reversible pebbling attacks on Argon2i-A and Argon2i-B have space-time cost $\pqpeb_\Cspacetime(G_\ArgonA)=\O{\frac{N^2\log\log N}{\sqrt{\log N}}}$, and $\pqpeb_\Cspacetime(G_\ArgonB)=\O{\frac{N^2}{\sqrt[3]{\log N}}}$ respectively. In the classical setting it was known that $\ppeb_{\Ccc}(G_{\ArgonA}) = \tilde{\mathcal{O}}(N^{1.708})$ and  $\ppeb_{\Ccc}(G_{\ArgonB}) = \tilde{\mathcal{O}}({N^{1.768} })$ \cite{EC:AlwBloPie17,TCC:BloZho17}. While these pebbling attacks achieve more impressive reductions in cost, we stress that the attacks are (1) non-quantum (i.e., non-reversible) and (2) the space-time complexity of these classical pebbling attacks is still $\Omega(N^2)$ since there will be a few pebbling rounds with $\Omega(N)$ pebbles on the graph. We remark that since any reversible pebbling is a legal classical pebbling that  it immediately follows that $\ppeb_\Cspacetime(G_\ArgonA)=\O{\frac{N^2\log\log N}{\sqrt{\log N}}}$, and $\ppeb_\Cspacetime(G_\ArgonB)=\O{\frac{N^2}{\sqrt[3]{\log N}}}$. The best known classical lower bounds for Argon2i-A and Argon2i-B are $\ppeb_{\Ccc}(G_{\ArgonA}) = \Omega({N^{1.66}})$ and  $\ppeb_{\Ccc}(G_{\ArgonB}) = \tilde{\Omega}({N^{1.75} })$  which immediately implies that $\pqpeb_\Cspacetime(G_\ArgonA)= \Omega({N^{1.66}})$, and $\pqpeb_\Cspacetime(G_\ArgonB)= \tilde{\Omega}({N^{1.75} })$. Thus, there remains a gap between the best upper/lower bounds for $\pqpeb_\Cspacetime(G_\ArgonA)$ and $\pqpeb_\Cspacetime(G_\ArgonB)$. Closing or tightening this gap is an interesting open research challenge. Similarly, it would be interesting to figure out if we can find better reversible pebbling strategies to reduce the cumulative cost of Argon2i-A and Argon2i-B, e.g., can one adapt the classical pebbling strategy of the previous work~\cite{C:AlwBlo16,EC:AlwBloPie17,TCC:BloZho17} to the reversible setting. 
\end{remark}
}{}

\subsection{Reversible Pebbling Attacks using an Induced Line Graph}\seclab{sec:attack:reducedgraph}

In this section, we give another general strategy to pebble DAGs by ``reducing" the DAG $G$ to a line graph, as shown in \figref{fig:metagraph}. Intuitively, given a DAG $G=(V,E)$ with $|V|=N$ and an integer parameter $b\geq 1$, we can partition $V$ into consecutive blocks $B_1,\ldots,B_{\lceil N/b\rceil}$ such that each block contains exactly $b$ nodes, while for the last block we can have less than $b$ nodes if $N/b$ is not an integer.

\begin{figure}[ht!]
\centering
\begin{tikzpicture}
[every node/.style={node distance=2cm},
unpeb/.style={circle,draw,inner sep=0pt,minimum width=0.2cm},
peb/.style={circle,draw,fill=black,inner sep=0pt,minimum width=0.2cm}]
% nodes in a DAG G
\node (G) at (-2,0) {$G$};
\node (cdot1) at (-1,0) {$\cdots$};
\node[unpeb] (1) at (0,0) {};
\node[unpeb] (2) at ($(1) + (0.5,0)$) {};
\node[peb] (3) at ($(2) + (0.5,0)$) {};
\node[peb] (4) at ($(3) + (0.5,0)$) {};
\node[unpeb] (5) at ($(4) + (0.5,0)$) {};
\node[peb] (6) at ($(5) + (1.5,0)$) {};
\node[peb] (7) at ($(6) + (0.5,0)$) {};
\node[peb] (8) at ($(7) + (0.5,0)$) {};
\node[peb] (9) at ($(8) + (0.5,0)$) {};
\node[peb] (10) at ($(9) + (0.5,0)$) {};
\node[peb] (11) at ($(10) + (1.5,0)$) {};
\node[peb] (12) at ($(11) + (0.5,0)$) {};
\node[peb] (13) at ($(12) + (0.5,0)$) {};
\node[peb] (14) at ($(13) + (0.5,0)$) {};
\node[peb] (15) at ($(14) + (0.5,0)$) {};
\node (cdot2) at ($(15) + (1,0)$) {$\cdots$};
\foreach \x in {3,8,13}{
\node[dashed,ellipse,draw,minimum width=2.5cm, minimum height=0.8cm] at (\x) {};}
\node at ($(3) + (0,-0.7)$) {$B_{i-1}$};
\node at ($(8) + (0,-0.7)$) {$B_i$};
\node at ($(13) + (0,-0.7)$) {$B_{i+1}$};
% edges in G
\path[draw,->] (cdot1) edge (1);
\path[draw,->] (1) edge (2);
\path[draw,->] (2) edge (3);
\path[draw,->] (3) edge (4);
\path[draw,->] (4) edge (5);
\path[draw,->] (5) edge (6);
\path[draw,->] (6) edge (7);
\path[draw,->] (7) edge (8);
\path[draw,->] (8) edge (9);
\path[draw,->] (9) edge (10);
\path[draw,->] (10) edge (11);
\path[draw,->] (11) edge (12);
\path[draw,->] (12) edge (13);
\path[draw,->] (13) edge (14);
\path[draw,->] (14) edge (15);
\path[draw,->] (15) edge (cdot2);
\path[draw,->] (3) edge[bend left=20] (11);
\path[draw,->] (4) edge[bend right] (7);
\path[draw,->] (4) edge[bend right] (14);
\path[draw,->] (9) edge[bend left] (13);
% nodes in G_m
\node (Gm) at (-2,-2) {$L_{\lceil N/b\rceil}$};
\node (cdot3) at (-1,-2) {$\cdots$};
\node[unpeb,minimum width=0.8cm] (m1) at ($(3) + (0,-2)$) {$v'_{i-1}$};
\node[peb,minimum width=0.8cm] (m2) at ($(8) + (0,-2)$) {$\textcolor{white}{v'_i}$};
\node[peb,minimum width=0.8cm] (m3) at ($(13) + (0,-2)$) {$\textcolor{white}{v'_{i+1}}$};
\node (cdot4) at ($(cdot2) + (0,-2)$) {$\cdots$};
% edges in G_m
\path[draw,->] (cdot3) edge (m1);
\path[draw,->] (m1) edge (m2);
\path[draw,->] (m2) edge (m3);
\path[draw,->] (m3) edge (cdot4);
\end{tikzpicture}
\caption{A line graph $L_{\lceil N/b\rceil}$ induced from a DAG $G$. Note that each block in an original graph corresponds to a node in the corresponding line graph, e.g., a block $B_i$ in $G$ that consists of five nodes correspond to the node $v'_i$ in $L_{\lceil N/b\rceil}$.}
\figlab{fig:metagraph}
\end{figure}
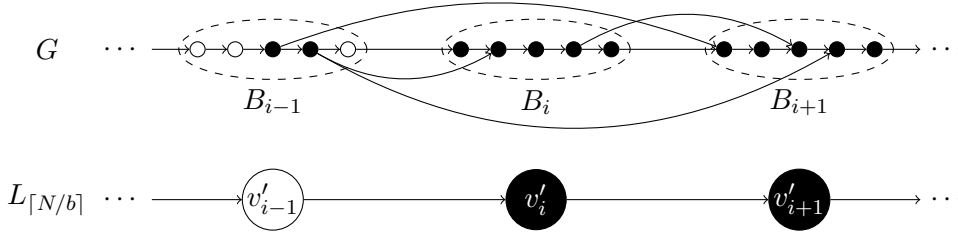

\paragraph{Notation.} 
Now we consider a reversible pebbling $P'$ of the line graph $L_{\lceil N/b\rceil}=(V'=[\lceil N/b \rceil],E')$. 
Intuitively, each node in $L_{\lceil N/b\rceil}$ corresponds to each block in $G$. 
To transform $P'$ into a pebbling $P$ of $G$, it will be useful to introduce some notation. 
Given a node $v' \in V'$ and the pebbling $P'$ of $L_{\lceil N/b\rceil}$, we define $\lastd(P',v')\coloneqq \max \left\{i :v'\in P_i'\right\}$ to denote the unique index $i$ such that node $v' \in P'_i$, but $v' \not \in P'_j$ for all rounds $j >i$, i.e., the pebble on node $v'$ was removed for the final time in round $i+1$. 
Similarly, it will be convenient to define $\lasta(P')\coloneqq \max \left\{ i :  \lceil N/b\rceil    \not\in P'_{i-1} \right\}$ to be the unique round where a pebble was placed on the last node $v=\lceil N/b\rceil$ for the final time (Note: it is possible that a legal pebbling $P'$ places/removes a pebble on node $v=\lceil N/b\rceil$ several times). 
We make a couple of basic observations. 
First, we note that if $u' < v'$ then $\lastd(P',u') > \lastd(P',v')$ since we need node $v'-1$ on the graph to remove a pebble from node $v'$. 
Similarly, we note that for any node $v' < \lceil N/b\rceil$ that $\lastd(P',v') > \lasta(P')$ since we need node $ \lceil N/b\rceil-1$ to be pebbled before we can place a pebble on the final node. 
Given our graph $G=(V,E)$, a parameter $b$, and a partition $B_1,\ldots,B_{\lceil N/b \rceil}$ of $V$ into consecutive blocks of size $b$, we define $\skipnode(B_i,G)$, for each $i$, to be the set of all skip nodes in block $B_i$, i.e., the set of nodes with an outgoing edge that skips over block $B_{i+1}$:
\begin{equation}\eqnlab{skipnodes}
\skipnode(B_i,G)\coloneqq\{v\in B_i:\exists j>i+1\text{ such that }v\in\parents(B_j,G)\}.
\end{equation}
We further define $\numskip(G,b)$ as the total number of skip nodes in $G=(V,E)$ after partitioning the set of nodes $V$ into consecutive blocks of size $b$, i.e., $\numskip(G,b)\coloneqq \sum_{i=1}^{\lceil N/b\rceil} |\skipnode(B_i,G)|$, where $B_i$'s are defined as before.

\paragraph{Pebbling Attempt 1.} Our first approach to convert $P'\in\pqPeb{L_{\lceil N/b\rceil}}$ to a legal reversible pebbling $P$ of $G$ is as follows. Since each node in $L_{\lceil N/b\rceil}$ corresponds to a block (of size at most $b$) in $G$, we can transform placing a pebble on a node in $L_{\lceil N/b\rceil}$ to pebbling all nodes in the corresponding block in $G$ in at most $b$ steps. Similarly, we can convert removing a pebble on a node in $L_{\lceil N/b\rceil}$ to removing pebbles from all nodes in the corresponding block in $G$ in at most $b$ steps. It gives us $\pqpeb_\Cspace(P)\leq b\pqpeb_\Cspace(P')$ since each node is transformed to a block of size at most $b$, and $\pqpeb_\Ctime(P)\leq b\pqpeb_\Ctime(P')$ since one pebbling/removing step in $L_{\lceil N/b\rceil}$ is transformed to at most $b$ pebbling/removing steps in $G$.

However, this transformation does \emph{not} yield a legal reversible pebbling of $G$ due to the skip nodes. In particular, given a reversible pebbling configuration $P'_k=\{v'\}$ of $L_{\lceil N/b\rceil}$, it is legal to proceed as $P'_{k+1}=\{v',v'+1\}$. However, when converting it to a reversible pebbling of $G$, one would need to place pebbles on block $B_{v'+1}$ while only having pebbles on block $B_{v'}$. This could be illegal if there is a node $v\in V$ such that $v\in B_i$ for $i<v'$ and $v\in\parents(B_{v'+1},G)$, i.e., $v$ is a skip node in $B_i$, because $v$ must be previously pebbled to place pebbles on block $B_{v'+1}$.

\paragraph{Reversible Pebbling Strategy.}
To overcome this barrier, when we convert $P'\in\pqPeb{L_{\lceil N/b\rceil}}$ to a legal reversible pebbling $P$ of $G$, we define a transformation $P=\trans(G,P',b)$ which convert placing/removing a pebble on/from a node $v'$ in $L_{\lceil N/b\rceil}$ to placing/removing pebbles on/from all nodes in the corresponding block $B_{v'}$ in $G$ in at most $b$ steps as our first attempt, but when we remove pebbles from $B_{v'}$ in $G$, we keep skip nodes for the block in the transformation until we delete pebbles from the block for the last time, i.e., after round $\lastd(P',v')$, since these skip nodes will no longer needed to pebble nodes in other blocks in the future. 

Furthermore, for the last block (in $G$), when a pebble is placed on the last node (in $L_{\lceil N/b\rceil}$) for the final time, i.e., in round $\lasta(P')$, we indeed want to only pebble the last node (sink node) in the block but not the entire block. Hence, we need additional (at most $b-1$) steps to remove pebbles from all nodes except for the last node in the block.

We can argue the legality of the converted pebbling of $G$ because pebbling steps in each block is legal and keeping skip nodes during the transformation does not affect the legality of pebbling. Intuitively, whenever we pebble a new node $v$ in  $L_{\lceil N/b\rceil}$ the node $v-1$ must have been pebbled in the previous round. Thus, in $G$ we will have pebbles on all nodes in the block $B_{v-1}$. Now for every node $w \in B_{v}$ and every edge of the form $(u,w)$  we either have (1) $u \in B_{v-1}$, (2) $u \in B_{v}$ or (3) $u \in B_{j}$ with $j < v-1$. In the third case, $u$ is a skip node and will already be pebbled allowing us to legally place a pebble on node $w$. Similarly, in the first case, we are guaranteed that $u$ is already pebbled before we begin pebbling nodes in block $B_v$ since every node in $B_{v-1}$ is pebbled, and in the second case, $u$ will be (re)pebbled before node $w$. A similar argument shows that all deletions are legal as well. The full proof of \lemref{lem:translegal} can be found in \FullVersion{\appref{missing}}{the full version}.

\newcommand{\translegal}{Let $G=(V=[N],E)$ and $b\in[N]$ be a parameter. If $P'\in\pqPeb{L_{\lceil N/b\rceil}}$, then $P=\trans(G,P',b)\in\pqPeb{G}$.}
\begin{lemma}\lemlab{lem:translegal}
\translegal
\end{lemma}

\FullVersion{The entire procedure $\trans(G,P',b)$ is formally described in \algref{alg:trans} in \appref{algorithms}, and an example for the reversible pebbling strategy can be found in \figref{fig:metapebbling} in \appref{examples}.}{The formal definition of the procedure $\trans(G,P',b)$ and an example for the reversible pebbling strategy can be found in the full version.} Now we observe the following theorem describing the space-time cost of the converted pebbling in terms of the cost of the reduced pebbling of the line graph. \FullVersion{}{We defer the proof of \thmref{thm:reduce} to the full version.}

\newcommand{\reducethm}{Given a DAG $G=(V,E)$ with $|V|=N$ nodes, a reduced line graph $L_{\lceil N/b\rceil}=(V',E')$ with $|V'|=\lceil N/b\rceil$ nodes (where $b$ is a positive integer), and a legal reversible pebbling $P'\in\pqPeb{L_{\lceil N/b\rceil}}$, there exists a legal reversible pebbling $P=\trans(G,P',b)\in\pqPeb{G}$ such that 
\[\pqpeb_\Cspacetime(P)\leq 2b^2 \pqpeb_\Cspacetime(P') + 2b\pqpeb_\Ctime(P')\cdot \numskip(G,b).\]}
\begin{theorem}\thmlab{thm:reduce}
\reducethm
\end{theorem}
\FullVersion{
\begin{proof}
Consider the algorithm $P=\trans(G,P',b)$ as shown in \algref{alg:trans} in \appref{algorithms}. We argue that the reversible pebbling $P$ is legal in \appref{missing} and focus here on analyzing the cost of the pebbling $P$. First, we consider the time cost of $P$. Notice that in each round $P'_j$ in $P'$ (of the line graph $L_{\lceil N/b\rceil}$), we have two cases: if $j\neq\lasta(P',\lceil N/b\rceil)$, we need $b$ rounds to place/remove pebbles in the corresponding blocks in $G$; otherwise, i.e., $j=\lasta(P',\lceil N/b\rceil)$, we need $b+N-\left(\lceil N/b\rceil-1\right)b-1 \leq 2b$ rounds to place/remove pebbles in the corresponding blocks in $G$. Hence, we have
\begin{equation*}
\pqpeb_\Ctime(P) \leq b\left(\pqpeb_\Ctime(P')-1\right) + 2b = b\left(\pqpeb_\Ctime(P')+1\right).
\end{equation*}
When it comes to the space cost of the pebbling $P$, we need space for the pebbling $P'$ multiplied by the block size since each node in $P'$ has a 1-1 correspondence between each block of size $b$ in $G$. Furthermore, we additionally need space for the skip nodes as they should not be removed to make the pebbling $P=\trans(G,P',b)$ legal. That is, we have
\begin{equation*}
\pqpeb_\Cspace(P) \leq b\cdot\pqpeb_\Cspace(P') + \numskip(G,b).
\end{equation*}
Combining these inequalities together, we can conclude that
\begin{align*}
\pqpeb_\Cspacetime(P) &= \pqpeb_\Cspace(P)\cdot\pqpeb_\Ctime(P)\\
&\leq \left( b\cdot\pqpeb_\Cspace(P') + \numskip(G,b)\right) \cdot b\left(\pqpeb_\Ctime(P')+1\right)\\
&= 2b^2 \pqpeb_\Cspacetime(P') + 2b\pqpeb_\Ctime(P')\cdot \numskip(G,b).\qedmath
\end{align*}
\end{proof}}{}

\paragraph{Analysis on DRSample.}
DRSample \cite{CCS:AlwBloHar17} is the first practical construction of an iMHF which modified the edge distribution of Argon2i. Consider a DAG $G=(V=[N],E)$. Intuitively, similar to Argon2i, each node $v\in V\setminus\{1\}$ has at most two parents, i.e., there is a directed edge $(v-1,v)\in E$ and a directed edge from a random predecessor $r(v)$. While Argon2i-A picks $r(v)$ uniformly at random from $[v-2]$, DRSample picks $r(v)$ according to the following random process: (1) We randomly select a bucket index $i \leq \log v$, (2) We randomly sample $r(v)$ from the bucket  $B_i(v)=\{u: 2^{i-1}<v-u\leq 2^i\}$. \FullVersion{We observe the following lemma which (whp) upper bounds the number of skip nodes when we sample $G$ according to this distribution.

\newcommand{\skipnodesDRS}{Let $G_\DRS=(V_\DRS=[N],E_\DRS)$ be a randomly sampled graph according to the DRSample edge distribution. Then with high probability, we have $\numskip\left(G_\DRS, \left\lceil \frac{N}{\log^2 N}\right\rceil\right)=\O{\frac{N\log\log N}{\log N}}$.}
\begin{lemma}\lemlab{lem:skipnodesDRS}
\skipnodesDRS
\end{lemma}

The full proof of \lemref{lem:skipnodesDRS} can be found in \appref{missing}. Here, we only give a brief intuition. To count the number of skip nodes, we need to find edges with length longer than $b$ so that the edge skips over a block. There are at most $\log v - \log b$ (out of $\log v$) buckets which potentially could result in a skip node i.e., any edge $(r(v),v)$ with length $v-r(v)\leq b$ cannot produce a new skip node. The probability that the edge $(r(v),v)$ is longer than $b$ is at most $1-\log b/\log v \leq 1- \log b/\log N = \log(N/b)/\log N$. Thus, the expected number of skip nodes in DRSample is at most $N \log(N/b)/\log N$ and standard concentration bounds imply that the number of skip nodes will be upper bounded by $\mathcal{O}(N \log(N/b)/\log N)$ with high probability. Setting $b=\lceil N/\log^2N\rceil$ we can conclude that the expected number of skip nodes in DRSample is at most $\mathcal{O}(N\log\log N/\log N)$ with high probability. Applying the result from \lemref{lem:skipnodesDRS} to \thmref{thm:reduce}, we have the following space-time cost of reversible pebbling for DRSample.}{We can upper bound the number of skip nodes when we sample $G$ according to this distribution. In particular, we observe that $\numskip\left(G_\DRS, \left\lceil \frac{N}{\log^2 N}\right\rceil\right)=\O{\frac{N\log\log N}{\log N}}$ where $G_\DRS$ is a randomly sampled graph according to the DRSample edge distribution. Intuitively, to count the number of skip nodes, we need to find edges with length $>b$ so that the edge skips over a block. There are at most $\log v-\log b$ (out of $\log v$) buckets which potentially could result in a skip node, which implies that the probability that the edge $(r(v),v)$ is longer than $b$ is at most $1-\log b/\log v\leq 1-\log b/\log N = \log(N/b)/\log N$. Thus, the expected number of skip nodes in DRSample is at most $N\log(N/b)/\log N$ and standard concentration bounds imply that the number of skip nodes will be upper bounded by $\mathcal{O}(N\log(N/b)/\log N)$ with high probability. Setting $b=\lceil N/\log^2N\rceil$ we can conclude that the expected number of skip nodes in DRSample is at most $\mathcal{O}(N\log\log N/\log N)$ with high probability. Further details can be found in the full version. Applying this result to \thmref{thm:reduce}, we have the following space-time cost of reversible pebbling for DRSample.}

\newcommand{\stcostDRS}{Let $G_\DRS=(V_\DRS=[N],E_\DRS)$ be a randomly sampled graph according to the DRSample edge distribution. Then with high probability, $\pqpeb_\Cspacetime(G_\DRS)=\O{\frac{N^2\log\log N}{\log N}}$.}
\begin{corollary}\corlab{cor:stcostDRS}
\stcostDRS
\end{corollary}
\FullVersion{
\begin{proof}
Given $G_\DRS$, we can consider a reduced line graph $L_{\lceil \log^2N \rceil}=(V',E')$ with $|V'|=\lceil \log^2N \rceil$. Then by \thmref{thm:reduce}, we have
\[ \pqpeb_\Cspacetime(G_\DRS)\leq2
\left(\frac{N}{\log^2N}\right)^2 \pqpeb_\Cspacetime(L_{\lceil \log^2N \rceil}) + \frac{2N}{\log^2N}\cdot\pqpeb_\Ctime(L_{\lceil \log^2N \rceil})\cdot \numskip\left(G_\DRS,\left\lceil\frac{N}{\log^2N}\right\rceil\right). \]
By \thmref{thm:line}, we have $\pqpeb_\Cspacetime(L_{\lceil \log^2N \rceil}) = \O{(\log N)^{2\left(1+\frac{2}{\sqrt{\log \log^2N}}\right)}}$ and $\pqpeb_\Ctime(L_{\lceil \log^2N \rceil}) = \O{\log^2 N}$. By \lemref{lem:skipnodesDRS}, we have
\begin{align*}
\pqpeb_\Cspacetime(G_\DRS) &\leq 2
\left(\frac{N}{\log^2N}\right)^2 \O{(\log N)^{2\left(1+\frac{2}{\sqrt{\log \log^2N}}\right)}} + 
\frac{2N}{\log^2N}\cdot\O{\log^2 N}\cdot \O{\frac{N\log\log N}{\log N}}\\
&= \O{\frac{N^2}{\log^{2(1-\sqrt{2/\log\log N})}N} + \frac{N^2\log\log N}{\log N}}=\O{\frac{N^2\log\log N}{\log N}}.\qedmath
\end{align*}
\end{proof}
}{The proof of \corref{cor:stcostDRS} is deferred to the full version and we only give a brief intuition here. Basically, we can reduce $G_\DRS$ to the induced line graph $L_{\lceil\log^2 N\rceil}$ of size $\lceil\log^2 N\rceil$. Then by plugging in the reversible time and space-time cost of $L_{\lceil\log^2 N\rceil}$ and the number of skip nodes of $G_\DRS$ in \thmref{thm:reduce} with setting $b=\lceil N/\log^2 N\rceil$, we can conclude that $\pqpeb_\Cspacetime(G_\DRS)=\O{\frac{N^2\log\log N}{\log N}}$.
}

%% file: genpeb.tex
% !TEX root = tcc-main.tex
% To compile, run the following command:
%   latexmk -pdf hw_template.tex 
%
% To edit, you can use your favorite text editor or LaTeX editors such as:
%   Texmaker and TeXworks.
%
% To setup your own TeX system, you can install TeX Live. See:
%   https://www.tug.org/texlive/
\section{Reversible Pebbling Attacks for Minimizing Cumulative Complexity}\seclab{sec:revpeb}

% \begin{definition}
%     A graph $G$ is \emph{$(e,d)$-depth reducible} if there's a set $S$ of size $e$ such that the depth of $G-S$ is at most $d$. We call $S$ a \emph{depth-reducing set}.
%     \label{depthReducible}
% \end{definition}
% \seunghoon{I moved this definition to p.6 in the introduction. We used the notion depth-robust/depth-reducible graph in Section 1.1 already.}

In this section, we adapt the depth-reducing pebbling attack $\genpeb$ from Alwen and Blocki~\cite{C:AlwBlo16} to a reversible pebbling attack with the same asymptotic CC. The pebbling attack of Alwen and Blocki~\cite{C:AlwBlo16} applies to any $(e,d)$-reducible DAG $G$ with $e=o(N)$ and $d=o(N)$. We first provide an overview of their pebbling strategy before describing how we extend the attack to obtain a reversible pebbling. 

\paragraph{Intuitive Overview of \cite{C:AlwBlo16} Attack.} Suppose that we are given a DAG $G=(V=[N], E)$ with constant indegree $\delta$ along with a depth-reducing set $S$ of size $|S| \leq e$. Intuitively, the pebbling attack of Alwen and Blocki~\cite{C:AlwBlo16} can be divided into a series of alternating ``light phases" and ``balloon phases." It is also helpful to imagine partitioning the nodes $[N]$ into intervals $I_i=[(i-1)g+1,ig]$ of $g$ consecutive nodes. 

\begin{itemize}
    \item \emph{Light Phases:} During the $i$\th light phase our goal will be to pebble all of the nodes in $I_i$ over the next $g$ consecutive pebbling rounds. The pre-condition for the $i$\th light phase is that we start off with pebbles on all of the nodes $\left(\parents(I_i) \cup S \right) \cap [(i-1)g]$ where $\parents(I_i) = \{ u : \exists v \in I_i ~\mathtt{s.t.}~(u,v) \in E \}$ denotes the set of parents of nodes in $I_i$. Similarly, the post-condition for the $i$\th light phase is that we have pebbles on all of the nodes $\left(\parents(I_i) \cup S \right) \cap [(i-1)g]  \cup I_i$. If $P_j=\left(\parents(I_i) \cup S \right) \cap [(i-1)g]$ denotes the initial pebbling configuration at the start of the light phase then we can set $P_{j+x} = P_{j} \cup [(i-1)g, (i-1)g+x]$ so that $P_{j+g}$ gives us our post-condition. During each light phase we keep {\em at most } $\left| \left(\parents(I_i) \cup S \right) \cap [(i-1)g]  \cup I_i \right| \leq e + \delta g + g$ pebbles on the graph. Thus, the total cost incurred during each light phase is at most $(e+\delta g + g) g$ and the total cost incurred over all $\frac{N}{g}$ light phases is at most $N(e+\delta g + g)$. 
    \item \emph{Balloon Phases:} The $i$\th balloon phase takes place immediately after the $i$\th light phase with the goal of quickly recovering previously discarded pebbles to satisfy the pre-condition for the next ($(i+1)$\st) light phase. In particular, the post-condition for the $i$\th balloon phase should match the pre-condition for the $(i+1)$\st light phase. The pre-condition for the $i$\th balloon phase is that our starting configuration contains pebbles on all of the nodes $S \cap [ig]$. During a balloon phase, we are not worried about space so we can recover pebbles on the entire set $[ig]$ within $d$ rounds by exploiting the fact that $G-S$ contains no directed path of length $d$. Once we have recovered pebbles on the entire set $[ig]$ we can then discard all of the pebbles that are not needed for the next light phase. Thus, the total cost incurred by each individual balloon phase is at most $dN$ and the total cost incurred over all $\frac{N}{g}$ balloon phases is at most $\frac{N^2 d}{g}$.
\end{itemize}

\paragraph{Formal Description of \cite{C:AlwBlo16} Pebbling Attack.}

    Let $G=([N],E)$ be an $(e,d)$-reducible graph and $S$ be a depth-reducing set of size $e$. 
 The pebbling $P=\left(P_1,\dots, P_N\right)$ from $\genpeb$ lasts $N$ rounds, pebbling each $i$ on round $i$. The algorithm operates in disjoint and consecutive intervals of $I_c=[(c-1)g+1,cg]$ where $g\in[d,N]$. At the start of $I_c$, we perform a ``light phase'' with the following start and end conditions:
\begin{enumerate}
    \item $\startlight_c\coloneqq\text{$P_{(c-1)g+1}= \{(c-1)g+1\}\cup S_{\le (c-1)g+1}\cup\parents({I_c})_{\le (c-1)g+1}$, and}$
    \item $\mathsf{EndLight}_c\coloneqq\text{$P_{cg}\subseteq S_{\le cg}\cup\{cg\}$.}$
\end{enumerate}
Intuitively, before we start the light phase, we need to have pebbles on the depth-reducing set, and the parents of the nodes we are about to pebble. By the end of the light phase, all we require for the light phase is that the depth-reducing set and $cg$ is pebbled. As the name suggests, we can define a low-CC pebbling to implement the light phase. 
For $j\in[g]$ and $k=cg + j$ we define the required pebbles for the $j$\th pebble of the $c$\th light phase as 
\[\lightreq^c_j=[(c-1)g+j:k]\cup S_{\le k}\cup\parents(I_c)_{\le k}.\]
To pebble the nodes in $I_c$, we will let $P_{(c-1)g+j}=\lightreq^c_{j}$. This allows the light phase to only add a pebble to $i\in I_c$ at step $i$, keeping the overall number of pebbles low.  
However, this leaves us unprepared for the next light phase (we need to satisfy $\startlight_{c+1}$ by step $cg+1$). To fix this, we wait until near the end of the light phase and start a ``balloon phase'', pebbling as many nodes as possible to quickly pebble a superset of the nodes needed for the next light phase. Since the depth-reducing set $S_{\le k}$ is pebbled on step $k$, we can always pebble $G([k])$ by step $k+d$, and then we can simply remove all the pebbles that are unneeded for $\startlight_{c+1}$. 
Specifically, the start and end conditions of the balloon phase are
\begin{enumerate}
    \item $\startballoon_c\coloneqq\text{$S_{\le cg-d+1}\subseteq P_{cg-d+1}$}$, and
    \item $\mathsf{EndBalloon}_c\coloneqq\text{$P_{cg}=[cg]$}$.
\end{enumerate}
This way at round $cg-d+1$ we start the balloon phase which pebbles all available nodes for $d$ steps as to eventually satisfy $\startlight_{c+1}$. Let $R(P_k)=\{v\mid \parents(v)\subseteq P_k\}$ denote the set of nodes that can be pebbled in the next step. Then we can define the balloon requirements per step as  $\balloonreq^c_{1}=\lightreq^c_{g-d}\cup R(\lightreq^c_{g-d})$ and for $1< j\le d$, \[\balloonreq^c_{j}=\balloonreq^c_{j-1}\cup R(\balloonreq^c_{j-1}).\] Now we can define the low-CC pebbling $P$ such that
\[P_{cg+j}=\begin{cases}
\lightreq^c_{j} &\text{if } j\le g-d\text{, and}\\
\lightreq^c_{j}\cup\balloonreq^c_{g-j}&\text{otherwise.}
\end{cases}\]
% For each $I_c$, we perform a ``light phase'' with the goal of simply pebbling each node $i\in I_c$ on round $i$. 
%     If the pebbling only consisted of the light phase, then $P_{cg}$ would be just $\lightreq_{cg}$. This may leave us unprepared for $\light_{c+1}$, so at round $cg-d+1$ we start of the ``balloon phase'' $\balloon_c$, which pebbles all available nodes for $d$ steloonReqc
% g−d+j = BalloonReqc
% g−d+j ∪ R(BalloonReqc
% g−d+j ).
% Now we can define the low CC pebbling sequence as to satisfy the precondition for $\light_{c+1}$. Since $S_{\le cg-d+1}$ is pebbled at the start of $\balloon_c$, $G([cg])$ has depth at most $d$, meaning $P_{cg}=[cg]$. Then, at the beginning of $\light_{c+1}$, we remove the necessary pebbles to get $P_{cg+1}=\lightreq^{c+1}_{1}$.
    
%     To summarize, below are the invariants for $\genpeb$ \cite{C:AlwBlo16}.
%     \begin{enumerate}
%         \item For each $i\in[N]$, $S_{\le i}\cup{i}\subseteq P_i$.
%         \item For $1\le c\le \ceil{N/g}$ and $k=cg+j<N$ for $j\in[g]$, the following must hold:
%         \begin{enumerate}
%             \item During the light phase $\light_c$, if $j< g-d$, $P_{k}=\lightreq_k$
%             % \item For $i \mod g\le g-d$, $P_{i}=S_{\le i}\cup $
%             \item During the balloon phase $\balloon_c$ ($j\ge g-d$), we require that $P_k=\lightreq_k\cup R(P_{k-1})_{\le cg+g}$, where $R(P_{k-1})$ denotes the set of nodes which could be legally pebbled in the pebbling step after $P_{k-1}$.
%         \end{enumerate}
%         \item $P_N=\{N\}$.
%     \end{enumerate}

It follows that $P$ is a legal pebbling for $G$ \cite{C:AlwBlo16}. We have that $\abs{\lightreq^c_j}\le e+g(\delta+1)$, since each $\lightreq^c_j$ contains at most $S$, the parents of $I_c$, and $I_c$ itself. Next we have that each $\abs{\balloonreq^c_j}\le N$, so $\ppeb_\Ccc(P)\le N\left(\frac{Nd}{g}+e+(\delta+1)g\right)$. We essentially take the depth-reducing pebbling attack $\genpeb$ from Alwen and Blocki~\cite{C:AlwBlo16}, and adapt it to be reversible without changing the upper bound asymptotically.  
    \newcommand{\cleanup}{\mathcal C}
    \subsection{A Reversible Pebbling Attack}
    In this section, we define a reversible pebbling extension of $\genpeb$. We begin with an intuitive overview. We first observe that most pebbling rounds in $\genpeb$ are monotonic, i.e., $P_{i+1} \supset P_i$. Since monotonic transitions do not involve removing pebbles, these transitions remain legal in the reversible pebbling. However, the $\genpeb$ pebbling strategy does include occasionally include a non-monotonic transition at the end of each balloon phase where unnecessary pebbles are simply discarded before the next light phase. Suppose that $P_i$ denotes the pebbling state at the end of the balloon phase and $P_{i+1}$ denotes the pebbling configuration after discarding all of the unnecessary pebbles for the next light phase. The non-monotonic transition from $P_i$ to $P_{i+1}$ will (almost certainly) not be a legal reversible pebbling transition. Our main challenge is to define a legal reversible pebbling sequence which takes us from the pebbling state at the end of each balloon phase to the pebbling state and the beginning of the next light phase. However, while $P_i \not \subseteq P_{i+1}$ we do have $P_{i+1} \subseteq P_i$ since $P_{i+1}$ was obtained by discarding pebbles. Our key idea is to argue that there is a short (i.e., $\leq d$ rounds) monotonic pebbling sequences which takes us $P_{i+1}$ to $P_i$, i.e., we exploit the fact that any node in $P_{i+1}$ has depth at most $d$ in $G-P_i$ and run a balloon phase. Since this short pebbling sequence is monotonic, it is also reversible. Thus, there is a legal reversible sequence from pebbling state $P_i$ to $P_{i+1}$ in at most $d$ steps.   
    
        In the pebbling corresponding to $\genpeb$, we must remove all unnecessary pebbles after the balloon phase to match the precondition for the following light phase. This is inherently irreversible since they are all removed at once (instead of unpebbling them). For ease of analysis, we start each balloon phase after the corresponding light phase and adjust our notation accordingly. We denote the $j$\th step of the $c$\th light phase as $\lightreq^c_j$ and the balloon phase as $\balloonreq^c_j$. The light phases themselves remain the same. The first half of each balloon phase remains the same, but $\balloonreq_{cg-d+j}$ for $j\in[d]$ as defined above. However, we need to ``clean up'' after each balloon phase in order to meet the precondition for the following light phase, taking care to ensure these pebbling sequences are reversible. In this new balloon phase, we must satisfy the following while maintaining reversibility:
        \begin{enumerate}
            \item $\startballoon_c\coloneqq\text{$S_{\le cg-d+1}\subseteq \balloonreq^c_{g-d+1}$}$,
            \item $\midballoon_c\coloneqq\text{$\balloonreq^c_{g}=[cg]$}$, and
            \item $\mathsf{EndBalloon}_c=\text{$\balloonreq^c_{g+1}=\lightreq^{c+1}_1$}$.
        \end{enumerate}

        For a sequence of pebbling configurations $P=(P_1,\dots,P_t)$ let $\rev(P)=(P_t,\dots, P_1)$. We will use ``monotonic'' pebbling sequences to generate the reversible pebbling defined above. 
        \begin{definition}
            A sequence of pebbling moves $\langle P_1,\dots, P_t\rangle$ is \emph{monotonic} if $P_1\subseteq P_2\subseteq\dots\subseteq P_t$.
        \end{definition}
            
    Immediately, we get that each light phase and balloon phase is monotone. The following result shows that monotonic pebbling sequences are reversible, and the formal proof is left to \FullVersion{\appref{missing}}{the full version}. Intuitively, the additional rules added to the reversible pebbling game only restrict which pebbles we can remove. If a sequence is monotonic then these additional restrictions do not apply. 
        \newcommand{\monorev}{
           If a legal (non-reversible) pebbling sequence $P=\langle P_1,\dots, P_t\rangle$ is monotonic, then $P$ is a legal reversible pebbling sequence.
        }
        \begin{lemma}\lemlab{monorev}
            \monorev
            \label{monorev}
        \end{lemma}
    
    Now we can use the reverse of the greedy pebbling sequence from $\lightreq_{cg+1}$ to $[cg]$.
    \begin{claim}[Satisfying $\startlight_c$ and $\mathsf{EndLight}_c$]
       The sequence $(\lightreq^c_1,\allowbreak\dots, \lightreq^c_{g})$ as defined above is a monotonic pebbling sequence. 
    \end{claim}
    \begin{proof}
        By construction each $\lightreq^c_i\subseteq \lightreq^c_{i+1}$.\qed
    \end{proof}
    
    Likewise, the first half of our balloon phase (which is the same as the classical version) is also monotonic, because it simply pebbles all possible nodes each round.
    \begin{claim}[Satisfying $\startballoon_c$ and $\midballoon_c$]
       The pebbling sequence $(\balloonreq^c_{1},\dots,\allowbreak\balloonreq^c_{d})$ is a legal monotonic pebbling sequence. 
    \end{claim}
    Next, we need to complete the balloon phase. For $j\in [d]$ we let \[\balloonreq^c_{d+j}=\balloonreq^c_{d-j+1}\cup{\lightreq^{c+1}_1.}\]
    \begin{claim}[Satisfying $\midballoon_c$ and $\Endballoon_c$]
        The pebbling sequence $(\balloonreq^c_{2d},\dots,\allowbreak\balloonreq^c_{d+1})$ is a legal monotonic pebbling sequence.
    \end{claim}
    \begin{proof}
        This follows from the fact that $(\balloonreq^c_{1},\dots, \balloonreq^c_d)$ is monotonic.\qed
    \end{proof}
  \newcommand{\pebrev}{P_{\text{rev}}}
Now we can define the first half of our low CC pebbling. Let $\lightreq^c=(\lightreq^c_1,\dots, \lightreq^c_g)$, $\balloonreq^c=(\balloonreq^c_1,\dots, \balloonreq^c_{2d})$, and \[\pebrev^1= \lightreq^1+\balloonreq^1+
\dots+\lightreq^{\ceil{N/g}}+\balloonreq^{\ceil{N/g}},\] where $+$ denotes sequence concatenation. The sequence $\pebrev^1$ is a legal reversible pebbling sequence by the construction of the $\lightreq^c_j$ and $\balloonreq^c_j$ and by \lemref{reverseMerge}. The proof of \lemref{reverseMerge} can be found in \FullVersion{\appref{missing}}{the full version}.
\newcommand{\reverseMerge}{
 Let $\langle P_1,\dots, P_t\rangle$ and $\langle P'_1,\dots, P'_{t'}\rangle$ be two legal reversible pebblings for some graph $G$ such that $P_t=P'_{t'}$. Then for any $T\subseteq P_t$, 
    \[\langle P_1,\dots, P_t, P'_{t'-1}\cup T, P'_{t'-2}\cup T,\dots, P'_{1}\cup T\rangle\] is also a legal reversible pebbling sequence for $G$.
}
\begin{lemma}\lemlab{reverseMerge}
        \reverseMerge
\end{lemma}

Now we can construct the last part of the pebbling, which simply cleans up by reversing all the prior steps while keeping $N$ pebbled. For a pebbling sequence $Q=(Q_1,\dots, Q_t)$ and a set $K$ we let $
Q(K)=(Q_1\cup K,\dots, Q_t\cup K)$. Let \[{\pebrev^2}'=\rev(\balloonreq^{\ceil{N/g}})+\rev(\lightreq^{\ceil{N/g}-1}),\dots,\rev( \balloonreq^1)+\rev(\lightreq^1)+(\emptyset)\] and \[\pebrev^2={\pebrev^2}'(\{N\}).\]
The final pebbling for $G$ is $\pebrev=\pebrev^1+\pebrev^2$.

For an arbitrary $(e,d)$-depth reducible DAG $G$ with depth-reducing set $S$ of size at most $e$ and any $g\in[d,N]$, we let $\rgenpeb(G,e,d,S,g)$ denote the pebbling for $G$ constructed exactly as $\pebrev$. The following lemma shows that $\pebrev$ is a legal reversible pebbling for $G$. The proof follows from \lemref{reverseMerge}.

\newcommand{\revpebLegal}{
 For any $(e,d)$-depth reducible graph $G$ with depth-reducing set $S$ of size at most $e$. Then for any $g\in[d,N]$, $\pebrev=\rgenpeb(G,e,d,S,g)$ is a legal reversible pebbling for $G$.
}
\begin{lemma}
    \revpebLegal
    \lemlab{revpebLegal}
\end{lemma}

Next we analyze the CC of $\pebrev$. This follows similarly to $\genpeb$, except we need to account for the cost of the extra length of the balloon phases and the cost of having to reverse the pebbling. 
\begin{theorem}\thmlab{thm:revcc}
    For any $(e,d)$-depth reducible graph $G$ on $N$ nodes and any $g\in[d, N]$, \[\pqpeb_\Ccc(G)\le 2N\left(\frac{2Nd}{g}+e+(\delta+1)g\right)+N+\frac{2Nd}{g}.\] 
\end{theorem}
\begin{proof}
    We already know that \[\sum_{c\in[\ceil{N/g}]}\sum_{i\in [g]}\abs{\lightreq^{c-1}_i}\le N/g(e+(\delta+1)g).\]
    Next $\balloonreq^c$ contains at most $2d$ pebbling steps, so \begin{align*}
        \sum_{c\in{\ceil{N/g}-1}}\sum_{i\in[d]}\abs{\balloonreq^{c-1}_i}&\le\sum_{c\in{\ceil{N/g}-1}}2Nd\\
        &\le N\frac{2Nd}{g}.
    \end{align*}

Then $\pqpeb_\Ccc(\pebrev^1)\le N\left(\frac{2Nd}{g}+e+(\delta+1)g\right)$. Next, it's immediate that \begin{align*}
    \pqpeb_\Ccc(\pebrev^2)\le \pqpeb_\Ccc(\pebrev^1)+\abs{\pebrev^1}&\le N\left(\frac{2Nd}{g}+e(\delta+1)g\right)+N+\frac{2Nd}{g},\\
    % &=N\left(\frac{3Nd}{g}+e(\delta+1)g+1\right),
\end{align*}
so \[\pqpeb_\Ccc(\pebrev)\le 2N\left(\frac{2Nd}{g}+e+(\delta+1)g\right)+N+\frac{2Nd}{g}.\qedmath\]
\end{proof}
For any iMHF corresponding to a DAG $G$ the reversible cumulative pebbling complexity obtained from our attack is identical to the attack from Alwen and Blocki~\cite{C:AlwBlo16}. In particular, for Argon2i-A and Argon2i-B we obtain \corref{Argon2iCMC}:

\begin{corollary} \corlab{Argon2iCMC}
Let $G_\ArgonA=(V_A=[N],E_A)$ and $G_\ArgonB=(V_B=[N],E_B)$ be randomly sampled graphs according to the Argon2i-A and Argon2i-B edge distributions, respectively. Then with high probability, we have $\pqpeb_\Ccc(G_\ArgonA)=\O{N^{1.75} \log N }$ and $\pqpeb_\Ccc(G_\ArgonB)=\O{N^{1.8}}$.
\end{corollary}
\begin{proof}
Alwen and Blocki~\cite{C:AlwBlo16} argued that (whp) a random Argon2i-A DAG $G_\ArgonA=(V_A=[N],E_A)$ is $(e,d)$-reducible with $d=\sqrt{N}$ and $e=\O{N^{0.75} \log N}$. The result for Argon2i-A now follows directly from  \thmref{thm:revcc} by setting $g=e$. Alwen and Blocki~\cite{ESP:AlwBlo17} also argued that (whp) a random Argon2i-B DAG $G_\ArgonB=(V_B=[N],E_B)$ is $(e,d)$-reducible with $d=N^{0.6}$ and $e=\O{N^{0.8}}$. The result for Argon2i-B now follows directly from \thmref{thm:revcc} by setting $g=e$.\qed
\end{proof}

Similar to Alwen and Blocki~\cite{C:AlwBlo16} we can also obtain a general upper bound for any DAG $G$ with constant indegree. 

\begin{corollary} \corlab{generalCMCDAG}
For any DAG $G=(V=[N],E)$ with constant indegree $\delta = \O{1}$ the reversible cumulative pebbling cost at most $\pqpeb_\Ccc(G)=\O{\frac{N^2 \log \log N}{ \log N}}$.
\end{corollary}
\begin{proof}
Any DAG $G=(V=[N],E)$ with constant indegree $\delta = \O{1}$ is $(e,d)$-reducible with $d= \frac{N}{\log^2 N}$ and $e= \O{\frac{N \log \log N}{\log N}}$. The result now follows immediately from \thmref{thm:revcc} by setting $g=e$.\qed
\end{proof}

%% file: openquestions.tex
\section{\FullVersion{Conclusion and Open Questions}{Open Questions}}
\seclab{sec:open}

\FullVersion{We introduced the parallel reversible pebbling game and applied it to analyze the reversible space-time complexity of a line graph, Argon2i-A, Argon2i-B, and DRSample. Our motivation is to understand the post-quantum resistance of these MHFs to brute-force pre-image attacks. In particular, we showed that the reversible space-time cost of pebbling a line graph of size $N$ is $\O{N^{1+\frac{2}{\sqrt{\log N}}}}$ by extending Bennett's reversible pebbling strategy \cite{Bennett89}. We also showed that there is a reversible pebbling strategy for an $(e,d)$-reducible indegree-2 DAG $G$ of size $N$ with the space-time cost $\O{Ne+Nd2^d}$, which becomes meaningful whenever $e=o(N)$ and $d2^d=o(N)$. We applied this attack to Argon2i-A and Argon2i-B to yield reversible pebbling attacks with space-time cost $\O{N^2\log\log N/\sqrt{\log N}}$ and $\O{N^2/\sqrt[3]{\log N}}$ for Argon2i-A and Argon2i-B, respectively. Finally, we introduced a general reversible pebbling attack on a DAG $G$ of size $N$ by reducing the graph to a line graph $L_{\lceil N/b\rceil}$, and given a legal quantum pebbling $P'$ of the line graph with space-time cost $\pqpeb_\Cspacetime(P')$, we provided a legal quantum pebbling $P$ of $G$ with space-time cost $\O{sN+b^2\pqpeb_\Cspacetime(P')}$, where $s$ denotes the number of skip nodes in $G$. Tuning the parameter $b=\O{N/\log^2N}$ the skip number for DRSample is $\O{\frac{N\log\log N}{\log N}}$ leading to a reversible pebbling attack with space-time cost $\O{N^2\log\log N/\log N}$. We also studied the cumulative pebbling cost of reversible pebblings by extending the depth-reducing attack from Alwen and Blocki~\cite{C:AlwBlo16} on depth-reducible graphs.}{}

One open question is to determine if there is a DAG with constant indegree having (parallel) reversible space-time cost $\Omega(N^2)$. Alternatively, is there a generic reversible pebbling attack which rules out this possibility. Blocki et al.~\cite{C:BHKLXZ19} proposed a new iMHF candidate called DRS+BRG (DRSample plus Bit-Reversal Graph) by overlaying a bit-reversal graph \cite{LT82,AC:ForLucWen14} on top of DRSample, which provides the best resistance to known \emph{classical} pebbling attacks. This graph could plausibly have parallel reversible space-time cost $\Omega(N^2)$. In particular, none of the reversible pebbling attacks we proposed perform well against DRS+BRG --- there is no small depth-reducing set for DRS+BRG and the extra bit-reversal edges ensure that the number of skip nodes will be large as well. 

Another research challenge is to either develop asymptotically stronger reversible pebbling attacks for iMHFs such as Argon2i or establish lower bounds on the parallel reversible space-time complexity. Finally, Alwen et al.~\cite{EC:AlwBloPie17} defined a recursive (non-reversible) pebbling attack for DAGs that are $(e_i,d_i)$-depth-reducible for a range of parameters $(e_i,d_i)$ with $e_i > e_{i+1}$ and $d_{i+1} < d_i$. The recursive pebbling attack often leads to  improved pebbling attacks with asymptotically lower cumulative pebbling cost (CC). Thus, extending the recursive pebbling attack to the reversible pebbling setting is a natural challenge.  

% We are also interested if there are more efficient quantum pebbling attacks for Argon2i-A/B, or there are stronger lower bounds for these DAGs. Finally, while the space-time cost is the most relevant metric for quantum preimage attacks, it is still worthwhile to continue to study the reversible cumulative pebbling cost (CC) of iMHF candidates, i.e., Argon2i-A/B, DRSample, DRS+BRG, etc. In particular, can the recursive (non-reversible) pebbling strategies of Alwen et al.~\cite{EC:AlwBloPie17} be extended to the setting of reversible pebblings?

%\added{Another intriguing open question is to find a pebbling reduction for the reversible pebbling game in the Quantum Random Oracle Model (QROM). It would be fascinating to find a similar pebbling reduction to the pebbling reductions in the \emph{classical} Parallel Random Oracle Model (PROM), however, there is a fundamental challenge that would need to be addressed: query extraction. In the QROM setting, we cannot directly observe the random oracle queries without disturbing the state while the extractor in the classical pebbling reduction works by simulating the PROM algorithm and observing the queries and extracting labels from those queries. There are a few recent techniques for straight line extraction in the QROM in restricted settings~\cite{EC:DFMS22} but it is unclear whether or not these techniques can be adapted and applied in the context of a QROM pebbling reduction.}

%% file: acknowledgement.tex
% !TEX root = tcc-main.tex
\section*{Acknowledgements}
Jeremiah Blocki was supported in part by the National Science Foundation under NSF CAREER Award CNS-2047272 and NSF Award CCF-1910659. 
Seunghoon Lee was supported in part by the Center for Science of Information (NSF CCF-0939370). 
Blake Holman was supported in part by a Ross Fellowship at Purdue University and by a Ford Foundation Fellowship. We would like to thank anonymous reviewers for helpful feedback which improved this paper. 

%% file: reversible.tex
% !TEX root = tcc-main.tex
\allowdisplaybreaks

\section{Reversible Pebbling Strategies on a Line Graph}\applab{app:line}

In this section, we review the (sequential) reversible pebbling strategy for a line graph $L_N$ from Li and Vit\'{a}nyi \cite{LiVit96} which translated Bennett's reversible simulation \cite{Bennett89} into a (sequential) reversible pebbling game, and we give a reversible pebbling strategy with a better space-time cost. We remark that a similar argument seems to be implicitly assumed in Bennett \cite{Bennett89}, though no explicit description of the reversible pebbling is provided. Hence, we include this result for completeness.

The strategy from Li and Vit\'{a}nyi \cite{LiVit96} is as follows: let $I(k)=I(k-1)\circ i_{k-1}\circ I(k-2)\circ i_{k-2}\circ \ldots \circ I_1\circ i_1 \circ I_0\circ i_0$, where for $j=0,1,\ldots,k$, $I(j)$ denotes the sequence of consecutive locations in $L_N$, $I(0)=\{\}$, and $i_j$ denotes the node incident to $I(j)$. Let $N(k)$ be the size of $I(k)$. Then we have $N(k)=\sum_{i=1}^{k-1}(N(i)+1)$ with $N(0)=0$, which implies $N(k)=2^k-1$. The reversible pebbling works as we pebble the block $I(j)$, pebble $i_j$, and unpebble $I(j)$. If $P$ denotes such reversible pebbling, then Li and Vit\'{a}nyi \cite{LiVit96} showed that $P\in\qPeb{I(k)}$ and $\qpeb_\Cspace(P)=\O{\log N(k)}$ and $\qpeb_\Ctime(P)=\O{N(k)^{\log 3}}$, since if we denote $S(k)$ (resp. $T(k)$) the space (resp. time) cost to reversibly pebble $I(k)$ then it satisfies the recurrence relation $S(k)=\max_i\{i+S(k-i)\}=S(k-1)+1$ (resp. $T(k)=2T(k-1)+1+2T(k-2)+1+\ldots+2T(1)+1=3T(k-1)+1$). Taken together, the reversible space-time cost for this pebbling strategy is $\qpeb_\Cspacetime(P)=\O{N(k)^{\log 3}\log N(k)}$, which implies that $\qpeb_\Cspacetime(L_N)=\O{N^{\log 3}\log N}$.

\subsubsection{A Reversible Pebbling Strategy with a Better ST Cost.}

We extend this approach and first recursively define the sequence of consecutive locations $I(k)$ (of nodes in a line graph) as
\[ I(k) = \begin{cases}
I(k-1)'\circ I(k-2)'\circ \ldots\circ I(0)',&\text{if }k>0\\
\{\},&\text{if }k=0,
\end{cases}
\]
where $\circ$ denotes concatenation and for $0\leq j<k$, $I(j)'$ is defined as
\[I(j)' := I(j)^{(1)}\circ i_j^{(1)}\circ I(j)^{(2)}\circ i_j^{(2)}\circ\ldots\circ I(j)^{(c)}\circ i_j^{(c)},\] 
where $A^{(\ell)}$ denotes the $\ell\th$ copy of $A$. 
Let $N(k)$ be the size of $I(k)$. Since $I(j)'$ consists of $c$ copies of $I(j)$ and a single node $i_j$, we observe that $N(k)$ satisfies the following recursive relation: 
\begin{align*}
N(k) &= c(N(k-1)+1) + c(N(k-2)+1) + \cdots + c(N(0)+1) \\
&= c(N(k-1)+1) + N(k-1) \\
&= (c+1)N(k-1) + c,
\end{align*}
which implies that $N(k) = \Theta((c+1)^k)$. We have the following pebbling strategy $\prevpeb(I(k))$ for $I(k)$ as shown in \algref{alg:prevpeb}. Here, $\prevpeb^{-1}(I(k))$ denotes the procedure which runs $\prevpeb(I(k))$ in reverse order, i.e., it starts with the final configuration of $\prevpeb(I(k))$ and ends with the starting configuration of $\prevpeb(I(k))$. Intuitively, it sequentially pebbles $I(k-1)',\ldots,I(0)'$ in this order. When we pebble $I(j)'=I(j)^{(1)}\circ i_j^{(1)}\circ I(j)^{(2)}\circ i_j^{(2)}\circ\ldots\circ I(j)^{(c)}\circ i_j^{(c)}$, we run $\prevpeb(I(j)^{(1)})$ to pebble the first block, and we pebble the incident node $i_j^{(1)}$. After that, for $\ell=2,\ldots,c$, we run $\prevpeb^{-1}(I(j)^{(\ell-1)})$ to remove pebbles from the previous block and we move forward to pebble the next block by running $\prevpeb(I(j)^{(\ell)})$. %Note that when running $\prevpeb^{-1}(I(j)^{(\ell-1)})$, this should contain removing pebble from $i_j^{(\ell-1)}$ to not leave any intermediate nodes at the final pebbling configuration.
\begin{algorithm}[ht!]
\caption{The Procedure $\prevpeb(I(k))$.}
\alglab{alg:prevpeb}
\DontPrintSemicolon
\SetAlgoLined
\KwIn{}
\KwOut{}
\vspace{0.3em}
\For{$j=k-1,\ldots,0$}{
	Run $\prevpeb(I(j)^{(1)})$\;
	Pebble node $i_j^{(1)}$\;
	\For{$\ell=2,\ldots,c$}{
	    Run $\prevpeb^{-1}(I(j)^{(\ell-1)})$\;
		Run $\prevpeb(I(j)^{(\ell)})$\; 
		Pebble node $i_j^{(\ell)}$
	}
}
\Return{}
\end{algorithm}

Now we have the following lemma.

\begin{lemma}\lemlab{lem:line}
For a line graph $L_N$, there exists a reversible pebbling $P\in \qPeb{L_N}$ such that $\qpeb_\Cspacetime(P) = \O{N^{1+(2+o(1))\frac{1}{\sqrt{\log N}}}}$, and a parallel reversible pebbling $P'\in \pqPeb{L_N}$ with $\pqpeb_\Cspacetime(P')=\O{N^{1+\frac{2}{\sqrt{\log N}}}}$.
\end{lemma}

\begin{proof}
Let $P=\prevpeb(I(k))$. Then we can easily see that $P\in\qPeb{I(k)}$.

We first consider the space cost of $P=\prevpeb(I(k))$. Intuitively, we first observe that when we pebble $I(k)$, the space cost of pebbling $I(k-1)'$ dominates the space cost of pebbling $I(k-2)',\ldots,I(0)'$ since they are recursively defined. Now when pebbling $I(k-1)'$, we would need to remove pebbles from $I(k-1)^{(\ell)}$ and add pebbles on $I(k-1)^{(\ell+1)}$ for each $\ell$, and further, we would need pebbling $c$ intermediate nodes $i_{k-1}^{(1)},\ldots,i_{k-1}^{(c)}$. Hence, the space complexity of $\prevpeb(I(k))$ satisfies the recurrence relation $\qpeb_\Cspace(\prevpeb(I(k))) \leq  \qpeb_\Cspace(\prevpeb(I(k-1))) + c$. Solving the recurrence relation gives us $\qpeb_\Cspace(P)=\O{ck}$.

When it comes to the time cost of $P$, we would need to be careful and we define $T_f(j)$ to be the amount of time to place a pebble on the last node of $I(j)$, without removing pebbles from earlier nodes in $I(j)$, and we define $T_r(j)$ to be the amount of time to remove such nodes afterwards. Since the pebbling is reversible, we can easily observe that $T_f(j)=T_r(j)$ for each $j$. In \algref{alg:prevpeb}, when we pebble $I(k)$, we pebble $I(k-1)'$ first which contains the procedure that (1) we pebble $I(k-1)^{(1)}$ and $i_{k-1}^{(1)}$, (2) we remove pebble from $I(k-1)^{(1)}$ and pebble $I(k-1)^{(2)}$, and (3) keep repeating this until the last copy $I(k-1)^{(c)}$ and $i_{k-1}^{(c)}$ is pebbled. Taken together, we have the following recurrence relation for $T_f(k)$:

\begin{align*}
T_f(k) &= 2c(T_f(k-1)+1) + T_r(k-1) + \underbrace{2c(T_f(k-2)+1) + T_r(k-2) + \cdots}_{=T_f(k-1)} \\
&= (2c+1)T_f(k-1)+T_r(k-1)+c = (2c+2)T_f(k-1)+c,
\end{align*}
which tells us that $T_f(k)=\O{(2c+2)^k}$. Hence, $\qpeb_\Ctime(P)=T_f(k)+T_r(k)=\O{(2c+2)^k}$ and we have $\qpeb_\Cspacetime(P)=\qpeb_\Cspace(P)\qpeb_\Ctime(P)=\O{ck(2c+2)^k}$.

To express $\qpeb_\Cspacetime(P)$ in terms of $N(k)=\Theta((c+1)^k)$, by setting $c=2^k$ we observe that
\[ \frac{\qpeb_\Cspacetime(P)}{N(k)} = \O{ck2^k} = \O{k4^k}. \]
We observe that $k4^k = (2^{k^2})^{\frac{2k+\log k}{k^2}} = N(k)^{\frac{2k+\log k}{k^2}}$. Since $N(k)=\Omega(2^{k^2})$ implies $k = \O{\sqrt{\log N(k)}}$, we have
\begin{align*}
    \qpeb_\Cspacetime(P) &= \O{N(k)\cdot k4^k} = \O{N(k)^{1+\frac{2k+\log k}{k^2}}}\\ 
    &= \O{N(k)^{1+\frac{2}{\sqrt{\log N(k)}}+\frac{\log\log N(k)}{2\log N(k)}}} = \O{N(k)^{1+(2+o(1))\frac{1}{\sqrt{\log N(k)}}}}.
\end{align*}
We can parallelize this strategy by removing pebbles from $I(k-1)^{(\ell)}$ and adding pebbles on $I(k-1)^{(\ell+1)}$ in parallel. If we denote this pebbling strategy $P'=\mathsf{PRevPeb}(I(k))$, then the recurrence relation for the space cost becomes $\pqpeb_\Cspace(\mathsf{PRevPeb}(I(k)))\leq 2\pqpeb_\Cspace(\mathsf{PRevPeb}(I(k-1)))+c$, which yields $\pqpeb_\Cspace(P')=\O{c2^k}$. On the other hand, parallelizing it could save time in each recursion by half, which implies that the recurrence relation for the time cost becomes $T_f(k)=(c+2)T_f(k-1)+c$, which gives us the time cost $\pqpeb_\Ctime(P')=\O{(c+2)^k}$. In this case, $\pqpeb_\Cspacetime(P')=\O{c(2c+4)^k}$. To express $\qpeb_\Cspacetime(P)$ in terms of $N(k)=\Theta((c+1)^k)$, by setting $c+1=2^k$ we observe that
\begin{align*}
\frac{\pqpeb_\Cspacetime(P')}{N(k)} &= \O{\frac{c(2c+4)^k}{(c+1)^k}}\\
&= \O{c2^k\left(1+\frac{1}{c+1}\right)^k}\\
&= \O{2^k\cdot 2^k \cdot 1}= \O{4^k},
\end{align*}
since $\left(1+\frac{1}{c+1}\right)^k = \left( 1+\frac{1}{2^k}\right)^k = \Theta(1)$\footnote{for $k>0$ we have $(1+\frac{1}{2^k})^k < (1+\frac{1}{2^k})^{2^k}<e$.}. Since $\O{4^k} = \O{(2^{k^2})^{2/k}} = \O{N(k)^{2/k}}$ and $N(k)=\Theta(2^{k^2})$ implies $k = \Theta(\sqrt{\log N(k)})$, we have
\[ \pqpeb_\Cspacetime(P') = \O{N(k)\cdot 4^k} = \O{N(k)^{1+\frac{2}{k}}} = \O{N(k)^{1+\frac{2}{\sqrt{\log N(k)}}}}. \qedmath\]
\end{proof}

% $2^{c+1}=2^k=\O{N(k)^{\frac{1}{\log\log N(k)}}}$, and $c = \O{k}=\O{\log N(k)}$. Hence, we have $\pqpeb_\Cspacetime(P) = \O{N(k)^{1+\frac{1}{\log\log N(k)}}\log N(k)}$, which is stronger than $\O{N(k)^{1+\epsilon}}$ (for any constant $\epsilon>0$) from \thmref{thm:line} or $\O{N(k)^{1+\epsilon}\log N(k)}$ from \cite{Bennett89}.

% \seunghoon{The following is from our previous discussions. I don't follow how we get such recurrence relations for $S(k),T(k)$ and formula for $N(k)$: Previously, $\pqpeb_\Cspacetime(P) = \O{N(k)^{1+\epsilon}\log N(k)}$ where $k=N(k)^\epsilon$. Previously, we would have had $S(k,c)=\O{ck}$ and $T(k,c)=2cT(k-1,c)+c=(2c)^k$. Letting $c=2^k$, we have $N(k)=(2^k)^k=2^{k^2}$, $T(k)/N(k)=\Theta(1)$, and $S(k) = c2^k = 4^k = N(k)^{2/k}$. Hence, $\pqpeb_\Cspacetime(P)=N(k)^{1+\frac{2}{k}}=N(k)^{1+\frac{2}{\sqrt{\log N(k)}}} \ll N(k)^{1+\frac{1}{\log\log N(k)}}$.}

%% file: examples.tex
% !TEX root = tcc-main.tex
\section{Reversible Pebbling Strategy Examples}\applab{examples}

\subsection{Example on an $(e,d)$-Reducible Graph}\applab{example:ed}

In this example, we give a DAG $G=(V=[N],E)$ with $N=16$, and $E=\{(i,i+1):i\in[15]\}\cup\{((i-1)4+1,(i-1)4+3),((i-1)4+1,(i-1)4+4),((i-1)4+1,(i-1)4+5),((i-1)4+1,(i-1)4+6):i\in[3]\}\cup\{(13,15),(13,16)\}$, as shown in \figref{fig:reducible}. We observe that $G$ is $(4,3)$-reducible. 
 
\input{figures/fig-reducible}

Recall that $P=(P_0,P_1,\ldots,P_{2N})$ such that $P_0=\varnothing$, for $v\in[N],P_v\coloneqq S_{\leq v}\cup B_v$, and for $N<v\leq 2N$, $P_v\coloneqq P_{2N-v}\cup\{N\}$ is a legal reversible pebbling for $G$, as shown in \lemref{lem:edlegal}, where $B_v := \bigcup_{j=1}^{d+1}\bigcup_{i=j}^{d+1} \left( A_{v+1-j,S,i}\cup A_{v-1+j,S,i} \right)$, with the definition $A_{w,S,i} \coloneqq \left\{v:\lpath_{G-S_{\leq w-1}}(v,w)= i\right\}$. For example, when we compute $P_8$ for the graph above, it is described as 
\begin{align*}
P_8 &= S_{\leq 8}\cup B_8\\
    &= \{1,5\} \cup \bigcup_{j=1}^4\bigcup_{i=j}^4 (A_{9-j,S,i}\cup A_{7+j,S,i})\\
    &= \{1,5\} \cup (A_{8,S,1}\cup A_{8,S,2} \cup A_{8,S,3}\cup A_{8,S,4}) \cup (A_{7,S,2}\cup A_{7,S,3}\cup A_{7,S,4} \cup A_{9,S,2} \cup A_{9,S,3} \cup A_{9,S,4})\\
    &\qquad \cup (A_{6,S,3} \cup A_{6,S,4} \cup A_{10,S,3} \cup A_{10,S,4}) \cup (A_{5,S,4} \cup A_{11,S,4})\\
    &= \{1,5\} \cup \{6,7 \} \cup \{6,7,8\} \cup \{ \} \cup \{2 \}\\
    &= \{1,2,5,6,7,8\}.
\end{align*}
Then entire pebbling process is illustrated in \figref{fig:edpebbling}. Note that for our example, $\pqpeb_\Ctime(P)=32=2N$ and $\pqpeb_\Cspace(P)=9$, which leads to $\pqpeb_\Cspacetime(P)=32\cdot 9 = 288$. While this is not a significant improvement on the na\"ive pebbling strategy for small $N=16$, the space-time costs scale with $\O{N}$ for the graphs defined above.

\input{figures/fig-edpebbling}

\newpage
\subsection{Example of a Reversible Pebbling Using an Induced Line Graph}\applab{example:meta}

In this example, we give a DAG $G=(V=[18],E)$ with the edge distribution as illustrated in \figref{fig:metapebbling}. As we discussed in \secref{sec:attack:reducedgraph}, we reduce our DAG $G$ to a line graph $L_6$ by choosing the block size $b=3$. Given an efficient reversible pebbling $P'$ of $L_6$ as shown in \figref{fig:metapebbling-linegraph}, we apply $\trans(G,P',b=3)$ to produce a legal reversible pebbling of $G$. Note that we have $\lasta(P',6)=6$, hence, in reversible pebbling rounds of $G$ that corresponds to $P'_6$, we pebble all nodes in $B_6$ and delete pebbles from the block in reverse topological order except for the last node as shown in \algref{alg:trans} in \appref{algorithms}, which takes $b+N-(\lceil N/b\rceil-1)b-1=3+18-(6-1)3-1 = 5$ steps to complete. We also note that pebbles colored in red are \emph{skip nodes}, which will be kept until the corresponding block is deleted for the last time, i.e., we keep a skip node $v\in B_i$ until we reach rounds that correspond to $P'_j$ (of $L_6$) with $j=\lastd(P',i)$.

\input{figures/fig-metapebbling}

%% file: figures/fig-reducible.tex
% !TEX root = ../ec-main.tex
\begin{figure}[ht!]
\resizebox{\linewidth}{!}{
  \begin{tikzpicture} 
  [scale=0.88,node distance=1cm,auto,font=\small,
    every node/.style={node distance=2cm},
    peb/.style={rectangle,draw,fill=black,inner sep=2pt, minimum width=0.4cm,minimum height=0.4cm},
    unpeb/.style={rectangle,draw,fill=white,inner sep=2pt, minimum width=0.4cm,minimum height=0.4cm},
    node/.style={circle,draw,black,fill=black!20,inner sep=2pt, minimum width=0.6cm},
    delete/.style={circle,draw,black!10,fill=black!5,text=black!10,inner sep=2pt, minimum width=0.6cm}
    ]
% original DAG G
\node[node] (1) at (0,0) {$1$};
\node[node] (2) at (1.5,0) {$2$};
\node[node] (3) at (3,0) {$3$};
\node[node] (4) at (4.5,0) {$4$};
\node[node] (5) at (6,0) {$5$};
\node[node] (6) at (7.5,0) {$6$};
\node[node] (7) at (9,0) {$7$};
\node[node] (8) at (10.5,0) {$8$};
\node[node] (9) at (12,0) {$9$};
\node[node] (10) at (13.5,0) {$10$};
\node[node] (11) at (15,0) {$11$};
\node[node] (12) at (16.5,0) {$12$};
\node[node] (13) at (18,0) {$13$};
\node[node] (14) at (19.5,0) {$14$};
\node[node] (15) at (21,0) {$15$};
\node[node] (16) at (22.5,0) {$16$};

\path[draw,thick,->] (1) -- (2);
\path[draw,thick,->] (2) -- (3);
\path[draw,thick,->] (3) -- (4);
\path[draw,thick,->] (4) -- (5);
\path[draw,thick,->] (5) -- (6);
\path[draw,thick,->] (6) -- (7);
\path[draw,thick,->] (7) -- (8);
\path[draw,thick,->] (8) -- (9);
\path[draw,thick,->] (9) -- (10);
\path[draw,thick,->] (10) -- (11);
\path[draw,thick,->] (11) -- (12);
\path[draw,thick,->] (12) -- (13);
\path[draw,thick,->] (13) -- (14);
\path[draw,thick,->] (14) -- (15);
\path[draw,thick,->] (15) -- (16);

\path[draw,thick,->] (1) edge[bend left=30] (3);
\path[draw,thick,->] (1) edge[bend left=35] (4);
\path[draw,thick,->] (1) edge[bend right=20] (5);
\path[draw,thick,->] (1) edge[bend right=25] (6);

\path[draw,thick,->] (5) edge[bend left=30] (7);
\path[draw,thick,->] (5) edge[bend left=35] (8);
\path[draw,thick,->] (5) edge[bend right=20] (9);
\path[draw,thick,->] (5) edge[bend right=25] (10);

\path[draw,thick,->] (9) edge[bend left=30] (11);
\path[draw,thick,->] (9) edge[bend left=35] (12);
\path[draw,thick,->] (9) edge[bend right=20] (13);
\path[draw,thick,->] (9) edge[bend right=25] (14);

\path[draw,thick,->] (13) edge[bend left=30] (15);
\path[draw,thick,->] (13) edge[bend left=35] (16);

% G-S where S={1,5,9,13}
\node[delete] (r1) at (0,-2) {$1$};
\node[node] (r2) at (1.5,-2) {$2$};
\node[node] (r3) at (3,-2) {$3$};
\node[node] (r4) at (4.5,-2) {$4$};
\node[delete] (r5) at (6,-2) {$5$};
\node[node] (r6) at (7.5,-2) {$6$};
\node[node] (r7) at (9,-2) {$7$};
\node[node] (r8) at (10.5,-2) {$8$};
\node[delete] (r9) at (12,-2) {$9$};
\node[node] (r10) at (13.5,-2) {$10$};
\node[node] (r11) at (15,-2) {$11$};
\node[node] (r12) at (16.5,-2) {$12$};
\node[delete] (r13) at (18,-2) {$13$};
\node[node] (r14) at (19.5,-2) {$14$};
\node[node] (r15) at (21,-2) {$15$};
\node[node] (r16) at (22.5,-2) {$16$};

\path[draw,thick,->,color=black!10] (r1) -- (r2);
\path[draw,thick,->] (r2) -- (r3);
\path[draw,thick,->] (r3) -- (r4);
\path[draw,thick,->,color=black!10] (r4) -- (r5);
\path[draw,thick,->,color=black!10] (r5) -- (r6);
\path[draw,thick,->] (r6) -- (r7);
\path[draw,thick,->] (r7) -- (r8);
\path[draw,thick,->,color=black!10] (r8) -- (r9);
\path[draw,thick,->,color=black!10] (r9) -- (r10);
\path[draw,thick,->] (r10) -- (r11);
\path[draw,thick,->] (r11) -- (r12);
\path[draw,thick,->,color=black!10] (r12) -- (r13);
\path[draw,thick,->,color=black!10] (r13) -- (r14);
\path[draw,thick,->] (r14) -- (r15);
\path[draw,thick,->] (r15) -- (r16);

\path[draw,thick,->,color=black!10] (r1) edge[bend left=30] (r3);
\path[draw,thick,->,color=black!10] (r1) edge[bend left=35] (r4);
\path[draw,thick,->,color=black!10] (r1) edge[bend right=20] (r5);
\path[draw,thick,->,color=black!10] (r1) edge[bend right=25] (r6);

\path[draw,thick,->,color=black!10] (r5) edge[bend left=30] (r7);
\path[draw,thick,->,color=black!10] (r5) edge[bend left=35] (r8);
\path[draw,thick,->,color=black!10] (r5) edge[bend right=20] (r9);
\path[draw,thick,->,color=black!10] (r5) edge[bend right=25] (r10);

\path[draw,thick,->,color=black!10] (r9) edge[bend left=30] (r11);
\path[draw,thick,->,color=black!10] (r9) edge[bend left=35] (r12);
\path[draw,thick,->,color=black!10] (r9) edge[bend right=20] (r13);
\path[draw,thick,->,color=black!10] (r9) edge[bend right=25] (r14);

\path[draw,thick,->,color=black!10] (r13) edge[bend left=30] (r15);
\path[draw,thick,->,color=black!10] (r13) edge[bend left=35] (r16);

  \end{tikzpicture}}
  \caption{An $(e,d)$-reducible DAG $G$ of $N=4N'$ nodes, with $e=N'$ and $d=3$ (we set $N'=4$ in the figure above). Note that with depth-reducing set $S=\{1,5,9,13\}$, we have an original DAG $G$ (top) and the induced subgraph $G-S$ (bottom).} \figlab{fig:reducible}
\end{figure}
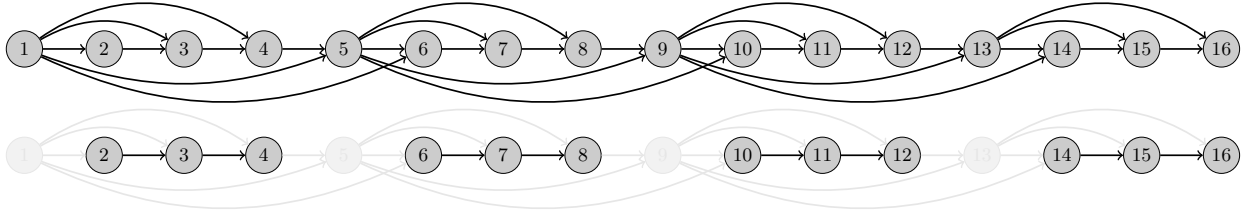

%% file: figures/fig-edpebbling.tex
% !TEX root = ../ec-main.tex
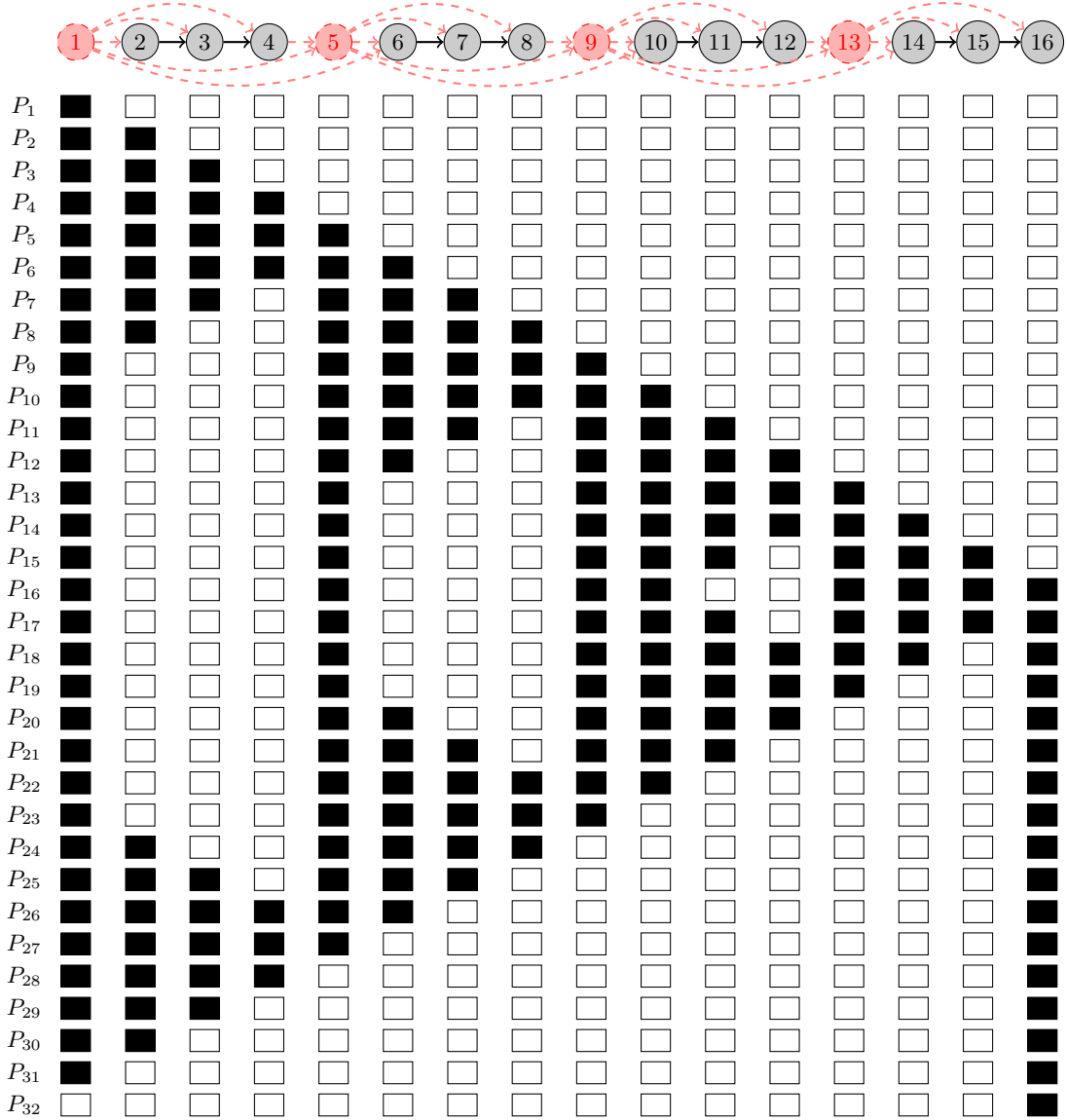
\begin{figure}[ht!]
\centering
  \begin{tikzpicture} 
  [scale=0.88,node distance=1cm,auto,font=\small,
    every node/.style={node distance=2cm},
    reduce/.style={draw,circle,dashed,red,fill=red!30,inner sep=2pt, minimum width=0.5cm},
    peb/.style={rectangle,draw,fill=black,inner sep=2pt, minimum width=0.4cm,minimum height=0.3cm},
    unpeb/.style={rectangle,draw,fill=white,inner sep=2pt, minimum width=0.4cm,minimum height=0.3cm},
    node/.style={circle,draw,black,fill=black!20,inner sep=2pt, minimum width=0.5cm}
    ]
\node[reduce] (1) at (0,0) {$1$};
\node[node] (2) at (1,0) {$2$};
\node[node] (3) at (2,0) {$3$};
\node[node] (4) at (3,0) {$4$};
\node[reduce] (5) at (4,0) {$5$};
\node[node] (6) at (5,0) {$6$};
\node[node] (7) at (6,0) {$7$};
\node[node] (8) at (7,0) {$8$};
\node[reduce] (9) at (8,0) {$9$};
\node[node] (10) at (9,0) {$10$};
\node[node] (11) at (10,0) {$11$};
\node[node] (12) at (11,0) {$12$};
\node[reduce] (13) at (12,0) {$13$};
\node[node] (14) at (13,0) {$14$};
\node[node] (15) at (14,0) {$15$};
\node[node] (16) at (15,0) {$16$};

\path[draw,thick,->,dashed,red!50] (1) -- (2);
\path[draw,thick,->] (2) -- (3);
\path[draw,thick,->] (3) -- (4);
\path[draw,thick,->,dashed,red!50] (4) -- (5);
\path[draw,thick,->,dashed,red!50] (5) -- (6);
\path[draw,thick,->] (6) -- (7);
\path[draw,thick,->] (7) -- (8);
\path[draw,thick,->,dashed,red!50] (8) -- (9);
\path[draw,thick,->,dashed,red!50] (9) -- (10);
\path[draw,thick,->] (10) -- (11);
\path[draw,thick,->] (11) -- (12);
\path[draw,thick,->,dashed,red!50] (12) -- (13);
\path[draw,thick,->,dashed,red!50] (13) -- (14);
\path[draw,thick,->] (14) -- (15);
\path[draw,thick,->] (15) -- (16);

\path[draw,thick,->,dashed,red!50] (1) edge[bend left=30] (3);
\path[draw,thick,->,dashed,red!50] (1) edge[bend left=35] (4);
\path[draw,thick,->,dashed,red!50] (1) edge[bend right=20] (5);
\path[draw,thick,->,dashed,red!50] (1) edge[bend right=25] (6);

\path[draw,thick,->,dashed,red!50] (5) edge[bend left=30] (7);
\path[draw,thick,->,dashed,red!50] (5) edge[bend left=35] (8);
\path[draw,thick,->,dashed,red!50] (5) edge[bend right=20] (9);
\path[draw,thick,->,dashed,red!50] (5) edge[bend right=25] (10);

\path[draw,thick,->,dashed,red!50] (9) edge[bend left=30] (11);
\path[draw,thick,->,dashed,red!50] (9) edge[bend left=35] (12);
\path[draw,thick,->,dashed,red!50] (9) edge[bend right=20] (13);
\path[draw,thick,->,dashed,red!50] (9) edge[bend right=25] (14);

\path[draw,thick,->,dashed,red!50] (13) edge[bend left=30] (15);
\path[draw,thick,->,dashed,red!50] (13) edge[bend left=35] (16);

% round 1
\node[peb] at (0,-1) {};
\foreach \x in {1,...,15}
	\node[unpeb] at (\x,-1) {};
% round 2
\foreach \x in {0,1}
	\node[peb] at (\x,-1.5) {};
\foreach \y in {2,...,15}
	\node[unpeb] at (\y,-1.5) {};
% round 3
\foreach \x in {0,1,2}
	\node[peb] at (\x,-2) {};
\foreach \y in {3,...,15}
	\node[unpeb] at (\y,-2) {};
% round 4
\foreach \x in {0,1,2,3}
	\node[peb] at (\x,-2.5) {};
\foreach \y in {4,...,15}
	\node[unpeb] at (\y,-2.5) {};
% round 5
\foreach \x in {0,1,2,3,4}
	\node[peb] at (\x,-3) {};
\foreach \y in {5,...,15}
	\node[unpeb] at (\y,-3) {};
% round 6
\foreach \x in {0,...,5}
	\node[peb] at (\x,-3.5) {};
\foreach \y in {6,...,15}
	\node[unpeb] at (\y,-3.5) {};
% round 7
\foreach \x in {0,1,2,4,5,6}
	\node[peb] at (\x,-4) {};
\foreach \y in {3,7,8,...,15}
	\node[unpeb] at (\y,-4) {};
% round 8
\foreach \x in {0,1,4,5,6,7}
	\node[peb] at (\x,-4.5) {};
\foreach \y in {2,3,8,9,...,15}
	\node[unpeb] at (\y,-4.5) {};
% round 9
\foreach \x in {0,4,5,6,7,8}
	\node[peb] at (\x,-5) {};
\foreach \y in {1,2,3,9,10,...,15}
	\node[unpeb] at (\y,-5) {};
% round 10
\foreach \x in {0,4,5,...,9}
	\node[peb] at (\x,-5.5) {};
\foreach \y in {1,2,3,10,11,...,15}
	\node[unpeb] at (\y,-5.5) {};
% round 11
\foreach \x in {0,4,5,6,8,9,10}
	\node[peb] at (\x,-6) {};
\foreach \y in {1,2,3,7,11,12,...,15}
	\node[unpeb] at (\y,-6) {};
% round 12
\foreach \x in {0,4,5,8,9,10,11}
	\node[peb] at (\x,-6.5) {};
\foreach \y in {1,2,3,6,7,12,13,14,15}
	\node[unpeb] at (\y,-6.5) {};
% round 13
\foreach \x in {0,4,8,9,...,12}
	\node[peb] at (\x,-7) {};
\foreach \y in {1,2,3,5,6,7,13,14,15}
	\node[unpeb] at (\y,-7) {};
% round 14
\foreach \x in {0,4,8,9,...,13}
	\node[peb] at (\x,-7.5) {};
\foreach \y in {1,2,3,5,6,7,14,15}
	\node[unpeb] at (\y,-7.5) {};
% round 15
\foreach \x in {0,4,8,9,10,12,13,14}
	\node[peb] at (\x,-8) {};
\foreach \y in {1,2,3,5,6,7,11,15}
	\node[unpeb] at (\y,-8) {};
% round 16
\foreach \x in {0,4,8,9,12,13,14,15}
	\node[peb] at (\x,-8.5) {};
\foreach \y in {1,2,3,5,6,7,10,11}
	\node[unpeb] at (\y,-8.5) {};
	
% round 17 (= 15 + N)
\foreach \x in {0,4,8,9,10,12,13,14,15}
	\node[peb] at (\x,-9) {};
\foreach \y in {1,2,3,5,6,7,11}
	\node[unpeb] at (\y,-9) {};
% round 18 (= 14 + N)
\foreach \x in {0,4,8,9,...,13,15}
	\node[peb] at (\x,-9.5) {};
\foreach \y in {1,2,3,5,6,7,14}
	\node[unpeb] at (\y,-9.5) {};	
% round 19
\foreach \x in {0,4,8,9,...,12,15}
	\node[peb] at (\x,-10) {};
\foreach \y in {1,2,3,5,6,7,13,14}
	\node[unpeb] at (\y,-10) {};
% round 20
\foreach \x in {0,4,5,8,9,10,11,15}
	\node[peb] at (\x,-10.5) {};
\foreach \y in {1,2,3,6,7,12,13,14}
	\node[unpeb] at (\y,-10.5) {};
% round 21 = 11 + N
\foreach \x in {0,4,5,6,8,9,10,15}
	\node[peb] at (\x,-11) {};
\foreach \y in {1,2,3,7,11,12,13,14}
	\node[unpeb] at (\y,-11) {};
% round 22
\foreach \x in {0,4,5,...,9,15}
	\node[peb] at (\x,-11.5) {};
\foreach \y in {1,2,3,10,11,...,14}
	\node[unpeb] at (\y,-11.5) {};
% round 23 = 9 + N
\foreach \x in {0,4,5,6,7,8,15}
	\node[peb] at (\x,-12) {};
\foreach \y in {1,2,3,9,10,...,14}
	\node[unpeb] at (\y,-12) {};
% round 24
\foreach \x in {0,1,4,5,6,7,15}
	\node[peb] at (\x,-12.5) {};
\foreach \y in {2,3,8,9,...,14}
	\node[unpeb] at (\y,-12.5) {};
% round 25
\foreach \x in {0,1,2,4,5,6,15}
	\node[peb] at (\x,-13) {};
\foreach \y in {3,7,8,...,14}
	\node[unpeb] at (\y,-13) {};
% round 26 = 6 + N
\foreach \x in {0,...,5,15}
	\node[peb] at (\x,-13.5) {};
\foreach \y in {6,...,14}
	\node[unpeb] at (\y,-13.5) {};
% round 27
\foreach \x in {0,1,2,3,4,15}
	\node[peb] at (\x,-14) {};
\foreach \y in {5,...,14}
	\node[unpeb] at (\y,-14) {};
% round 28
\foreach \x in {0,1,2,3,15}
	\node[peb] at (\x,-14.5) {};
\foreach \y in {4,...,14}
	\node[unpeb] at (\y,-14.5) {};
% round 29
\foreach \x in {0,1,2,15}
	\node[peb] at (\x,-15) {};
\foreach \y in {3,...,14}
	\node[unpeb] at (\y,-15) {};	
% round 30
\foreach \x in {0,1,15}
	\node[peb] at (\x,-15.5) {};
\foreach \y in {2,...,14}
	\node[unpeb] at (\y,-15.5) {};	
% round 31
\node[peb] at (0,-16) {};
\node[peb] at (15,-16) {};
\foreach \x in {1,...,14}
	\node[unpeb] at (\x,-16) {};
% round 32
\node[peb] at (15,-16.5) {};
\foreach \x in {0,...,14}
	\node[unpeb] at (\x,-16.5) {};

% legends
\node at (-0.8,-1) {$P_1$};
\node at (-0.8,-1.5) {$P_2$};
\node at (-0.8,-2) {$P_3$};
\node at (-0.8,-2.5) {$P_4$};
\node at (-0.8,-3) {$P_5$};
\node at (-0.8,-3.5) {$P_6$};
\node at (-0.8,-4) {$P_7$};
\node at (-0.8,-4.5) {$P_8$};
\node at (-0.8,-5) {$P_9$};
\node at (-0.8,-5.5) {$P_{10}$};
\node at (-0.8,-6) {$P_{11}$};
\node at (-0.8,-6.5) {$P_{12}$};
\node at (-0.8,-7) {$P_{13}$};
\node at (-0.8,-7.5) {$P_{14}$};
\node at (-0.8,-8) {$P_{15}$};
\node at (-0.8,-8.5) {$P_{16}$};
\node at (-0.8,-9) {$P_{17}$};
\foreach \x in {18,...,32}
	\node at (-0.8,-\x/2-0.5) {$P_{\x}$};
  \end{tikzpicture}
  \caption{A parallel reversible pebbling strategy for an $(e,d)$-reducible graph with $N=16,e=N/4=4,$ and $d=3$. A filled square denotes a pebble on the node and an unfilled square denotes an unpebbled node.} \figlab{fig:edpebbling}
\end{figure}

%% file: figures/fig-metapebbling.tex
% !TEX root = ../ec-main.tex
\begin{figure}[ht!]
\resizebox{\linewidth}{!}{
  \begin{tikzpicture} 
  [scale=0.88,node distance=1cm,auto,font=\footnotesize,
    every node/.style={node distance=2cm},
    stuck/.style={draw,circle,red,fill=red!30,inner sep=2pt, minimum width=0.5cm},
    meta/.style={draw,thick,dashed,rectangle,forestgreen,minimum width=2.25cm,minimum height=1.25cm,pattern=north west lines, pattern color=forestgreen},
    peb/.style={rectangle,draw,fill=black,inner sep=2pt, minimum width=0.4cm,minimum height=0.3cm},
    unpeb/.style={rectangle,draw,fill=white,inner sep=2pt, minimum width=0.4cm,minimum height=0.3cm},
    node/.style={circle,draw,black,fill=black!20,inner sep=2pt, minimum width=0.5cm}
    ]
\node (G) at (-1.1,0.4) {\large $G$};
\node[node] (1) at (0,0.4) {$1$};
\node[stuck] (2) at ($(1)+(1,0)$) {$2$};
\node[node] (3) at ($(2)+(1,0)$) {$3$};

\node[stuck] (4) at ($(3)+(2,0)$) {$4$};
\node[node] (5) at ($(4)+(1,0)$) {$5$};
\node[stuck] (6) at ($(5)+(1,0)$) {$6$};

\node[node] (7) at ($(6)+(2,0)$) {$7$};
\node[node] (8) at ($(7)+(1,0)$) {$8$};
\node[node] (9) at ($(8)+(1,0)$) {$9$};

\node[node] (10) at ($(9)+(2,0)$) {$10$};
\node[node] (11) at ($(10)+(1,0)$) {$11$};
\node[stuck] (12) at ($(11)+(1,0)$) {$12$};

\node[node] (13) at ($(12)+(2,0)$) {$13$};
\node[node] (14) at ($(13)+(1,0)$) {$14$};
\node[node] (15) at ($(14)+(1,0)$) {$15$};

\node[node] (16) at ($(15)+(2,0)$) {$16$};
\node[node] (17) at ($(16)+(1,0)$) {$17$};
\node[node] (18) at ($(17)+(1,0)$) {$18$};

\begin{scope}[on background layer]
\foreach \x in {2,5,...,17}{
	\node[dashed,thick,ellipse,draw,forestgreen,fill=forestgreen!20,minimum width=2.7cm,minimum height=1cm] at (\x) {};}
\end{scope}

\path[draw,thick,->] (1) -- (2);
\path[draw,thick,->] (2) -- (3);
\path[draw,thick,->] (3) -- (4);
\path[draw,thick,->] (4) -- (5);
\path[draw,thick,->] (5) -- (6);
\path[draw,thick,->] (6) -- (7);
\path[draw,thick,->] (7) -- (8);
\path[draw,thick,->] (8) -- (9);
\path[draw,thick,->] (9) -- (10);
\path[draw,thick,->] (10) -- (11);
\path[draw,thick,->] (11) -- (12);
\path[draw,thick,->] (12) -- (13);
\path[draw,thick,->] (13) -- (14);
\path[draw,thick,->] (14) -- (15);
\path[draw,thick,->] (15) -- (16);
\path[draw,thick,->] (16) -- (17);
\path[draw,thick,->] (17) -- (18);

\path[draw,thick,->] (1) edge[bend left=30] (3);
\path[draw,thick,->] (3) edge[bend left=30] (5);
\path[draw,thick,->] (4) edge[bend left=30] (6);
\path[draw,thick,->] (6) edge[bend left=30] (9);
\path[draw,thick,->] (9) edge[bend left=30] (11);
\path[draw,thick,->] (12) edge[bend left=30] (14);
\path[draw,thick,->] (13) edge[bend left=30] (15);
\path[draw,thick,->] (16) edge[bend left=30] (18);

\path[draw,thick,->] (2) edge[bend right=30] (4);
\path[draw,thick,->] (2) edge[bend right=30] (7);
\path[draw,thick,->] (4) edge[bend right=30] (12);
\path[draw,thick,->] (5) edge[bend right=30] (8);
\path[draw,thick,->] (6) edge[bend right=30] (10);
\path[draw,thick,->] (11) edge[bend right=30] (13);
\path[draw,thick,->] (12) edge[bend right=30] (17);
\path[draw,thick,->] (14) edge[bend right=30] (16);

% P1
\node at (-1.65,-1.5) {\large\textcolor{forestgreen}{$P'_{1}$}};
	\draw[forestgreen,decorate,decoration={brace,amplitude=5pt,mirror}] (-1.1,-1) -- (-1.1,-2) node {};
\node[meta] at (1,-1.5) {};
\node at (-0.8,-1) {$P_1$};
\node at (-0.8,-1.5) {$P_2$};
\node at (-0.8,-2) {$P_3$};
% round 1
\node[peb] at (0,-1) {};
%\foreach \x in {1,...,15}
%	\node[unpeb] at (\x,-1) {};
% round 2
\foreach \x in {0,1}
	\node[peb] at (\x,-1.5) {};
%\foreach \y in {2,...,15}
%	\node[unpeb] at (\y,-1.5) {};
% round 3
\foreach \x in {0,1,2}
	\node[peb] at (\x,-2) {};
%\foreach \y in {3,...,15}
%	\node[unpeb] at (\y,-2) {};

%P2
\node at (-1.65,-3.5) {\large\textcolor{forestgreen}{$P'_{2}$}};
	\draw[forestgreen,decorate,decoration={brace,amplitude=5pt,mirror}] (-1.1,-3) -- (-1.1,-4) node {};
\node[meta] at (1,-3.5) {};	
\node[meta] at (5,-3.5) {};
\node at (-0.8,-3) {$P_4$};
\node at (-0.8,-3.5) {$P_5$};
\node at (-0.8,-4) {$P_6$};
% round 4
\foreach \x in {0,1,2,4}
	\node[peb] at (\x,-3) {};
%\foreach \y in {4,...,15}
%	\node[unpeb] at (\y,-2.5) {};
% round 5
\foreach \x in {0,1,2,4,5}
	\node[peb] at (\x,-3.5) {};
%\foreach \y in {5,...,15}
%	\node[unpeb] at (\y,-3) {};
% round 6
\foreach \x in {0,1,2,4,5,6}
	\node[peb] at (\x,-4) {};
%\foreach \y in {6,...,15}
%	\node[unpeb] at (\y,-3.5) {};

%P3
\node at (-1.65,-5.5) {\large\textcolor{forestgreen}{$P'_{3}$}};
	\draw[forestgreen,decorate,decoration={brace,amplitude=5pt,mirror}] (-1.1,-5) -- (-1.1,-6) node {};
\node[meta] at (1,-5.5) {};
\node[meta] at (5,-5.5) {};
\node[meta] at (9,-5.5) {};
\node at (-0.8,-5) {$P_7$};
\node at (-0.8,-5.5) {$P_8$};
\node at (-0.8,-6) {$P_9$};
% round 7
\foreach \x in {0,1,2,4,5,6,8}
	\node[peb] at (\x,-5) {};
% round 8
\foreach \x in {0,1,2,4,5,6,8,9}
	\node[peb] at (\x,-5.5) {};
% round 9
\foreach \x in {0,1,2,4,5,6,8,9,10}
	\node[peb] at (\x,-6) {};
	
%P4
\node at (-1.65,-7.5) {\large\textcolor{forestgreen}{$P'_{4}$}};
	\draw[forestgreen,decorate,decoration={brace,amplitude=5pt,mirror}] (-1.1,-7) -- (-1.1,-8) node {};
\node[meta] at (1,-7.5) {};
\node[meta] at (9,-7.5) {};
\node[meta] at (13,-7.5) {};	
\node at (-0.8,-7) {$P_{10}$};
\node at (-0.8,-7.5) {$P_{11}$};
\node at (-0.8,-8) {$P_{12}$};
% round 10
\foreach \x in {0,1,2,5,8,9,10,12}
	\node[peb] at (\x,-7) {};
\foreach \x in {4,6}
	\node[peb,red!70] at (\x,-7) {};
% round 11
\foreach \x in {0,1,2,8,9,10,12,13}
	\node[peb] at (\x,-7.5) {};
\foreach \x in {4,6}
	\node[peb,red!70] at (\x,-7.5) {};
\foreach \y in {5}
	\node[unpeb] at (\y,-7.5) {};
% round 12
\foreach \x in {0,1,2,8,9,10,12,13,14}
	\node[peb] at (\x,-8) {};
\foreach \x in {4,6}
	\node[peb,red!70] at (\x,-8) {};
\foreach \y in {5}
	\node[unpeb] at (\y,-8) {};
	
%P5
\node at (-1.65,-9.5) {\large\textcolor{forestgreen}{$P'_{5}$}};
	\draw[forestgreen,decorate,decoration={brace,amplitude=5pt,mirror}] (-1.1,-9) -- (-1.1,-10) node {};
\node[meta] at (9,-9.5) {};
\node[meta] at (13,-9.5) {};
\node[meta] at (17,-9.5) {};
\node at (-0.8,-9) {$P_{13}$};
\node at (-0.8,-9.5) {$P_{14}$};
\node at (-0.8,-10) {$P_{15}$};
% round 13
\foreach \x in {0,8,9,10,12,13,14,16}
	\node[peb] at (\x,-9) {};
\foreach \x in {1,4,6}
	\node[peb,red!70] at (\x,-9) {};
\foreach \y in {2,5}
	\node[unpeb] at (\y,-9) {};
% round 14
\foreach \x in {0,8,9,10,12,13,14,16,17}
	\node[peb] at (\x,-9.5) {};
\foreach \x in {1,4,6}
	\node[peb,red!70] at (\x,-9.5) {};
\foreach \y in {2,5}
	\node[unpeb] at (\y,-9.5) {};
% round 15
\foreach \x in {8,9,10,12,13,14,16,17,18}
	\node[peb] at (\x,-10) {};
\foreach \x in {1,4,6}
	\node[peb,red!70] at (\x,-10) {};
\foreach \y in {0,2,5}
	\node[unpeb] at (\y,-10) {};

%P6
\node at (-1.65,-12) {\large\textcolor{forestgreen}{$P'_{6}$}};
	\draw[forestgreen,decorate,decoration={brace,amplitude=5pt,mirror}] (-1.1,-11) -- (-1.1,-13) node {};
\node[meta,minimum height=2.1cm] at (9,-12) {};
\node[meta,minimum height=2.1cm] at (13,-12) {};
\node[meta,minimum height=2.1cm] at (17,-12) {};
\node[meta,minimum height=2.1cm] at (21,-12) {};
\node at (-0.8,-11) {$P_{16}$};
\node at (-0.8,-11.5) {$P_{17}$};
\node at (-0.8,-12) {$P_{18}$};
\node at (-0.8,-12.5) {$P_{19}$};
\node at (-0.8,-13) {$P_{20}$};
% round 16
\foreach \x in {8,9,10,12,13,14,16,17,18,20}
	\node[peb] at (\x,-11) {};
\foreach \x in {1,4,6}
	\node[peb,red!70] at (\x,-11) {};
\foreach \y in {0,2,5}
	\node[unpeb] at (\y,-11) {};
% round 17
\foreach \x in {8,9,10,12,13,14,16,17,18,20,21}
	\node[peb] at (\x,-11.5) {};
\foreach \x in {1,4,6}
	\node[peb,red!70] at (\x,-11.5) {};
\foreach \y in {0,2,5}
	\node[unpeb] at (\y,-11.5) {};
% round 18
\foreach \x in {8,9,10,12,13,14,16,17,18,20,21,22}
	\node[peb] at (\x,-12) {};
\foreach \x in {1,4,6}
	\node[peb,red!70] at (\x,-12) {};
\foreach \y in {0,2,5}
	\node[unpeb] at (\y,-12) {};
% round 19
\foreach \x in {8,9,10,12,13,14,16,17,18,20,22}
	\node[peb] at (\x,-12.5) {};
\foreach \x in {1,4,6}
	\node[peb,red!70] at (\x,-12.5) {};
\foreach \y in {0,2,5,21}
	\node[unpeb] at (\y,-12.5) {};
\node at (18.5,-11) {\large\textcolor{forestgreen}{D}};
\node at (22.5,-11) {\large\textcolor{forestgreen}{A}};
% round 20
\foreach \x in {8,9,10,12,13,14,16,17,18,22}
	\node[peb] at (\x,-13) {};
\foreach \x in {1,4,6}
	\node[peb,red!70] at (\x,-13) {};
\foreach \y in {0,2,5,20,21}
	\node[unpeb] at (\y,-13) {};

%P7	
\node at (-1.65,-14.5) {\large\textcolor{forestgreen}{$P'_{7}$}};
	\draw[forestgreen,decorate,decoration={brace,amplitude=5pt,mirror}] (-1.1,-14) -- (-1.1,-15) node {};
\node[meta] at (9,-14.5) {};
\node[meta] at (13,-14.5) {};
\node[meta] at (1,-14.5) {};
\node[meta] at (21,-14.5) {};	
\node at (-0.8,-14) {$P_{21}$};
\node at (-0.8,-14.5) {$P_{22}$};
\node at (-0.8,-15) {$P_{23}$};
% round 21
\foreach \x in {0,1,8,9,10,12,13,14,16,17,22}
	\node[peb] at (\x,-14) {};
\foreach \x in {4,6}
	\node[peb,red!70] at (\x,-14) {};
\foreach \y in {2,5,18,20,21}
	\node[unpeb] at (\y,-14) {};
\node at (14.5,-14) {\large\textcolor{forestgreen}{D}};
% round 22
\foreach \x in {0,1,8,9,10,12,13,14,16,22}
	\node[peb] at (\x,-14.5) {};
\foreach \x in {4,6}
	\node[peb,red!70] at (\x,-14.5) {};
\foreach \y in {2,5,17,18,20,21}
	\node[unpeb] at (\y,-14.5) {};
% round 23
\foreach \x in {0,1,2,8,9,10,12,13,14,22}
	\node[peb] at (\x,-15) {};
\foreach \x in {4,6}
	\node[peb,red!70] at (\x,-15) {};
\foreach \y in {5,16,17,18,20,21}
	\node[unpeb] at (\y,-15) {};

%P8
\node at (-1.65,-16.5) {\large\textcolor{forestgreen}{$P'_{8}$}};
	\draw[forestgreen,decorate,decoration={brace,amplitude=5pt,mirror}] (-1.1,-16) -- (-1.1,-17) node {};
\node[meta] at (1,-16.5) {};
\node[meta] at (5,-16.5) {};
\node[meta] at (9,-16.5) {};
\node[meta] at (21,-16.5) {};	
\node at (-0.8,-16) {$P_{24}$};
\node at (-0.8,-16.5) {$P_{25}$};
\node at (-0.8,-17) {$P_{26}$};
% round 24
\foreach \x in {0,1,2,4,6,8,9,10,12,13,22}
	\node[peb] at (\x,-16) {};
\foreach \y in {5,14,20,21}
	\node[unpeb] at (\y,-16) {};
% round 25
\foreach \x in {0,1,2,4,5,6,8,9,10,12,22}
	\node[peb] at (\x,-16.5) {};
\foreach \y in {13,14,20,21}
	\node[unpeb] at (\y,-16.5) {};
% round 26
\foreach \x in {0,1,2,4,5,6,8,9,10,22}
	\node[peb] at (\x,-17) {};
\foreach \y in {12,13,14,20,21}
	\node[unpeb] at (\y,-17) {};
\node at (10.5,-16) {\large\textcolor{forestgreen}{D}};
	
%P9
\node at (-1.65,-18.5) {\large\textcolor{forestgreen}{$P'_{9}$}};
	\draw[forestgreen,decorate,decoration={brace,amplitude=5pt,mirror}] (-1.1,-18) -- (-1.1,-19) node {};
\node[meta] at (1,-18.5) {};
\node[meta] at (5,-18.5) {};
\node[meta] at (21,-18.5) {};	
\node at (-0.8,-18) {$P_{27}$};
\node at (-0.8,-18.5) {$P_{28}$};
\node at (-0.8,-19) {$P_{29}$};
% round 27
\foreach \x in {0,1,2,4,5,6,8,9,22}
	\node[peb] at (\x,-18) {};
\foreach \y in {10,20,21}
	\node[unpeb] at (\y,-18) {};
\node at (6.5,-18) {\large\textcolor{forestgreen}{D}};
% round 28
\foreach \x in {0,1,2,4,5,6,8,22}
	\node[peb] at (\x,-18.5) {};
\foreach \y in {9,10,20,21}
	\node[unpeb] at (\y,-18.5) {};
% round 29
\foreach \x in {0,1,2,4,5,6,22}
	\node[peb] at (\x,-19) {};
\foreach \y in {8,9,10,20,21}
	\node[unpeb] at (\y,-19) {};
	
%P10
\node at (-1.65,-20.5) {\large\textcolor{forestgreen}{$P'_{10}$}};
	\draw[forestgreen,decorate,decoration={brace,amplitude=5pt,mirror}] (-1.1,-20) -- (-1.1,-21) node {};
\node[meta] at (1,-20.5) {};
\node[meta] at (21,-20.5) {};
\node at (-0.8,-20) {$P_{30}$};
\node at (-0.8,-20.5) {$P_{31}$};
\node at (-0.8,-21) {$P_{32}$};
% round 30
\foreach \x in {0,1,2,4,5,22}
	\node[peb] at (\x,-20) {};
\foreach \y in {6,20,21}
	\node[unpeb] at (\y,-20) {};
\node at (2.5,-20) {\large\textcolor{forestgreen}{D}};
% round 31
\foreach \x in {0,1,2,4,22}
	\node[peb] at (\x,-20.5) {};
\foreach \y in {5,6,20,21}
	\node[unpeb] at (\y,-20.5) {};	
% round 32
\foreach \x in {0,1,2,22}
	\node[peb] at (\x,-21) {};
\foreach \y in {4,5,6,20,21}
	\node[unpeb] at (\y,-21) {};
	
%P11
\node at (-1.65,-22.5) {\large\textcolor{forestgreen}{$P'_{11}$}};
	\draw[forestgreen,decorate,decoration={brace,amplitude=5pt,mirror}] (-1.1,-22) -- (-1.1,-23) node {};
\node[meta] at (21,-22.5) {};	
\node at (-0.8,-22) {$P_{33}$};
\node at (-0.8,-22.5) {$P_{34}$};
\node at (-0.8,-23) {$P_{35}$};
% round 33
\foreach \x in {0,1,22}
	\node[peb] at (\x,-22) {};
\foreach \y in {2,20,21}
	\node[unpeb] at (\y,-22) {};
% round 34
\foreach \x in {0,22}
	\node[peb] at (\x,-22.5) {};
\foreach \y in {1,2,20,21}
	\node[unpeb] at (\y,-22.5) {};
% round 35
\foreach \x in {22}
	\node[peb] at (\x,-23) {};
\foreach \y in {0,1,2,20,21}
	\node[unpeb] at (\y,-23) {};

  \end{tikzpicture}}
  \caption{A parallel reversible pebbling $P=\{P_1,\ldots,P_{35}\}$ of a DAG $G$ using an induced line graph $L_6$. The (underlying) reversible pebbling for $L_6$, which is $P'=\{P'_1,\ldots,P'_{11}\}$, is shown in \figref{fig:metapebbling-linegraph}. Pebbles colored in red are skip pebbles that cannot be removed until we remove the block of pebbles for the last time, i.e., for each block $B_i$, we keep pebbles on the skip nodes until we reach $P'_j$ with $j=\lastd(P',i)$.}
  \figlab{fig:metapebbling}
\end{figure}
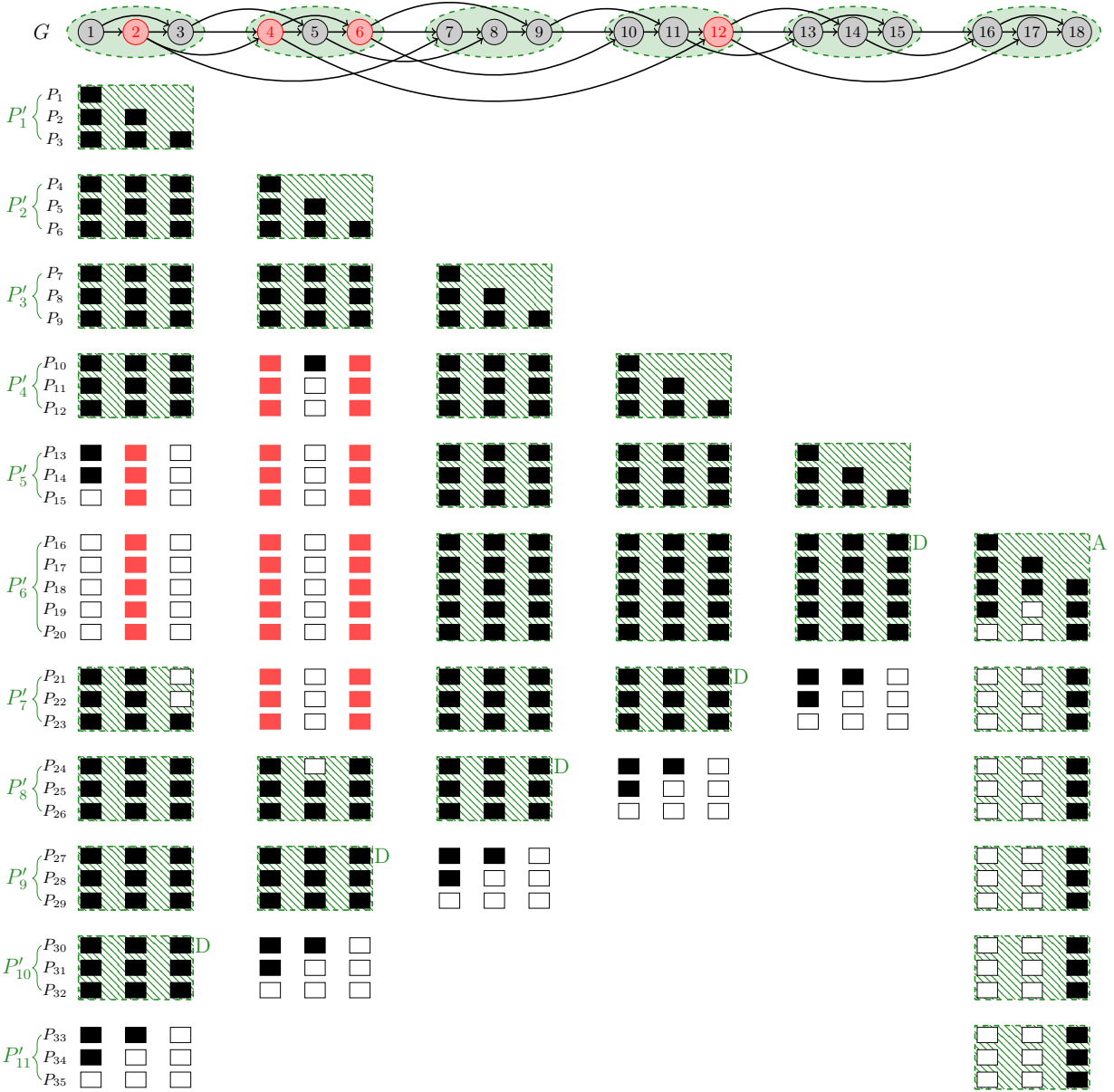

\begin{figure}[ht!]
\centering
  \begin{tikzpicture} 
  [scale=0.88,node distance=1cm,auto,font=\small,
    every node/.style={node distance=2cm},
    peb/.style={rectangle,draw,forestgreen,pattern=north west lines, pattern color=forestgreen,inner sep=2pt, minimum width=0.4cm,minimum height=0.3cm},
    unpeb/.style={rectangle,draw,fill=white,inner sep=2pt, minimum width=0.4cm,minimum height=0.3cm},
    node/.style={circle,draw,forestgreen,fill=forestgreen!20,inner sep=2pt, minimum width=0.5cm}
    ]
\node (LG) at (-0.8,-0.2) {$L_6$};
\node[node] (1) at (0,-0.2) {$1$};
\node[node] (2) at (1,-0.2) {$2$};
\node[node] (3) at (2,-0.2) {$3$};
\node[node] (4) at (3,-0.2) {$4$};
\node[node] (5) at (4,-0.2) {$5$};
\node[node] (6) at (5,-0.2) {$6$};

%\node[node] (10) at (9,0) {$10$};
%\node[node] (11) at (10,0) {$11$};
\path[draw,thick,->] (1) -- (2);
\path[draw,thick,->] (2) -- (3);
\path[draw,thick,->] (3) -- (4);
\path[draw,thick,->] (4) -- (5);
\path[draw,thick,->] (5) -- (6);

%\path[draw,thick,->] (9) -- (10);
%\path[draw,thick,->] (10) -- (11);
% round 1
\node[peb] at (0,-1) {};
\foreach \x in {1,...,5}
	\node[unpeb] at (\x,-1) {};
% round 2
\foreach \x in {0,1}
	\node[peb] at (\x,-1.5) {};
\foreach \y in {2,...,5}
	\node[unpeb] at (\y,-1.5) {};
% round 3
\foreach \x in {0,1,2}
	\node[peb] at (\x,-2) {};
\foreach \y in {3,...,5}
	\node[unpeb] at (\y,-2) {};
% round 4
\foreach \x in {0,2,3}
	\node[peb] at (\x,-2.5) {};
\foreach \y in {1,4,5}
	\node[unpeb] at (\y,-2.5) {};
% round 5
\foreach \x in {2,3,4}
	\node[peb] at (\x,-3) {};
\foreach \y in {0,1,5}
	\node[unpeb] at (\y,-3) {};
% round 6
\foreach \x in {2,...,5}
	\node[peb] at (\x,-3.5) {};
\foreach \y in {0,1}
	\node[unpeb] at (\y,-3.5) {};
\node at (4.37,-3.5) {\footnotesize \textcolor{forestgreen}{D}};
\node at (5.37,-3.5) {\footnotesize \textcolor{forestgreen}{A}};
% round 7
\foreach \x in {0,2,3,5}
	\node[peb] at (\x,-4) {};
\foreach \y in {1,4}
	\node[unpeb] at (\y,-4) {};
\node at (3.37,-4) {\footnotesize \textcolor{forestgreen}{D}};
% round 8
\foreach \x in {0,1,2,5}
	\node[peb] at (\x,-4.5) {};
\foreach \y in {3,4}
	\node[unpeb] at (\y,-4.5) {};
\node at (2.37,-4.5) {\footnotesize \textcolor{forestgreen}{D}};
% round 9
\foreach \x in {0,1,5}
	\node[peb] at (\x,-5) {};
\foreach \y in {2,3,4}
	\node[unpeb] at (\y,-5) {};
\node at (1.37,-5) {\footnotesize \textcolor{forestgreen}{D}};
% round 10
\foreach \x in {0,5}
	\node[peb] at (\x,-5.5) {};
\foreach \y in {1,2,3,4}
	\node[unpeb] at (\y,-5.5) {};
\node at (0.37,-5.5) {\footnotesize \textcolor{forestgreen}{D}};
% round 11
\foreach \x in {5}
	\node[peb] at (\x,-6) {};
\foreach \y in {0,...,4}
	\node[unpeb] at (\y,-6) {};

% dividing chunks
\foreach \x in {2.5}
	\path[draw,thick,dashed] (\x,0) -- (\x,-6.2);
% legends
\foreach \x in {1,...,11}
	\node at (-0.8,-\x/2-0.5) {\textcolor{forestgreen}{$P'_{\x}$}};
  \end{tikzpicture}
\caption{A reversible pebbling for a line graph with $6$ nodes. Note that we mark a pebble on node $i$ in round $j$ with ``D'' if $j=\lastd(P',i)$, and with ``A'' if $j=\lasta(P',i)$.} \figlab{fig:metapebbling-linegraph}
\end{figure}
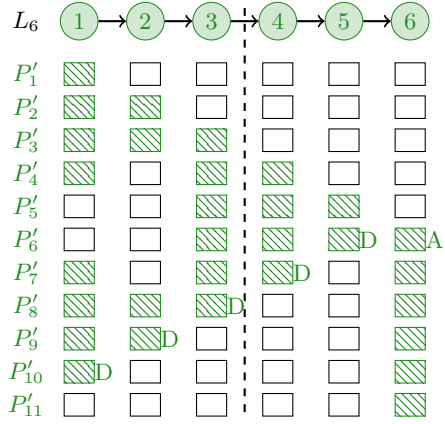

%% file: missingproof.tex
% !TEX root = tcc-main.tex
\allowdisplaybreaks
\section{Missing Proofs}\applab{missing}

\begin{remindertheorem}{\claimref{claim:Pv}}
\claimPv
\end{remindertheorem}

\begin{proofof}{\claimref{claim:Pv}}
We observe that for $v\in[N]$, $P_v\setminus P_{v-1}\subseteq \bigcup_{i=0}^{d} A_{v+i,S,i+1}$, and $P_{v-1}\setminus P_v \subseteq A_{v,S,2}\cup\left(\bigcup_{i=1}^{d+1} A_{v-i,S,i}\right)$. Then by \claimref{claim:AwSi} below, we have 
\begin{align*}
\parents(P_v\setminus P_{v-1},G)\setminus S &\subseteq \bigcup_{i=0}^{d}\parents(A_{v+i,S,i+1},G)\\
&\subseteq \bigcup_{i=0}^{d} A_{v+i,S,i+2}\\
&= \left( \bigcup_{i=0}^{d-1} A_{v+i,S,i+2} \right) \cup A_{v-1+d,S,d+2}\\
&= \bigcup_{i=0}^{d-1} A_{v+i,S,i+2} \subseteq P_{v-1}\cap P_v,
\end{align*}
and
\begin{align*}
\parents(P_{v-1}\setminus P_v,G)\setminus S &\subseteq \parents(A_{v,S,2},G)\cup\left(\bigcup_{i=1}^{d+1}\parents(A_{v-i,S,i},G)\right)\\
&\subseteq A_{v,S,3}\cup\left(\bigcup_{i=1}^{d+1} A_{v-i,S,i+1}\right)\\
&= A_{v,S,3}\cup\left(\bigcup_{i=1}^{d} A_{v-i,S,i+1}\right) \cup A_{v-d,S,d+2}\\
&= A_{v,S,3}\cup\left(\bigcup_{i=1}^{d} A_{v-i,S,i+1}\right) \subseteq P_{v-1}\cap P_v,
\end{align*}
where we have $A_{v-1+d,S,d+2}=A_{w-d,S,d+2}=\varnothing$ by the $(e,d)$-reducibility. Taken together, we have $\parents(P_v\setminus P_{v-1},G)\cup\parents(P_{v-1}\setminus P_v,G)\subseteq S_{\leq v-1}\cup(P_{v-1}\cap P_v) = P_{v-1}\cap P_v$.
\end{proofof}

\begin{claim}\claimlab{claim:AwSi}
$\parents(A_{w,S,i},G)\setminus S \subseteq A_{w,S,i+1}$.
\end{claim}

\begin{proof}
If $x\in A_{w,S,i}$ then by definition we have $\lpath_{G-S_{\leq w-1}}(x,w)=i$. For any $x'\in\parents(x,G)\setminus S$, we observe that $\lpath_{G-S_{\leq w-1}}(x',w)= 1+\lpath_{G-S_{\leq w-1}}(x,w)=i+1$, which completes the proof.\qed
\end{proof}

\begin{remindertheorem}{\lemref{lem:argon2i}}
\argonlemma
\end{remindertheorem}

\begin{proofof}{\lemref{lem:argon2i}}
We divide $N$ nodes into $\lambda$ layers of size $N/\lambda$ and reduce the depth of each layer to $d'$ so that the final depth becomes $d_1=d_2=d'\lambda$ for both Argon2i-A and Argon2i-B. 
To do so, we (a) delete all nodes with parents in the same layer, and (b) delete one out of $d'$ nodes in each layer. 
Let $\mathsf{Delete}_i$ be the event that a node $v$ in $i\th$ layer is deleted in step (a), i.e., $r(v)$ remains in the same layer.
\begin{enumerate}
\item For $G_\ArgonA$, since all the layers have the same number of nodes and $r(v)$ is picked uniformly at random from $[v-2]$, we observe that $\Pr[\mathsf{Delete}_i]\leq\frac{1}{i}$. It is clear that we delete $N/d'$ nodes in step (b). Hence,
\begin{align*}
e_1 &= \frac{N}{d'} + \text{(\# nodes deleted in step (a))}\\
  &= \frac{N}{d'} + \sum_{i=1}^\lambda \Pr[\mathsf{Delete}_i]\cdot\frac{N}{\lambda} \simeq \frac{N}{d'} + \frac{N\ln\lambda}{\lambda}.
\end{align*}
\item For $G_\ArgonB$, since we have $i\left(1-\frac{x^2}{N^2}\right)\in(j-1,j]$ if and only if $N\sqrt{1-\frac{j}{i}}\leq x<N\sqrt{1-\frac{j-1}{i}}$, we have that $\Pr[r(i)=j]=\sqrt{1-\frac{j-1}{i}}-\sqrt{1-\frac{j}{i}}$. Similarly, we have $\Pr[a< r(i)<b]=\Pr_{x\in[N]}\left[i\left(1-\frac{x^2}{N^2}\right)\in(a,b-1]\right] = \sqrt{1-\frac{a}{i}}-\sqrt{1-\frac{b-1}{i}}$. Thus,
\begin{align*}
\Pr[\mathsf{Delete}_i] &= \Pr\left[\frac{(i-1)N}{\lambda}<r(v)<v\right]\\
&= \sqrt{1-\frac{(i-1)N/\lambda}{v}}-\sqrt{1-\frac{v-1}{v}}\\
&= \sqrt{1-\frac{(i-1)N}{\lambda v}}-\sqrt{\frac{1}{v}}\\
&\leq \sqrt{1-\frac{i-1}{i}}-\sqrt{\frac{\lambda}{iN}} = \sqrt{\frac{1}{i}}-\sqrt{\frac{\lambda}{iN}},
\end{align*}
where the last inequality holds since $\sqrt{1-\frac{(i-1)N}{\lambda v}}-\sqrt{\frac{1}{v}}$ is an increasing function of $v$ and the largest possible $v$ is $iN/\lambda$ since it should lie in the $i\th$ layer. Hence,
\begin{align*}
e_2 &= \frac{N}{d'} + \sum_{i=1}^\lambda \Pr[\mathsf{Delete}_i]\cdot\frac{N}{\lambda}\\
&\leq \frac{N}{d'} + \left(\frac{N}{\lambda}-\sqrt{\frac{N}{\lambda}}\right)\sum_{i=1}^{\lambda}\sqrt{\frac{1}{i}}\\
&\leq \frac{N}{d'} +\left(\frac{N}{\lambda}-\sqrt{\frac{N}{\lambda}}\right)\left(\int_1^\lambda \frac{dx}{\sqrt{x}} + 1\right)\\
&= \frac{N}{d'} + \left(\frac{N}{\lambda}-\sqrt{\frac{N}{\lambda}}\right)(2\sqrt{\lambda}-1) \leq \frac{N}{d'} + \frac{2N}{\sqrt{\lambda}}.
\end{align*}
\end{enumerate}
\end{proofof}

\begin{remindertheorem}{\corref{cor:argon2i}}
\argoncorollary
\end{remindertheorem}

\begin{proofof}{\corref{cor:argon2i}}
From \thmref{thm:ed-reducible} and \lemref{lem:argon2i}, we have
\[ \pqpeb_\Cspacetime(G_\ArgonA) \leq \O{N+Ne+Nd2^d} \simeq \O{ N+\frac{N^2}{d'}+\frac{N^2\ln\lambda}{\lambda} + \lambda d' 2^{\lambda d'}N }. \]
To make the upper bound optimal, we want to make the upper bound as small as possible. Hence, we want to find $d'$ and $\lambda$ such that $\frac{N^2}{d'}\approx\frac{N^2\ln\lambda}{\lambda} \approx \lambda d' 2^{\lambda d'}N$ as much as possible. Hence, $d'=\frac{\lambda}{\ln\lambda}$ and $\lambda$ should satisfy $\frac{\lambda^3}{(\ln\lambda)^2}2^{\lambda^2/\ln\lambda}\approx N$. Setting $\lambda=\sqrt{\log N}$, we have $d'=\frac{\lambda}{\ln\lambda}=\frac{2\sqrt{\log N}}{\ln\log N}$ and $d=d'\lambda = \frac{2\log N}{\ln\log N}$. Thus,
\begin{align*}
\pqpeb_\Cspacetime(G_\ArgonA) &\leq \O{N + \frac{2N^2\ln\log N}{2\sqrt{\log N}} + \frac{2N\log N}{\ln\log N}2^{2\log N/\ln\log N} } \\
&= \O{N + \frac{2N^2\ln\log N}{2\sqrt{\log N}} + \frac{2N^{1+\frac{2}{\ln\log N}}\log N}{\ln\log N} }\\
&= \O{\frac{N^2\log\log N}{\sqrt{\log N}}},
\end{align*}
since $\ln x = (\ln2)(\log x)$ for any $x>0$.

For Argon2i-B, we have
\[ \pqpeb_\Cspacetime(G_\ArgonB) \leq \O{N+Ne+Nd2^d} \simeq \O{ N+\frac{N^2}{d'}+\frac{2N^2}{\sqrt{\lambda}} + \lambda d' 2^{\lambda d'}N }. \]
Similarly, to make the upper bound optimal, we want to make $\frac{N^2}{d'}\approx\frac{2N^2}{\sqrt{\lambda}}\approx\lambda d' 2^{\lambda d'}N$ as much as possible. Hence, we have $d'\approx\sqrt{\lambda}/2$ and plugging in $\lambda=\sqrt[3]{\log^2 N}$ and $d'=\sqrt[3]{\log N}/2$, we have
\begin{align*}
\pqpeb_\Cspacetime(G_\ArgonB) &\leq \O{N+\frac{4N^2}{\sqrt[3]{\log N}}+\frac{N\sqrt{N}\log N}{2}}\\
&=\O{\frac{N^2}{\sqrt[3]{\log N}}}.
\end{align*}
\end{proofof}

\begin{remindertheorem}{\lemref{lem:translegal}}
\translegal
\end{remindertheorem}

\begin{proofof}{\lemref{lem:translegal}}
We want to show that it satisfies conditions in \defref{def:quantum-pebbling}.\\[5pt]
\noindent \ul{Conditions (1) and (5): $P_{tb+N-(\lceil N/b\rceil-1)b-1}=\{N\}$.}
\begin{itemize}
\item It is clear by construction because we remove all nodes except for the target node $N$.
\end{itemize}

\noindent \ul{Condition (2): $\forall j\in[tb+N-(\lceil N/b\rceil-1)b-1] : v\in (P_j\setminus P_{j-1}) \Rightarrow \parents(v,G)\subseteq P_{j-1}$.}
\begin{itemize}
\item We first observe that whenever we pebble a new node $w$ in $L_{\lceil N/b\rceil}$, the node $w-1$ must have been pebbled in the previous round.
\item Suppose that $v\in B_w$ for some $w\in[\lceil N/b\rceil]$. For every edge of the form $(u,v)$, we have the following possibilities:
\begin{enumerate}
\item[(a)] If $u\in B_w$, $u$ must be (re)pebbled before node $v$ since both $u$ and $v$ corresponds to placing the node $w$ in $L_{\lceil N/b\rceil}$. Hence, $u\in P_{j-1}$.
\item[(b)] If $u\in B_{w-1}$, we are guaranteed that $u$ is already pebbled before we begin pebbling nodes in block $B_w$ since every node in $B_{w-1}$ is pebbled. Hence, $u\in P_{j-1}$.
\item[(c)] If $u\in B_j$ with $j<w-1$, then $u$ is a skip node and will already be pebbled before placing a pebble on $v$. Hence, $u\in P_{j-1}$.
\end{enumerate}
\item Taken together, we have $\parents(v,G)\subseteq P_{j-1}$.
\end{itemize}

\noindent \ul{Condition (3): $\forall j\in[tb+N-(\lceil N/b\rceil-1)b-1] : v\in (P_{j-1}\setminus P_j) \Rightarrow \parents(v,G)\subseteq P_{j-1}$.}
\begin{itemize}
\item We first observe that whenever we remove a pebble from $w$ in $L_{\lceil N/b\rceil}$, the node $w-1$ must have been pebbled in the previous round.
\item Suppose that $v\in B_w$ for some $w[\in\lceil N/b\rceil]$. For every edge of the form $(u,v)$, we have the following possibilities:
\begin{enumerate}
\item[(a)] If $u\in B_w$, a pebble on $u$ is not yet removed in the previous round because we remove pebbles in $B_w$ in a reverse topological order. Hence, $u\in P_{j-1}$.
\item[(b)] If $u\in B_{w-1}$, we are guaranteed that $u$ is already pebbled before we begin removing nodes in block $B_w$ since every node in $B_{w-1}$ is pebbled. Hence, $u\in P_{j-1}$.
\item[(c)] If $u\in B_j$ with $j<w-1$, then $u$ is a skip node and will already be pebbled before removing a pebble from $v$. Hence, $u\in P_{j-1}$.
\end{enumerate}
\item Taken together, we have $\parents(v,G)\subseteq P_{j-1}$.
\end{itemize}

\noindent \ul{Condition (4): $\forall j\in[tb+N-(\lceil N/b\rceil-1)b-1]:v\in \parents(P_j\setminus P_{j-1},G)\cup\parents(P_{j-1}\setminus P_j,G)$, then $v\in P_j$.}
\begin{itemize}
\item If $v\in\parents(P_j\setminus P_{j-1},G)$, then there exists some $v'\in P_j\setminus P_{j-1}$ and some $w\in[\lceil N/b\rceil]$ such that $(v,v')\in E$ and $v'\in B_w$. Now we have the following possibilities:
\begin{enumerate}
\item[(a)] If $v\in B_w$, then $v$ must be (re)pebbled before node $v'$ and keep pebbled since both $u$ and $v$ corresponds to placing the node $w$ in $L_{\lceil N/b\rceil}$. Hence, $v\in P_j$.
\item[(b)] If $v\in B_{w-1}$, we are guaranteed that $v$ is already pebbled when we begin pebbling nodes in block $B_w$ since every node in $B_{w-1}$ is pebbled. Hence, $v\in P_j$.
\item[(c)] If $v\in B_j$ with $j<w-1$, then $v$ is a skip node and will already be pebbled and keep pebbled when placing a pebble on $v'$. Hence, $v\in P_j$.
\end{enumerate}
\item If $v\in\parents(P_{j-1}\setminus P_j,G)$, then there exists some $v''\in P_{j-1}\setminus P_j$ and some $w'\in[\lceil N/b\rceil]$ such that $(v,v'')\in E$ and $v''\in B_{w'}$. Now we have the following possibilities:
\begin{enumerate}
\item[(a)] If $v\in B_{w'}$, a pebble on $v$ is not yet removed in $P_j$ because we remove pebbles in $B_{w'}$ in a reverse topological order. Hence, $v\in P_j$.
\item[(b)] If $v\in B_{w'-1}$, we are guaranteed that $v$ is already pebbled when we begin removing nodes in block $B_{w'}$ since every node in $B_{w'-1}$ is pebbled. Hence, $v\in P_j$.
\item[(c)] If $v\in B_j$ with $j<w'-1$, then $v$ is a skip node and will already be pebbled and keep pebbled when removing a pebble from $v''$. Hence, $v\in P_j$.
\end{enumerate}
\end{itemize}
Taken together, we can conclude that if $P'\in\pqPeb{L_{\lceil N/b\rceil}}$, then $P=\trans(G,P',b)\in\pqPeb{G}$.
\end{proofof}

\begin{remindertheorem}{\lemref{lem:skipnodesDRS}}
\skipnodesDRS
\end{remindertheorem}

\begin{proofof}{\lemref{lem:skipnodesDRS}}
For each $v\in V_\DRS$, let $Y_v$ be an indicator random variable for the event that $v-r(v) > b$. Then we observe that $\numskip(G_\DRS,b)\leq \sum_{v\in V_\DRS} Y_v$, since $\numskip(G_\DRS,b)$ is upper bounded by the number of edges that skip over a block. Since there are at most $\log v$ buckets for $r(v)$ and $\log b$ buckets with $v-r(v)\leq b$, we have $\Pr[v-r(v)>b]\leq 1-\frac{\log b}{\log v} \leq 1-\frac{\log b}{\log N} = \frac{\log(N/b)}{\log N}$. Hence, by linearity of expectation it follows that
\[ \Ex{\numskip(G_\DRS,b)}\leq \sum_{v\in V_\DRS}\Ex{Y_v} = \sum_{v\in V_\DRS} \Pr[v-r(v)>b] \leq \sum_{v\in V_\DRS} \frac{\log(N/b)}{\log N} = \frac{N\log(N/b)}{\log N}. \]
As the expected value is the sum of independent random variables, we can use Chernoff bounds with $\mu=\frac{N\log(N/b)}{\log N}\geq \sum_{v\in V_\DRS}\Ex{Y_v}$ to show that for any constant $\delta>0$, we have
\[ \Pr[\numskip(G_\DRS,b)>(1+\delta)\mu] < \exp\left( -\frac{\delta^2 N\log(N/b)}{3\log N} \right).\]
Hence, with high probability, we have $\numskip(G_\DRS,b)=\O{\frac{N\log(N/b)}{\log N}}$. Setting $b=\frac{N}{\log^2N}$, we get the desired result.
\end{proofof}

\begin{remindertheorem}{\lemref{reverseMerge}}
\reverseMerge
\end{remindertheorem}
\begin{proof}
    First we'll show that $\langle P'_1\cup T, \dots, P'_{t'-1}\cup T,P'_{t'}\rangle$ is a legal reversible pebbling. See that since $P'$ satisfies requirements $(2)$, $(3)$, and $(4)$, so does $\langle P_1'\cup T,\dots, P'_{t'-1}\cup T, P'_{t'}\rangle$ since no pebbles are removed. Since $P'_t=P'_{t'}$ we have that $\langle P_1,\dots, P_t, P'_{t'-1}\cup T, P'_{t'-2}\cup T,\dots, P'_{1}\cup T\rangle$ is a legal reversible pebbling sequence.\qed
    % By construction $\langle P_1,\dots, P_t\rangle$ is legal, and by Lemma \ref{reverse} is \[\langle P'_{t'-1}\cup T,\dots, P'_{1}\cup T\rangle\] is legal as well. Now, it suffices to show that $\langle P_t,P'_{t'-1}\cup T\rangle$ is legal. This is true since $P_t=P'_{t'}$, meaning $\langle P_{t'},P'_{t'-1}\cup T\rangle$. Since $\langle P'_{t'-1}\cup T,\dots, P'_{1}\cup T\rangle$ is a legal reversible pebbling sequence, so is $\langle P_{t'},P'_{t'-1}\cup T\rangle$.  
\end{proof}
\begin{remindertheorem}{\lemref{monorev}}
    \monorev
\end{remindertheorem}
 \begin{proof}
            Since $P$ is a legal (standard parallel) pebbling and no nodes are deleted, then it suffices to show reversibility. Suppose $x\in\parents(P_i\setminus P_{i-1})\cup\parents(P_{i-1}\setminus P_i)$. Since \[\parents(P_i\setminus P_{i-1})\cup\parents(P_{i-1}\setminus P_i)=\parents(P_i\setminus P_{i-1}),\] $x\in P_i$ as it was pebbled in some prior pebbling step and $P$ never removes any pebbles.\qed
        \end{proof}

\begin{remindertheorem}{\lemref{revpebLegal}}
    \revpebLegal
\end{remindertheorem}
\begin{proof}
    In the discussion above we've shown that each $\mathcal P^c$ and $\mathcal Q^c$ are reversible pebblings, so for the first half, it suffices to show that $\langle P_{cg}, Q^{c,1}_{1}\rangle$ and $\langle Q^{c,2}_{1}, P_{cg+1}\rangle$ are legal reversible pebblings. Since $Q^{c,1}_1=P_{cg}\cup R(P_{cg})$, it is a legal monotonic (and thus reversible) sequence. Since $Q^{c,2}_0=P_{cg+1}$ (recall we defined $Q^c_d$ but didn't include it in $\mathcal Q^c$), $\langle Q^{c,2}_{1}, P_{cg+1}\rangle$ is also monotonic and thus reversible. Since $\pebrev^1$ is a reversible pebbling so is $\pebrev$ by \lemref{reverseMerge}. \qed
\end{proof}

%% file: algorithms.tex
% !TEX root = tcc-main.tex
\section{Reversible Pebbling Strategy using an Induced Line Graph}\applab{algorithms}

\begin{algorithm}[ht!]
\caption{The Procedure $\trans(G,P',b)$.}
\alglab{alg:trans}
\DontPrintSemicolon
\SetAlgoLined
\KwIn{A constant-indegree DAG $G=(V=[N],E)$, a parameter $b$ (size of the block), and a legal reversible pebbling $P'=\{P'_0,P'_1,\ldots,P'_t\}\in\pqPeb{L_{\lceil N/b\rceil}}$ for an induced line graph $L_{\lceil N/b\rceil}$}
\KwOut{A legal reversible pebbling $P\in\pqPeb{G}$ of $G$}
\vspace{0.3em}
Partition $V=[N]$ into $B_1,\ldots,B_{\lceil N/b\rceil}$ where $B_i=\{(i-1)b+1,(i-1)b+2,\ldots,ib\}$ for $i\in[\lceil N/b\rceil-1]$ and $B_{\lceil N/b\rceil} = \{(\lceil N/b\rceil-1)b+1,(\lceil N/b\rceil-1)b+2,\ldots,N\}$.\;
Initialize $P^{(i)}_{0,b}=\varnothing$ and $P^{(i)}_{j,k}=\varnothing$ for each $i\in[\lceil N/b\rceil]$, $j\in[t]$, and $k\in[f(j)]$, where $f(j)=b+N-\left(\lceil \frac{N}{b}\rceil-1\right)b-1$ if $j=\lasta(P',\lceil N/b\rceil)$, and $f(j)=b$ elsewhere.\;
\For(\tcp*[f]{for each block $B_1,\ldots, B_{\lceil N/b\rceil-1}$ except for the last one}){$i=1,\ldots,\lceil N/b\rceil-1$}{
	Compute $S_i\coloneqq\skipnode(B_i,G)$ using \eqnref{skipnodes}.\;
	\For(\tcp*[f]{for each round in $P'$}){$j=1,\ldots,t$}{
		\If{$j\neq \lasta(P',\lceil N/b\rceil)$}{
			$\{P^{(i)}_{j,1},\ldots,P^{(i)}_{j,b}\}\leftarrow\blockpebble(B_i,b,S_i,P',P^{(i)}_{j-1,f(j)},i,j)$.
		}
		\Else(\tcp*[f]{i.e., $j = \lasta(P',\lceil N/b\rceil)$}){
			$\{P^{(i)}_{j,1},\ldots,P^{(i)}_{j,b}\}\leftarrow\blockpebble(B_i,b,S_i,P',P^{(i)}_{j-1,f(j)},i,j)$.\;
			Maintain pebbles for the extra $N-\left(\lceil \frac{N}{b}\rceil-1\right)b-1\leq b-1$ steps, i.e., $P^{(i)}_{j,b} = P^{(i)}_{j,b+1} = \cdots = P^{(i)}_{j,b+N-\left(\lceil \frac{N}{b}\rceil-1\right)b-1}$.
		}
	}
}
\For(\tcp*[f]{for the last block $B_{\lceil N/b\rceil}$ and for each round in $P'$}){$j=1,\ldots,t$}{
	\If{$j \neq \lasta(P',\lceil N/b\rceil)$}{
		$\{P^{(\lceil \frac{N}{b}\rceil)}_{j,1},\ldots,P^{(\lceil \frac{N}{b}\rceil)}_{j,b}\}\leftarrow\lastblockpebble(N,b,P',P^{(\lceil \frac{N}{b}\rceil)}_{j-1,f(j)},j)$.
	}
	\Else(\tcp*[f]{i.e., $j = \lasta(P',\lceil N/b\rceil)$}){
		$\{P^{(\lceil \frac{N}{b}\rceil)}_{j,1},\ldots,P^{(\lceil \frac{N}{b}\rceil)}_{j,b}\}\leftarrow\lastblockpebble(N,b,P',P^{(\lceil \frac{N}{b}\rceil)}_{j-1,f(j)},j)$.\;
		Delete pebbles from the block in a reverse topological order, except for the sink node, with $N-(\lceil N/b\rceil-1)b-1$ steps, i.e., $P^{(\lceil \frac{N}{b}\rceil)}_{b+k} = P^{(\lceil \frac{N}{b}\rceil)}_{b+k-1} \setminus\{N-k\}$ for $k=1,\ldots,N-(\lceil N/b\rceil-1)b-1$.
	}
}
\For{$j=1,\ldots,t$}{
	\For{$k=1,\ldots,f(j)$}{
		$P_{j,k} = \bigcup_{i=1}^{\lceil N/b\rceil}P^{(i)}_{j,k}$.\;
		\If(\tcp*[f]{Ordering the pebbling configurations}){$j\leq\lasta(P',\lceil N/b\rceil)$}{
			$P_{(j-1)b+k} \leftarrow P_{j,k}$	
		}
		\Else{
			$P_{N-(\lceil N/b \rceil-1)b-1+(j-1)b+k} \leftarrow P_{j,k}$
		}
	}
}
\Return{$P=\{P_1,\ldots,P_{tb+N-(\lceil N/b\rceil-1)b-1}\}$.}
\end{algorithm}

\begin{algorithm}[ht!]
\caption{The Subfunction $\blockpebble(B,b,S,P',P_0,i,j)$.}
\alglab{alg:block}
\DontPrintSemicolon
\SetAlgoLined
\KwIn{A set of nodes $B$, a parameter $b$ (size of the set), a set of skip pebbles $S$, a legal reversible pebbling $P'=\{P'_0,P'_1,\ldots,P'_t\}$, a pebbling configuration $P_0$ on $B$, and parameters $i$ and $j$}
\KwOut{A legal relaxed reversible pebbling $P=\{P_1,\ldots,P_b\}$ of the set $B$}
\vspace{0.3em}
Assert $|B|=b$.\;
\If{$i\in P_j'\setminus P_{j-1}'$}{
	Place pebbles in the block $B$ with $b$ steps, i.e., $P_1 = P_0 \cup \{(i-1)b+1\}$, and $P_k = P_{k-1} \cup \{(i-1)b+k\}$ for $k=2,\ldots,b$.
}
\ElseIf{$i\in P_{j-1}'\setminus P_j'$}{
	\If{$j-1=\lastd(P',i)$}{
		Delete pebbles from the block $B$ in a reverse topological order with $b$ steps, i.e., $P_1 = P_0 \setminus \{ib\}$, and $P_k = P_{k-1} \setminus \{ib-(k-1)\}$ for $k=2,\ldots,b$.
	}
	\Else{
		Delete pebbles from the block $B$ \emph{except for the skip nodes}, i.e., $P_1 = P_0\setminus (\{ib\}\setminus S)$, and $P_k = P_{k-1}\setminus (\{ib-(k-1)\}\setminus S)$ for $k=2,\ldots,b$.
	}
}
\Else{
	Maintain pebbles in the block $B$ for $b$ steps, i.e., $P_0 = P_1 = \cdots = P_b$.
}
\Return{$P=\{P_1,\ldots,P_b\}$}
\end{algorithm}

\begin{algorithm}[ht!]
\caption{The Subfunction $\lastblockpebble(N,b,P',P_0,j)$.}
\alglab{alg:lastblockpebble}
\DontPrintSemicolon
\SetAlgoLined
\KwIn{A parameter $N$, $b$, a legal reversible pebbling $P'=\{P'_0,P'_1,\ldots,P'_t\}$, a pebbling configuration $P_0$ of the last block, and a parameter $j$}
\KwOut{A legal relaxed reversible pebbling $P=\{P_1,\ldots,P_b\}$ of the last block}
\vspace{0.3em}
\If{$\lceil N/b\rceil\in P_j'\setminus P_{j-1}'$}{
	Place pebbles in the block with $N-(\lceil N/b\rceil-1)b$ steps, and maintain the status for the next $b-N+(\lceil N/b\rceil-1)b$ steps, i.e., $P_1 = P_0 \cup \{(\lceil N/b\rceil-1)b+1\}$, $P_k = P_{k-1} \cup \{(\lceil N/b\rceil-1)b+k\}$ for $k=2,\ldots,N-(\lceil N/b\rceil-1)b$, and $P_{N-(\lceil N/b\rceil-1)b}=P_{N-(\lceil N/b\rceil-1)b+1}=\cdots=P_b$.
}
\ElseIf{$\lceil N/b\rceil\in P_{j-1}'\setminus P_j'$}{
	Delete pebbles from the block in a reverse topological order with $N-(\lceil N/b\rceil-1)b$ steps, and maintain the status for the next $b-N+(\lceil N/b\rceil-1)b$ steps, i.e., $P_1 = P_0 \setminus \{N\}$, $P_k = P_{k-1} \setminus \{N-(k-1)\}$ for $k=2,\ldots,N-(\lceil N/b\rceil-1)b$, and $P_{N-(\lceil N/b\rceil-1)b}=P_{N-(\lceil N/b\rceil-1)b+1}=\cdots=P_b$.
}
\Else{
	Maintain pebbles in the block for $b$ steps, i.e., $P_0 = P_1 = \cdots = P_b$.
}
\Return{$P=\{P_1,\ldots,P_b\}$}
\end{algorithm}

%% file: full-main.bbl
\newcommand{\etalchar}[1]{$^{#1}$}
\begin{thebibliography}{BHK{\etalchar{+}}19}

\bibitem[AB16]{C:AlwBlo16}
Jo{\"e}l Alwen and Jeremiah Blocki.
\newblock Efficiently computing data-independent memory-hard functions.
\newblock In Matthew Robshaw and Jonathan Katz, editors, {\em CRYPTO~2016,
  Part~II}, volume 9815 of {\em {LNCS}}, pages 241--271. Springer, Heidelberg,
  August 2016.

\bibitem[AB17]{ESP:AlwBlo17}
Jo{\"e}l Alwen and Jeremiah Blocki.
\newblock Towards practical attacks on argon2i and balloon hashing.
\newblock In {\em Security and Privacy (EuroS\&P), 2017 IEEE European Symposium
  on}, pages 142--157. IEEE, 2017.

\bibitem[ABH17]{CCS:AlwBloHar17}
Jo{\"e}l Alwen, Jeremiah Blocki, and Ben Harsha.
\newblock Practical graphs for optimal side-channel resistant memory-hard
  functions.
\newblock In Bhavani~M. Thuraisingham, David Evans, Tal Malkin, and Dongyan Xu,
  editors, {\em ACM CCS 2017}, pages 1001--1017. {ACM} Press,
  October~/~November 2017.

\bibitem[ABP17]{EC:AlwBloPie17}
Jo{\"e}l Alwen, Jeremiah Blocki, and Krzysztof Pietrzak.
\newblock Depth-robust graphs and their cumulative memory complexity.
\newblock In Jean-S{\'{e}}bastien Coron and Jesper~Buus Nielsen, editors, {\em
  EUROCRYPT~2017, Part~III}, volume 10212 of {\em {LNCS}}, pages 3--32.
  Springer, Heidelberg, April~/~May 2017.

\bibitem[ABP18]{EC:AlwBloPie18}
Jo{\"e}l Alwen, Jeremiah Blocki, and Krzysztof Pietrzak.
\newblock Sustained space complexity.
\newblock In Jesper~Buus Nielsen and Vincent Rijmen, editors, {\em
  EUROCRYPT~2018, Part~II}, volume 10821 of {\em {LNCS}}, pages 99--130.
  Springer, Heidelberg, April~/~May 2018.

\bibitem[AS15]{STOC:AlwSer15}
Jo{\"e}l Alwen and Vladimir Serbinenko.
\newblock High parallel complexity graphs and memory-hard functions.
\newblock In Rocco~A. Servedio and Ronitt Rubinfeld, editors, {\em 47th ACM
  STOC}, pages 595--603. {ACM} Press, June 2015.

\bibitem[AT17]{TCC:AlwTac17}
Jo{\"e}l Alwen and Bj{\"o}rn Tackmann.
\newblock Moderately hard functions: Definition, instantiations, and
  applications.
\newblock In Yael Kalai and Leonid Reyzin, editors, {\em TCC~2017, Part~I},
  volume 10677 of {\em {LNCS}}, pages 493--526. Springer, Heidelberg, November
  2017.

\bibitem[BBBV97]{BennettBBV97}
Charles~H. Bennett, Ethan Bernstein, Gilles Brassard, and Umesh~V. Vazirani.
\newblock Strengths and weaknesses of quantum computing.
\newblock {\em {SIAM} J. Comput.}, 26(5):1510--1523, 1997.

\bibitem[BCS16]{AC:BonCorSch16}
Dan Boneh, Henry {Corrigan-Gibbs}, and Stuart~E. Schechter.
\newblock Balloon hashing: {A} memory-hard function providing provable
  protection against sequential attacks.
\newblock In Jung~Hee Cheon and Tsuyoshi Takagi, editors, {\em ASIACRYPT~2016,
  Part~I}, volume 10031 of {\em {LNCS}}, pages 220--248. Springer, Heidelberg,
  December 2016.

\bibitem[BDF{\etalchar{+}}11]{AC:BDFLSZ11}
Dan Boneh, {\"O}zg{\"u}r Dagdelen, Marc Fischlin, Anja Lehmann, Christian
  Schaffner, and Mark Zhandry.
\newblock Random oracles in a quantum world.
\newblock In Dong~Hoon Lee and Xiaoyun Wang, editors, {\em ASIACRYPT~2011},
  volume 7073 of {\em {LNCS}}, pages 41--69. Springer, Heidelberg, December
  2011.

\bibitem[BDK15]{BDK15}
Alex Biryukov, Daniel Dinu, and Dmitry Khovratovich.
\newblock Fast and tradeoff-resilient memory-hard functions for
  cryptocurrencies and password hashing.
\newblock Cryptology ePrint Archive, Paper 2015/430, 2015.
\newblock \url{https://eprint.iacr.org/2015/430}.

\bibitem[BDKJ16]{BDKJ16}
Alex Biryukov, Daniel Dinu, Dmitry Khovratovich, and Simon Josefsson.
\newblock The memory-hard argon2 password hash and proof-of-work function.
\newblock In {\em Internet-Draft draft-irtf-cfrg-argon2-00, Internet
  Engineering Task Force}, 2016.

\bibitem[Ben89]{Bennett89}
Charles~H. Bennett.
\newblock Time/space trade-offs for reversible computation.
\newblock {\em SIAM J. Comput.}, 18(4):766–776, aug 1989.

\bibitem[BHK{\etalchar{+}}19]{C:BHKLXZ19}
Jeremiah Blocki, Benjamin Harsha, Siteng Kang, Seunghoon Lee, Lu~Xing, and
  Samson Zhou.
\newblock Data-independent memory hard functions: New attacks and stronger
  constructions.
\newblock In Alexandra Boldyreva and Daniele Micciancio, editors, {\em
  CRYPTO~2019, Part~II}, volume 11693 of {\em {LNCS}}, pages 573--607.
  Springer, Heidelberg, August 2019.

\bibitem[BHZ18]{SP:BloHarZho18}
Jeremiah Blocki, Benjamin Harsha, and Samson Zhou.
\newblock On the economics of offline password cracking.
\newblock In {\em 2018 {IEEE} Symposium on Security and Privacy}, pages
  853--871. {IEEE} Computer Society Press, May 2018.

\bibitem[BLZ21]{ITC:BloLeeZho21}
Jeremiah Blocki, Seunghoon Lee, and Samson Zhou.
\newblock {On the Security of Proofs of Sequential Work in a Post-Quantum
  World}.
\newblock In Stefano Tessaro, editor, {\em 2nd Conference on
  Information-Theoretic Cryptography (ITC 2021)}, volume 199 of {\em Leibniz
  International Proceedings in Informatics (LIPIcs)}, pages 22:1--22:27,
  Dagstuhl, Germany, 2021. Schloss Dagstuhl -- Leibniz-Zentrum f{\"u}r
  Informatik.

\bibitem[BZ17]{TCC:BloZho17}
Jeremiah Blocki and Samson Zhou.
\newblock On the depth-robustness and cumulative pebbling cost of {Argon2i}.
\newblock In Yael Kalai and Leonid Reyzin, editors, {\em TCC~2017, Part~I},
  volume 10677 of {\em {LNCS}}, pages 445--465. Springer, Heidelberg, November
  2017.

\bibitem[Cob66]{Cob66}
Alan Cobham.
\newblock The recognition problem for the set of perfect squares.
\newblock In {\em 7th Annual Symposium on Switching and Automata Theory (swat
  1966)}, pages 78--87, 1966.

\bibitem[Coo73]{Coo73}
Stephen~A. Cook.
\newblock An observation on time-storage trade off.
\newblock In {\em Proceedings of the Fifth Annual ACM Symposium on Theory of
  Computing}, STOC '73, page 29–33, New York, NY, USA, 1973. Association for
  Computing Machinery.

\bibitem[Div00]{Divincenzo00thephysical}
David~P. Divincenzo.
\newblock The physical implementation of quantum computation.
\newblock {\em Fortschr. Phys}, 48:2000, 2000.

\bibitem[EGS75]{EGS75}
P.~Erdös, R.L. Graham, and E.~Szemerédi.
\newblock On sparse graphs with dense long paths.
\newblock {\em Computers \& Mathematics with Applications}, 1(3):365 -- 369,
  1975.

\bibitem[FLW14]{AC:ForLucWen14}
Christian Forler, Stefan Lucks, and Jakob Wenzel.
\newblock Memory-demanding password scrambling.
\newblock In Palash Sarkar and Tetsu Iwata, editors, {\em ASIACRYPT~2014,
  Part~II}, volume 8874 of {\em {LNCS}}, pages 289--305. Springer, Heidelberg,
  December 2014.

\bibitem[FR21]{STOC:FefRem21}
Bill Fefferman and Zachary Remscrim.
\newblock Eliminating intermediate measurements in space-bounded quantum
  computation.
\newblock In {\em Proceedings of the 53rd Annual ACM SIGACT Symposium on Theory
  of Computing}, STOC 2021, page 1343–1356, New York, NY, USA, 2021.
  Association for Computing Machinery.

\bibitem[GNP{\etalchar{+}}17]{NIST17}
Paul Grassi, Elaine Newton, Ray Perlner, Andrew Regenscheid, William Burr,
  Justin Richer, Naomi Lefkovitz, Jamie Danker, Yee-Yin Choong, Kristen Greene,
  and Mary Theofanos.
\newblock Digital identity guidelines: Authentication and lifecycle management,
  2017-06-22 2017.

\bibitem[Gro96]{STOC:Grover96}
Lov~K. Grover.
\newblock A fast quantum mechanical algorithm for database search.
\newblock In {\em 28th ACM STOC}, pages 212--219. {ACM} Press, May 1996.

\bibitem[HPV77]{HPV77}
John Hopcroft, Wolfgang Paul, and Leslie Valiant.
\newblock On time versus space.
\newblock {\em J. ACM}, 24(2):332–337, April 1977.

\bibitem[Kal00]{rfc2898}
Burt Kaliski.
\newblock {PKCS \#5: Password-Based Cryptography Specification Version 2.0}.
\newblock {RFC} 2898, RSA Laboratories, September 2000.

\bibitem[KPB00]{PatBra00}
A.~Kumar~Pati and S.~Braunstein.
\newblock Impossibility of deleting an unknown quantum state.
\newblock {\em Nature}, 404:164--165, 2000.

\bibitem[Kr{\'a}01]{10.1007/3-540-45627-9_26}
Richard Kr{\'a}l'ovi{\v{c}}.
\newblock Time and space complexity of reversible pebbling.
\newblock In Leszek Pacholski and Peter Ru{\v{z}}i{\v{c}}ka, editors, {\em
  SOFSEM 2001: Theory and Practice of Informatics}, pages 292--303, Berlin,
  Heidelberg, 2001. Springer Berlin Heidelberg.

\bibitem[KSS21]{kornerup2021spooky}
Niels Kornerup, Jonathan Sadun, and David Soloveichik.
\newblock The spooky pebble game, 2021.

\bibitem[LT82]{LT82}
Thomas Lengauer and Robert~E. Tarjan.
\newblock Asymptotically tight bounds on time-space trade-offs in a pebble
  game.
\newblock {\em J. ACM}, 29(4):1087–1130, October 1982.

\bibitem[LV96]{LiVit96}
Ming Li and Paul Vit\'anyi.
\newblock Reversibility and adiabatic computation: Trading time and space for
  energy.
\newblock {\em Proceedings of the Royal Society of London. Series A:
  Mathematical, Physical and Engineering Sciences}, 452(1947):769–789, Apr
  1996.

\bibitem[MSR{\etalchar{+}}19]{8715092}
Giulia Meuli, Mathias Soeken, Martin Roetteler, Nikolaj Bjorner, and
  Giovanni~De Micheli.
\newblock Reversible pebbling game for quantum memory management.
\newblock In {\em 2019 Design, Automation Test in Europe Conference Exhibition
  (DATE)}, pages 288--291, 2019.

\bibitem[NC02]{nielsen2002}
Michael~A Nielsen and Isaac Chuang.
\newblock Quantum computation and quantum information, 2002.

\bibitem[Pau75]{Pau75}
Wolfgang~J. Paul.
\newblock A 2.5 n-lower bound on the combinational complexity of boolean
  functions.
\newblock In {\em Proceedings of the Seventh Annual ACM Symposium on Theory of
  Computing}, STOC '75, page 27–36, New York, NY, USA, 1975. Association for
  Computing Machinery.

\bibitem[PH70]{HP70}
Michael~S. Paterson and Carl~E. Hewitt.
\newblock {\em Comparative Schematology}, page 119–127.
\newblock Association for Computing Machinery, New York, NY, USA, 1970.

\bibitem[PM99]{USENIX:ProMaz99}
Niels Provos and David Mazi\`{e}res.
\newblock A future-adaptive password scheme.
\newblock In {\em Proceedings of the Annual Conference on USENIX Annual
  Technical Conference}, ATEC '99, page~32, USA, 1999. USENIX Association.

\bibitem[PTC76]{PTC76}
Wolfgang~J. Paul, Robert~Endre Tarjan, and James~R. Celoni.
\newblock Space bounds for a game on graphs.
\newblock In {\em Proceedings of the Eighth Annual ACM Symposium on Theory of
  Computing}, STOC '76, page 149–160, New York, NY, USA, 1976. Association
  for Computing Machinery.

\bibitem[PV76]{PV76}
Nicholas Pippenger and Leslie~G. Valiant.
\newblock Shifting graphs and their applications.
\newblock {\em J. ACM}, 23(3):423–432, July 1976.

\bibitem[Tom81]{Tom81}
Martin Tompa.
\newblock Corrigendum: Time-space tradeoffs for computing functions, using
  connectivity properties of their circuits.
\newblock {\em J. Comput. Syst. Sci.}, 23(1):106, 1981.

\end{thebibliography}
